\documentclass[3p]{elsarticle}

\usepackage{amsfonts}
\usepackage{amsmath}
\usepackage{amssymb}
\usepackage{graphicx}
\usepackage[all]{xy}

\newcommand{\ac}{\cdot}
\newcommand{\Ad}{\mathop{\mathrm{Conj}}\nolimits}
\newcommand{\ad}{\mathop{\mathrm{Ad}}\nolimits}
\newcommand{\ads}{\mathop{\mathrm{ad}}\nolimits}
\newcommand{\Aset}{\mathcal{A}}
\newcommand{\Adm}{\mathcal{AV}}

\newcommand{\dd}{\mathrm{d}}

\newcommand{\dist}{\mathop{\mathrm{dist}}\nolimits}
\newcommand{\dlog}{\mathop{\mathrm{dlog}}\nolimits}

\newcommand{\EL}{\mathcal{EL}}

\newcommand{\expp}[1]{{\mathop{\mathrm{exp}}#1}}

\newcommand{\Fr}{\mathop{\mathrm{fr}}\nolimits}

\newcommand{\Hom}{\mathop{\mathrm{Hom}}\nolimits}
\newcommand{\Id}{\mathop{\mathrm{Id}}\nolimits}

\newcommand{\Int}{\mathop{\mathrm{int}}\nolimits}
\newcommand{\inv}{\mathop{\mathrm{inv}}\nolimits}
\newcommand{\Ext}{\mathop{\mathrm{ext}}\nolimits}

\newcommand{\La}{\mathcal{L}}

\newcommand{\LL}{\mathbb{L}}

\newcommand{\Map}{\mathop{\mathrm{Map}}\nolimits}

\newcommand{\odd}{\mathrm{odd}}
\newcommand{\oo}{\mathrm{o}}

\newcommand{\proj}{\mathop{\mathrm{proj}}\nolimits}

\newcommand{\RR}{\mathbb{R}}

\newcommand{\sgn}{\mathrm{sgn}}

\newcommand{\Trace}{\mathop{\mathrm{trz}}\nolimits}

\newcommand{\vol}{\mathrm{vol}}

\newcommand{\X}{\mathfrak{X}}
\newcommand{\ZZ}{\mathbb{Z}}

\newtheorem{theorem}{Theorem }
\numberwithin{theorem}{section}
\newdefinition{define}[theorem]{Definition }
\newproof{proof}{Proof }
\numberwithin{equation}{section}
\newtheorem{corollary}[theorem]{Corollary}
\newtheorem{lemma}[theorem]{Lemma}
\newtheorem{prop}[theorem]{Proposition}

\newdefinition{example}[theorem]{Example}
\newdefinition{remark}[theorem]{Remark }

\title{Symmetry-preserving discretization of variational field theories}

\author[UNL]{Ana~Cristina~Casimiro\corref{correspondiente} }
\ead{amc@unl.fct.pt}

\author[Lisboa]{C\'{e}sar~Rodrigo}
\ead{crodrigo@geomat-pt.com}

\cortext[correspondiente]{Corresponding Author. \\ Fax:(+351)212948391, Tlf.~no.:(+351)212948388(Ext.10825)}

\address[UNL]{CMA - Centro de Matem\'{a}tica e Aplica\c{c}\~{o}es, Portugal. \\
Departmento de Matem\'{a}tica,
Faculdade de Ci\^{e}ncias e Tecnologia, Universidade Nova de Lisboa,\\ Quinta da Torre 2829-516 Caparica, Portugal}

\address[Lisboa]{CMAF-CIO - Centro de Matem\'{a}tica, Aplica\c{c}\~{o}es Fundament\'{a}is e Investiga\c{c}\~ao Operacional \\ CINAMIL, Academia Militar\\ Av.~Conde Castro Guimar\~{a}es, 2720-113 Amadora, Portugal}

\begin{document}
\begin{abstract}
The present paper develops a variational theory of discrete fields defined on abstract cellular complexes. The discrete formulation is derived solely from a variational principle associated to a discrete Lagrangian density on a discrete bundle, and developed up to the notion of symmetries and conservation laws for solutions of the discrete field equations. The notion of variational integrator for a Cauchy problem associated to this variational principle is also studied. The theory is then  connected with the classical (smooth) formulation of variational field theories, describing a functorial method to derive a discrete Lagrangian density from a smooth Lagrangian density on a Riemannian fibered manifold, so that all symmetries of the Lagrangian turn into symmetries of the corresponding discrete Lagrangian. Elements of the discrete and smooth theories are compared and all sources of error between them are identified. Finally the whole theory is illustrated with the discretization of the classical variational formulation of the kinematics of a Cosserat rod.
\end{abstract}

\begin{keyword}
Discrete Variational Principles \sep Discrete field theory \sep Conserved currents \sep Variational Integrators  \sep Abstract cellular
complex \sep Covariant discretization \sep Simplicial geodesic interpolation.

\MSC[2010]
39A12 \sep 
49M25 \sep 
65D05 \sep 
65M08 \sep 
65M22 \sep 
74H15 \sep 
76M10 
\end{keyword}

\maketitle

\section{Introduction}
Most field equations arising in Physics can be derived from some variational principle. The group of symmetries of the action functional is a relevant aspect from a physical point of view, leading to new formulations. In the presence of symmetries one may derive equivalent equations on some reduced space, associated multisymplectic forms, or Poisson brackets that describe the field from a different perspective. In order to take advantage of these symmetries (reflecting the geometrical properties of the field), the most convenient formulation (see for example \cite{GarcGarcRodr06,GoldStern,Gotay98}) is to identify the field as a mapping $y(x)\colon X\rightarrow Y$, section of some fibered manifold $\pi\colon Y\rightarrow X$, and to characterize the field equations as a system of Euler-Lagrange equations, that is, some second order partial differential equations on the set $\Gamma(X,Y)$ of all possible sections, representing necessary conditions for the section to minimize the action with respect to some set of admissible variations. We refer to the previous references for technical details. In the same manner as for the references, for this work all objects are assumed to be infinitely differentiable.

The simplest case, mechanics, has as base space $X=\mathbb{R}$, the time line, and as bundle a product $Y=\mathbb{R}\times Q$ where $Q$ is a finite dimensional manifold, the configuration space. A given Lagrangian function $L\colon \mathbb{R}\times TQ\rightarrow \mathbb{R}$ determines an action functional $\mathbb{L}_{[0,s]}$ defined by:
	\begin{equation*}
	q\in \Map([0,s],Q)\mapsto \mathbb{L}_{[0,s]}(q)=\int_0^s L(t,q(t),\dot q(t))\,\dd t
	\end{equation*}
A variational principle seeks for trajectories $q(t)$ that are stationary for $\mathbb{L}_{[0,s]}$, with respect to certain admissible infinitesimal variations. Different choices of admissible variations lead to different equations that appear in classical mechanics, control theory, and constrained mechanics. For the choice of variations $\delta q(t)\in T_{q(t)}Q$ along the trajectory $q(t)$, whose support is contained in the interior $(0,s)$ of the integration domain $[0,s]$, criticality is characterized by Euler equations 
	\begin{equation}\label{Eulersimples}
	\frac{\dd}{\dd t}\left(\frac{\partial L}{\partial \dot q^k}(t,q(t),\dot q(t))\right)=\frac{\partial L}{\partial q^k}(t,q(t),\dot q(t)),\qquad \forall k
	\end{equation}

Moreover, any time-dependent vector field $D=\xi^k(t,q)\frac{\partial}{\partial q^k}\in \mathfrak{X}(Q)$ that is infinitesimal symmetry of the Lagrangian function (that is, $\xi^k\frac{\partial L}{\partial q^k}+(\frac{\partial\xi^k}{\partial t}+\dot q^s\frac{\partial\xi^k}{\partial q^s} )\frac{\partial L}{\partial \dot q^k}=0$) induces a corresponding conserved quantity  \cite{Noether18}, a function $\mu_D=\xi^k \frac{\partial L}{\partial \dot q^k}$ on $\mathbb{R}\times TQ$, such that its restriction to $(t,q(t),\dot q(t))$ for any trajectory satisfying Euler equations turns out to be constant:
	\begin{equation*}
	\mu_D(t,q(t),\dot q(t))-\mu_D(0,q(0),\dot q(0))=0,\qquad \forall t\in [0,s],\quad \text{ if } q\in \Map([0,s],Q)\text{ critical}
	\end{equation*}

The general case (field theories, in a proper sense) is determined by some fibered manifold $Y\rightarrow X$ over some $n$-dimensional manifold $X$, a fixed volume element $\vol^X\in \Omega^n(X)$, and a function $\La\colon J^1Y\rightarrow \mathbb{R}$ on the first jet bundle $J^1Y$ associated to the fibered manifold. Both of them determine $\La\cdot\vol^X$, the Lagrangian density, which leads to functionals (depending on the choice of domain of integration $\Aset\subseteq X$):
	\begin{equation*}
	\mathbb{L}_{\Aset}(y)=\int_{\Aset} (j^1y)^*\La\cdot\vol^X
	\end{equation*}
where $j^1y\colon X\rightarrow J^1Y$ is the 1-jet extension of $y(x)$. Two sections $y(x),\bar y(x)$ determine on a point $x\in X$ the same jet $j^1_x y=j^1_x\bar y$ if their values and directional derivatives coincide on $x$. If we consider fibered local coordinates $(x^i,y^k)$ on the bundle $Y$, there exist induced fibered local coordinates $(x^i,y^k,\partial_iy^k)$ on $J^1Y$, so that $(\partial_iy^k)(j^1_xy)=\partial y^k(x)/\partial x^i$.

When one considers as admissible variations those with compact support contained in the interior of $\Aset$, a necessary condition for a section $y(x)$ to be a minimum of the action functional is that at the interior points of $\Aset$ the section $y(x)$ satisfies a system of second order partial differential equations known as Euler-Lagrange equations. In a system of fibered local coordinates $(x^i,y^k)$, if $\vol^X=\dd x^1\wedge\ldots \wedge \dd x^n$ and $\La=\La(x^i,y^k,\partial_iy^k)$, this system of equations is:
	\begin{equation}\label{EL}
	\sum_i\frac{\dd}{\dd x^i}\left(\frac{\partial \La}{\partial (\partial_iy^k)}(x,y(x),(\partial y/\partial x)(x))\right)=\frac{\partial \La}{\partial y^k}(x,y(x),(\partial y/\partial x)(x)),\qquad \forall k
	\end{equation}
Again, for any vertical vector field $D=\xi^k(x,y)\cdot \frac{\partial}{\partial y^k}\in \mathfrak{X}(Y)$, such that
	\begin{equation*}
	\xi^k\frac{\partial\La}{\partial y^k}+\left(\frac{\partial \xi^k}{\partial x^i}+(\partial_i y^s)\frac{\partial \xi^k}{\partial y^s}\right)\frac{\partial \La}{\partial(\partial_i y^k)}=0
	\end{equation*}
(the 1-jet extension $j^1D\in\mathfrak{X}(J^1Y)$ is infinitessimal symmetry for the Lagrangian function $\La$), there exists a corresponding Noether current, a differential form $\mu_D=\xi^k\frac{\partial \La}{\partial (\partial_iy^k)} \cdot i_{\partial/\partial x^i}\vol^X$ on $J^1Y$ that, restricted to any solution of Euler-Lagrange equations turns out to be closed, so that its integration on any regular domain vanishes.
	\begin{equation}\label{NoetSmooth}
		\int_{\partial \Aset} \mu_D\circ (j^1y)=0,\qquad \forall \Aset\subset X\text{ if }y\in\Gamma(X,Y)\text{ critical}
	\end{equation}
These currents have the same role in field theories as conserved quantities do in mechanics. The reader is referred to \cite{GoldStern} for more details. For the more general theory of variational problems with constraints on fibered manifolds we refer to \cite{GarcGarcRodr06} and references therein

The existence of symmetries and conserved quantities becomes a central tool for solving Euler-Lagrange equations, reducing them to a smaller configuration space, and leading to equivalent formulations and additional structures.

Recent developments \cite{BouMarsden09,FerGarRod12,HaiLubWan06,LeokShin12,McLaQuis06,MoseVese91,WendMars97,West04} in numerical methods for Ordinary Differential Equations show that a strong tool to obtain integrators for equations (\ref{Eulersimples}) with good long-term properties is to exploit the variational origin of these equations and to develop a discrete analogue of this formalism. In this way one obtains the whole discrete variational theory, including Noether conservation laws. The discrete analogue of Euler equations determines integrators for these equations \cite{FernBlocOlve12,FerGarRod12,LeokShin12,McLaQuis06,MoseVese91,WendMars97,West04}, whose solutions can be compared to the solutions of the original smooth theory.

This kind of ideas has been applied in the case of Partial Differential Equations (\ref{EL}) derived from a variational principle in field theories \cite{Bahr,CasiRodr12,CasiRodr12b,Guo,LeonMarrDieg08,LewMarsOrtiWest03,MarsPatrShko98,Vank07}. The authors have given in \cite{CasiRodr12} a variational formalism for discrete fields defined on a discrete space that has the structure of cellular complex. That work led to applications \cite{CasiRodr12b} where integrators with energy-preserving properties are obtained.

There already exists a certain amount of works dealing with discretization of dynamics of different continuous materials, from a variational point of view \cite{Demoures15,LeonMarrDieg08,LewMarsOrtiWest03}. In most of these cases the continuous material in not discretized at all \cite{LeonMarrDieg08} or is discretized in finitely many elements $a\in A$, which ensures the existence of a large-dimensional configuration space, and the evolution on this configuration space is then discretized on the time variable. The introduction of discrete Lagrangian densities in these theories is performed using known techniques from discrete mechanics (that discretize certain ODEs with several unknowns), adapted to this situation. In these formalisms, the discrete Lagrangians that arise are expressed in some form:
	\begin{equation*}
	L(q_a(t_{k-1}),q_a(t_k)) \qquad\text{ several values for index }a\in A\text{; fixed value }k\in\mathbb{Z}
	\end{equation*}
which is a function defined in a large manifold $(\prod_{a\in A}Q)\times (\prod_{a\in A}Q)$. Dealing in appropriate manner with this mechanical problem leads to Euler equations and conservation laws in the sense of mechanics that, for this particular situation, can be expressed locally in terms of the discretized material elements.

In other works there exist a discrete Lagrangian defined on a low-dimensional space \cite{Guo,MarsPatrShko98,Vank07}, for a particular discrete model of the plane. In particular, the widely used lattice field theories \cite{Schwalm,Wise06}, lattice gravity \cite{Hamber97} and Regge calculus \cite{Bahr} display discrete analogues of gauge theories using lattices to describe the physical gauge field. Finite element theory is another domain where discrete fields appear as elements of a finite-dimensional space, that of finite elements \cite{ArnFalWin06,ChriMuntOwre11,Deylen15}. However, this approach focuses mainly on error bounds, leaving a too narrow margin for the introduction of geometrical tools that may appear in the presence of symmetries. In these theories one may find so-called approximate conservation laws, which represent a discrete version of the smooth conserved quantities given by Noether's theorem, together with error bounds for these objects with respect to Noether currents appearing in the smooth formulation.

Therefore, both numerical and geometrical concerns appear when we aim to discretize some field theory arising from a variational principle. This task demands the consideration of two usually incompatible aspects: the study of symmetries, and the study of errors. For the first one, relevant aspects are geometrical transformations that respect all objects with geometrical and physical interest in the theory. For the second one, the relevant aspect is the behavior of error and how it can be bounded, an error that is usually measured with respect to a norm without geometrical interpretation, whose interest lies on floating point arithmetic properties.

In the present work we provide a general geometrical variational formulation of discrete field theory on abstract cellular complexes with an extensive exploration of its analogies with smooth variational theories. Though it is not always explicitly stated, the model of discrete space most suited for geometrical considerations seems to be that of abstract cellular complex (see the case of triangulations or quad-graphs in \cite{AdlBobSur04,ArnFalWin06,CasiRodr12,DesbKansTong06,Deylen15,Schwalm,Sengupta04}). See \cite{ArnFalWin06,ChriMuntOwre11} for approaches to triangulations, cellular complexes, its refinements, and how to associate finite element systems on such a structure. 
In this work we shall present a variational problem on a discrete space with a structure of abstract cellular complex.

\section{Abstract Cellular Complexes}

We shall follow \cite{CasiRodr12} (see also \cite{DesbKansTong06} for the case of simplicial cells):

\begin{define}\label{defineACC}
We call $n$-dimensional abstract cellular complex a set $V$ (the cells) together with two mappings $\dim\colon V\rightarrow \{0,1,\ldots,n\}$ and $[\,\colon]\colon V\times V\rightarrow \{-1,0,+1\}$ (the dimension and incidence mappings, respectively) such that:
    \begin{enumerate}
    \item\label{cond1} There exists $\beta\in X$ with $\dim\beta=n$.
    \item\label{cond2} If $\alpha,\beta\in V$ satisfy $[\beta\colon\alpha]\neq 0$ then $\dim\alpha=\dim\beta-1$.
    \item\label{cond3} For each cell $\alpha\in V$ there exists only a finite number of cells $\gamma\in V$ such that $[\alpha\colon\gamma]\neq 0$, and a finite number of cells $\beta\in V$ such that $[\beta\colon\alpha]\neq 0$.
    \item\label{cond4} For each pair of cells $\beta,\gamma\in V$ there holds
        \begin{equation}\label{complex}
        \sum_{\alpha\in V} [\beta\colon \alpha]\cdot [\alpha\colon \gamma]=0
				\end{equation}
    \end{enumerate}
\end{define}
As we shall see, the incidence mapping determines the boundary and coboundary operators. Condition \ref{cond1} defines the dimension of a complex as the highest dimension of its cells. Condition \ref{cond2} establishes that boundaries (defined below in (\ref{discstokes})) have codimension 1. Condition \ref{cond3} is a finiteness hypothesis in order to avoid infinite sums, and condition \ref{cond4} states that the boundary of a boundary vanishes.

The dimension mapping allows to define $V_k=\{\beta\in V\colon \dim\beta=k\}$ so that $V=\bigsqcup_{k=0}^n V_k$ ($V_0$ is the set of $0$-cells
or vertices, $V_1$ is the set of oriented $1$-cells or edges, and $V_k$ is the set of oriented $k$-cells). We say a cell $\alpha\in V$ is incident to a cell $\beta\in V$ with compatible orientation if $[\beta\colon\alpha]=1$,
incident with opposite orientations if $[\beta\colon\alpha]=-1$ and non-incident if $[\beta\colon \alpha]=0$.

The incidence mapping introduces a topology on $V$, where a cell $\alpha\in V_{k-l}$ is said to be adherent to other cell $\beta\in V_k$ (and we write $\alpha\prec\beta$) if $\alpha=\beta$ or if there exists a sequence of cells $\alpha=\alpha_{k-l},\alpha_{k-l+1},\alpha_{k-l+2},\ldots,
\alpha_{k}=\beta$ each incident to the next one: $0\neq[\alpha_{k-i+1}\colon \alpha_{k-i}]$, for $i=1,\ldots,l$. Moreover, we may talk of  discrete oriented domains of integration $c_k\in C_k(V,\ZZ)$ (or $k$-chains) and discrete $k$-forms $\omega_k\in \Omega^k(V)$ (or $k$-cochains)
    \begin{equation*}
    C_k(V,\ZZ)=\bigoplus_{\alpha\in V_k}\ZZ,\quad \text{ free abelian group generated by }V_k
    \end{equation*}
    \begin{equation*}
    \Omega^k(V)=\Map(V_k,\RR)=\prod\limits_{\alpha\in V_k}\RR,\quad \text{ real functions defined on }V_k
    \end{equation*}
 between these spaces there exists a natural bilinear product (discrete integration operator) and the incidence mapping generates two linear operators: the differential $\dd_k$ of discrete $k$-forms and the boundary $\partial_{k+1}$ of discrete $(k+1)$-chains (domains of integration):
    \begin{equation*}
    \langle c,\omega\rangle=\sum_{\alpha\in V_k} c(\alpha)\cdot \omega(\alpha),\qquad c\in C_k(V,\ZZ),\, \omega\in \Omega^k(V)
    \end{equation*}
    \begin{equation*}
    (\dd_k \omega)(\beta)=\sum_{\alpha\prec\beta} [\beta\colon\alpha]\cdot\omega(\alpha),\quad \omega\in\Omega^k(V),\,\beta\in V_{k+1}
    \end{equation*}
    \begin{equation*}
    (\partial_{k+1} c)(\alpha)= \sum_{\alpha\prec\beta} [\beta\colon\alpha]\cdot c(\beta),\quad c\in C_{k+1}(V,\ZZ),\,\alpha\in V_k
    \end{equation*}
Here $\dd_k$ is well defined because of first part in condition \ref{cond3} in definition \ref{defineACC}. 
As one can readily see, these objects satisfy a discrete analogue of Stokes' formula:
    \begin{equation}\label{discstokes}
    \langle c,\dd \omega\rangle=\langle \partial c,\omega\rangle, \quad c\in C_{k+1}(V,\ZZ), \, \omega\in \Omega^k(V)
    \end{equation}
(when there is no possibility of misunderstanding, we will suppress the index $k$ for the boundary $\partial_k$ and the differential $d_k$ operators)

Any $k$-cell $\beta\in V_k$ can be seen as a $k$-chain $c_\beta$ (taking value $1$ on $\beta$ and $0$ on any other cell). In fact, $k$-cells generate the free $\mathbb{Z}$-module $C_k(V,\mathbb{Z})$. The boundary of $\beta$ (that is, of $c_\beta$) may be considered as the set of its incident cells $\beta\in V_{k-1}$ each with
positive or negative weight depending on the compatibility of its orientation and that of $\gamma$. Condition (\ref{complex}) imposed for the incidence
mapping is equivalent to $\partial_{k}\circ\partial_{k+1}(c_\beta)=0$, so $\partial_k\circ\partial_{k+1}=0$ on any chain.

Following $\partial_{k}\circ\partial_{k+1}=0$ and discrete Stokes' formula (\ref{discstokes}), we have $\dd_{k+1}\circ\dd_k=0$.  Therefore we have a chain complex $\partial_k\colon C_k(V,\mathbb{Z})\rightarrow C_{k-1}(V,\mathbb{Z})$ and a co-chain complex $\dd_k\colon \Omega^k(V)\rightarrow \Omega^{k+1}(V)$ that
define homology and co-homology groups $\ker\partial_k/\mathop{\mathrm{Im}}\partial_{k+1}$, $\ker\dd_{k+1}/\mathop{\mathrm{Im}}\dd_{k}$.

This is the minimal machinery required to introduce the notion of Lagrangian density and the variational problem
associated to it (see \cite{CasiRodr12}). In these problems we shall restrict ourselves to integration domains given as a characteristic chain of some finite subset of $n$-cells:
\begin{define}
Any finite subset $\Aset\subset V_n$ of $n$-cells defines a corresponding characteristic chain $c_{\Aset}\in C_n(V,\mathbb{Z})$, whose value on some $n$-cell $\beta\in V_n$ will be $1$ if $\beta\in \Aset$ or else $0$. We shall refer to $\langle c_{\Aset},\omega\rangle=\sum_{\beta\in\Aset} \omega(\beta)$ as the integration over $\Aset$ of some $n$-cochain $\omega\in \Omega^n(V)$.
\end{define}
A particularly interesting case of set and characteristic chain is the star or sphere associated to a vertex:
\begin{define}
We shall call spherical chain associated to some vertex $v\in V_0$ the chain $sc_v\in C_n(V,\mathbb{Z})$ whose value at some $n$-cell $\beta\in V_n$ is $1$ if $v$ is adherent to $\beta$, and $0$ else. This spherical chain is the characteristic chain of a finite subset $S_v\subset V_n$ (finite because of second part of condition \ref{cond3} in definition \ref{defineACC}), called the sphere with center $v$ (in the literature this set $S_v$ is also called the star associated to $v$).
\end{define}
This notion of sphere allows to distinguish interior, exterior and frontier vertices of any set $\Aset\subset V_n$:
\begin{define}
We say a vertex $v\in V_0$ is interior to $\Aset\subseteq V_n$ if $S_v\subset \Aset$. We say a vertex $v\in V_0$ is exterior to $\Aset$ if $S_v\subset V_n\setminus \Aset$. We say a vertex is frontier to $\Aset$ if it is not interior nor exterior to $\Aset$. The sets of interior, exterior and frontier vertices of $\Aset$ shall be denoted by $\Int \Aset$, $\Ext \Aset$, $\Fr \Aset$, respectively.
\end{define}

Let us now focus our attention on the study of certain particularly simple cellular complexes: 

\subsection*{Particular case: Simplicial Complexes}

Let $X$ be a set, whose elements we call vertices. For any $k\in\mathbb{N}=\{0,1,\ldots\}$ we call abstract $k$-simplex on $X$ any (non-ordered) subset $\alpha\subset X$ of $k+1$ vertices ($\sharp \alpha=k+1$). By ``abstract'' simplices we mean that we don't want to consider $X$ to be any affine space nor a simplex to be the convex hull of its vertices. However, this generally employed affine model of a simplex is perfectly compatible with our presentation of abstract simplicial complexes and might be helpful for the visualization of the different notions. With this affine model in mind, on affine spaces a 0-simplex would be a point, a 1-simplex a segment, a 2-simplex a (possibly degenerate) triangle, a 3-simplex a tetrahedron, and $n$-simplices, the higher dimensional generalization of these objects.

\begin{define}\label{def24}
We call abstract simplicial complex on a set $X$ any abstract cellular complex $(V,\dim,[\,\colon])$, in the sense of definition \ref{defineACC}, with the following particular characteristics:
    \begin{itemize}
    \item Elements $\alpha\in V_k$  are $k$-dimensional abstract simplices on $X$, i.e., subsets $\alpha\subseteq X$ containing exactly $k+1$ elements of $X$.
    \item Given any $\beta\in V$, its adherent cells are precisely all the nonempty subsets of $\beta$:
				\begin{equation}\label{fechado}
				\emptyset\neq \alpha\subseteq \beta\in V \Rightarrow  \alpha\in V
				\end{equation}
        \begin{equation}\label{incsim}
        [\beta\colon\alpha]\neq 0\Leftrightarrow \alpha=\beta\setminus\{v\}\neq \emptyset, \text{ for some }v\in \beta
        \end{equation}
	\item For any 1-cell $\alpha=\{v,w\}\in V_1$, its adherent vertices can be distinguished by the incidence mapping:
			\begin{equation}\label{initialfinalvertex}
			[\alpha\colon v]\neq [\alpha\colon w],\qquad \forall \alpha=\{v,w\}\in V_1
			\end{equation}
    \end{itemize}
\end{define}

These conditions indicate that any cell of an abstract simplicial complex can be seen as an abstract simplex, in all aspects regarding the dimension notion.
Following (\ref{incsim}) we conclude that $\alpha\prec\beta\Leftrightarrow \alpha\subseteq\beta$, for any cells $\alpha,\beta\in V$. Condition (\ref{initialfinalvertex}) shows that any edge has two adherent vertices, one of them with $[\alpha,w]=1$ and the other with $[\alpha,v]=-1$.

\begin{define}\label{inicialfinal}
For any simplicial edge $\alpha\in V_1$ with adherent vertices $v,w$, if $[\alpha\colon w]=1$ and $[\alpha\colon v]=-1$ we shall call $v$ the initial vertex of the edge, and $w$ the final vertex of the edge.
\end{define}
\begin{define}
We call ordering of a $k$-simplex $\alpha\subset X$ any bijective mapping $v \colon i\in \{0,\ldots,k\}\mapsto v_i\in \alpha$. We denote the set of orderings of $\alpha$ by $\mathrm{Ord}(\alpha)$.
\end{define}
Giving a $k$-simplex is giving $k+1$ vertices, but not in a particular sequence (we don't fix the ordering). However we shall see that the cellular complex structure determines an orientation for each simplex.
\begin{define}
We call orientation on a $k$-simplex $\alpha$ any mapping $\oo_\alpha\colon\mathrm{Ord}(\alpha)\rightarrow \{\pm1\}$ with the following property
    \begin{equation*}
    \oo_\alpha(v\circ\varphi)=\oo_\alpha(v)\cdot \sgn(\varphi),\quad \forall v\in \mathrm{Ord}(\alpha),\,\forall \varphi\in\mathrm{Bij}(\{0,1,\ldots, k\})
    \end{equation*}
\end{define}
Giving an orientation $\oo_\alpha$ allows to split the set of orderings $v\in\mathrm{Ord}(\alpha)$ into two classes: those for which $\oo_\alpha(v)=+1$, the ordering is {\em compatible} with the orientation, and those that are not. More precisely (we leave the proof to the reader), we may derive a particular orientation $\oo_\alpha$ for every cell, from the incidence mapping $[\cdot,\cdot]$, and conversely:
%
		\begin{equation*}
		\alpha=\{v\}\in V_0\Rightarrow o_\alpha(v)=1
		\end{equation*}
    \begin{equation}\label{orderincid}
     \alpha=\in V_k,\, k\geq 1,\, v=(v_0,\ldots, v_k)\in \mathrm{Ord}(\alpha)\Rightarrow \oo_\alpha(v)=\prod\limits_{j=1}^k [\{v_{j-1},\ldots, v_k\}\colon\{v_j,\ldots,v_{k}\}]
    \end{equation}
		\begin{equation}\label{incidb}
	\beta=\{v_0,\ldots, v_k\}, \alpha= \beta\setminus\{v_j\}\Rightarrow [\beta\colon\alpha]=(-1)^j\cdot \oo_\beta(v_0,\ldots, v_k)\cdot \oo_\alpha(v_0,\ldots, v_{j-1},v_{j+1},\ldots, v_k)
	\end{equation}
Formula (\ref{incidb}) represents the classical notion of boundary operator on simplicial complexes, assuming that we choose for each simplex contained in $\beta$ the  orientation determined by the given ordering $(v_0,\ldots,v_k)$. It is common in some mathematical areas and in the literature (see \cite{DesbKansTong06} for example) to fix some ordering or orientation on each simplex and to derive a notion of boundary operator taking these orientations as given. In this article we shall avoid talking about orderings or orientations of any simplex, and will consider just the topological object of our interest: the incidence mapping or, equivalently, the boundary operator on the simplicial complex. The equivalence of both approaches is derived from formulas (\ref{orderincid}),(\ref{incidb}).

Giving some $k$-simplex is giving $k+1$ vertices, but not in a particular sequence (we don't fix the ordering). The incidence morphism determines an orientation of each cell (that is, determines some preferred ordering on each cell, up to positive permutations of the vertices).

The incidence morphism generates a family of orientations $(\oo_\alpha)_{\alpha\in V}$, one for each cell. For any cell of an abstract simplicial complex (in the sense of Definition \ref{def24}) there exists a preferred ordering of its adherent vertices, up to a positive permutation. If $\alpha$ is a simplicial edge and $(v,w)$ is a positive ordering of its vertices, then the incidence mapping defined by (\ref{orderincid}) is $[\alpha\colon w]=1$, $[\alpha\colon v]=-1$, hence $v$ is the initial vertex and $w$ the final vertex, in the sense given in definition \ref{inicialfinal}.

The following two particular models of cellular complex shall be used in the applications:
\subsection*{1. Cubic cellular complex on $\RR^n$}
In \cite{CasiRodr12}, a discrete model of Euclidean space was introduced using a Cartesian lattice of points, segments, squares, cubes, and so on, on $\RR^n$, also used by many other works on discrete field theories \cite{AdlBobSur04,AdlVes04,Demoures15,FadVol94,Guo,MarsPatrShko98,Vank07}. This complex arises when we consider on $\mathbb{R}^n$ the hyperplanes $x_i=a_i\in\mathbb{Z}$, and all possible connected domains bounded by these hyperplanes. A convenient way to describe the resulting cells is to give its barycenter, which is a vector formed by integer, and half-integer components. Our presentation here differs only in notation from \cite{CasiRodr12}, where coordinates were doubled to avoid working with half-integers. We consider $\alpha\in (\frac12\ZZ)^n$ to  represent the center of the following convex set:
	\begin{equation*}
	\alpha\in(\frac12\ZZ)^n\Rightarrow K_\alpha=\left\{ (x_1,\ldots,x_n)\in\RR^n\,\colon \begin{aligned} x_i&=\alpha_i \text{ if }\alpha_i\in \mathbb{Z}\\ |x_i-\alpha_i|&<1/2\text{ if }\alpha_i\in \mathbb{Z}+\frac12 \end{aligned}\right\}\subset\RR^n
	\end{equation*}
If $\alpha$ has $n-k$ integer components, then $K_\alpha$ is contained in exactly $n-k$ hyperplanes of the form $x_i=\alpha_i\in\mathbb{Z}$, and $K_\alpha$ is equivalent (as affine semi-space) to the $k$-dimensional open cube $]0,1[^k$. We shall say $\dim\alpha=k$ or $\alpha\in V_k$. The space $\mathbb{R}^{n}$ can be seen as a disjoint union of these cells $K_\alpha$. In particular, any $\alpha\in V_0=\mathbb{Z}^n$ is associated to a ``node'' $K_\alpha=\{x\}\subset\mathbb{R}^n$ given by $x=\alpha\in\ZZ^n$; any segment $\{x+s\cdot e\}_{0<s<1}$ joining two nodes $x\in\mathbb{Z}^n$ and $x+e$ (with fixed $x,e\in\mathbb{Z}^n$, being $e$ unit vector) is the convex set $K_\alpha$ represented by $\alpha=x+\frac12 e\in(\frac12\mathbb{Z})^n$; any square $\{x+s_1\cdot e_1+s_2\cdot e_2\}_{0<s_1,s_2<1}$ (with fixed $x,e_1,e_2\in\mathbb{Z}^n$, and $e_1,e_2$ linearly independent unit vectors) is the convex set $K_\alpha$ represented by $\alpha=x+\frac12 e_1+\frac12 e_2\in(\frac12\mathbb{Z})^n$, and so on.

If $\alpha\in(\frac12\mathbb{Z})^n$ has a set of half-integer coordinates $\alpha_i\notin \mathbb{Z}$ at positions $\alpha^{half}=(i_1<i_2<\ldots<i_k)$, the set $K_\alpha$ is an open convex set contained in the supporting $k$-dimensional affine subspace $H_\alpha\subset\RR^n$ defined by equations $x_j=\alpha_j$ $\forall j\notin \alpha^{half}$. On the affine space $H_\alpha$  we choose, as convention, the orientation $\vol_\alpha=\dd x^{i_1}\wedge\ldots\wedge \dd x^{i_k}$. The incidence mapping between cells may be defined using this convention for orienting $K_\alpha$ and using the outward pointing vector associated to a pair of adherent cells, leading to the following cellular complex:
\begin{define}[See \cite{CasiRodr12}]
We call cubic cellular complex on $\mathbb{R}^n$ (or $n$-D cubic cellular complex) the abstract cellular complex whose cells are a set $V=(\frac12\mathbb{Z})^n$. The dimension mapping is $\dim \alpha=\sharp\alpha^{half}$, where $\alpha^{half}$ is the set of indices $i\in\{1,\ldots, n\}$ with half-integer coordinate. The incidence mapping $[\beta\colon\alpha]\in \{0,\pm 1\}$ is given by:
    \begin{equation*}
		\begin{aligned}
    &\alpha\in V_k,\, \beta\in V_{k+1},\, \|\alpha-\beta\|&=1 \Rightarrow
    [\beta\colon\alpha]\cdot \vol_\alpha\sim
     i_{(\alpha-\beta)}\vol_\beta \\
    &\text{ else } [\beta\colon\alpha]=0
    \end{aligned}
    \end{equation*}
where the equivalence $\sim$ holds if both affine forms determine the same orientation on $H_\alpha$, where $\vol_\alpha$, $\vol_\beta$ represent the affine conventional orientation forms, as describe above.
\end{define}

For the $(k+1)$-cell $\beta$, its boundary may be computed as:
    \begin{equation*}
    \partial c_\beta=\sum_{j=1}^{k+1} \sum_{s=\pm 1} (-1)^{j-1}\cdot s\cdot c_{\beta+s\cdot e_{i_j}},\qquad
    \beta^\odd=(i_1<i_2<\ldots<i_{k+1})
    \end{equation*}
where $e_1,\ldots,e_n$ is the canonical basis on $\ZZ^n$

Furthermore, for any subset $S\subseteq[n]=\{1,2,\ldots,n\}$, consider the vector:
	\begin{equation}\label{informula}
	e_S\in\ZZ^n,\quad \text{ with }x_i(e_S)=\left\{\begin{aligned} &1\text{ if }i\in S \\ &0 \text{ if }i\notin S\end{aligned} \right.
	\end{equation}
The set $\{e_S\}_{S\subseteq [n]}$ coincides with $\{0,1\}^n\subseteq\RR^n$. For the cubic cellular complex, the sphere $S_v$ centered at $v\in V_0=\mathbb{Z}^n$ is the set of n-cells $\beta=v+\frac12e_S-\frac12 e_{\bar{S}}\in V_n$ (where $S\sqcup \bar{S}=[n]$), there exist $2^n$ different $n$-cells on this sphere, whose adherent vertices have the form $v+e$ (where $e\in \{-1,0,1\}^n\subset \ZZ^n$). For each $v\in \mathbb{Z}^n=V_0$, the sphere $S_v$ has $3^n$ different adherent vertices. Vertices can be seen as points with integer coordinates on $\mathbb{R}^n$ and the corresponding spheres as hypercubes with diameter 2.

\begin{remark}
In particular, for any $\alpha\in V_k$, the set $K_\alpha$ is the interior on $H_\alpha\subseteq \RR^n$ of the convex hull of all nodes $x$ associated to vertices $v\prec\alpha$. We shall call $K_\alpha$ the open convex hull associated to $\alpha\in V_k$. Points on $\mathbb{R}^n$ that don't lay on the open convex hull associated to any $n$-cell are points whose coordinates have some integer entry. The disjoint union of the open convex hulls of all $n$-cells may be seen as $\mathbb{R}^n\setminus H$, where $H$ is the closed set obtained by the union of all supporting hyperplanes $x_i=m\in\mathbb{Z}$:
	\begin{equation*}
	\bigcup\limits_{\beta\in V_n} K_\beta=\mathbb{R}^n\setminus H,\qquad H=\bigcup\limits_{i\in[n],m\in\mathbb{Z}} \{x_i=m\}
	\end{equation*}
The $n$-D cubic cellular complex is obtained partitioning $\mathbb{R}^n$ using these hyperplanes.
\end{remark}

\subsection*{2. Coxeter-Freudenthal-Kuhn simplicial complex on $\RR^n$}

Simplicial decomposition of space is also used in many works \cite{ArnFalWin06,DesbKansTong06,Deylen15,Hamber97} to deal with discrete field theories. For its combinatorial simplicity and relation to the cubic complex we introduce here a particular partition of the hypercube into simplices, employed in different areas of topology, computing sciences, and numerical methods, known as Freudenthal's triangulation, Kuhn's partition, or Coxeter-Freudenthal-Kuhn (CFK) triangulation. We focus on the simplicial cellular complex structure generated by this mechanism, denoting it as CFK simplicial complex.

Consider $X=\mathbb{Z}^n$ (the set of vertices). For any $v\in\mathbb{Z}^n$, call weight of this vertex the value $\omega(v)=x_1(v)+\ldots+x_n(v)$. For arbitrary $S\subseteq[n]$ the weight function applied to $e_S$ (defined in (\ref{informula})) is $\omega(e_S)=\sharp S$, the cardinal of the set $S$.

\begin{define}\label{asc}
We call Coxeter-Freudenthal-Kuhn simplicial complex on $X=\ZZ^n$ (or $n$-D CFK simplicial complex) the following:
	\begin{equation}\label{CFK}
	V=\left\{\alpha\subset X\,\colon\, \alpha\neq \emptyset, \,(v-w)\text{ or }(w-v)\in\{0,1\}^n,\,\forall v,w\in\alpha\right\}
	\end{equation}
The dimension of $\alpha\in V$ is given as $\dim \alpha=\sharp\alpha-1$, and the incidence mapping will be determined in definition \ref{incidenceCFK} using certain conventional orientation on affine subspaces.
\end{define}
Before giving the convention to fix the incidence mapping, we study the set $V$ of cells and model each of them as a convex subset on $\RR^n$.

The following lemma shows that any element $\alpha\in V$ is determined as a nonempty collection of at most $n+1$ vertices $v\in X=\ZZ^n$ hence, following definition \ref{def24}, $V$ is an $n$-dimensional simplicial complex modelled on $X=V_0=\mathbb{Z}^n$.
 
\begin{lemma}
Let $\alpha\in V$ be a simplex with dimension $k$ from the CFK complex given definition \ref{asc}. Let $v_0,v_1,\ldots, v_k$ be its vertices, ordered by increasing weights. Then there exists a unique ordered sequence of $k$ disjoint nonempty sets $S_a\subseteq \{1,2,\ldots, n\}$ (where $a=1,\ldots,k$) and corresponding vectors $e_{S_1},\ldots, e_{S_k}$ defined by (\ref{informula}) such that:
	\begin{equation*}
	v_{a}=v_{a-1}+e_{S_{a}},\qquad \forall a=1,\ldots, k
	\end{equation*}
hence $v_i=v_0+e_{S_1}+\ldots +e_{S_i}$
\end{lemma}
\begin{proof}
From (\ref{CFK}) we see that for any pair of vertices $v\neq w\prec\alpha$ there holds $v-w=\pm e_S$, for some nonempty subset $S\subseteq [n]$. As the weight function is linear and $\omega(e_S)> 0$ for any nonempty subset $S\subseteq [n]$, and we chose $\omega(v_a)\geq \omega(v_{a-1})$,  we may conclude that for each $a=1,\ldots, k$ there exists a unique nonempty set $S_a$ such that $v_a=v_{a-1}+e_{S_a}$. Adding all these equalities we get $v_k=v_0+e_{S_1}+\ldots+e_{S_k}$. Also $v_k$ has higher weight than $v_0$ and both of them belong to a common simplex $\alpha$. Hence $v_k-v_0=e_{\bar{S}}\in\{0,1\}^n$, for some set $\bar{S}$. We conclude then that $e_{\bar{S}}=e_{S_1}+\ldots+e_{S_k}$ and consequently no index $i\in\{1,\ldots, n\}$ is repeated in the sets $S_1,\ldots, S_k$. These are disjoint sets.\qed
\end{proof}
\begin{corollary}
Consider the set of pairs $(v,\mathcal{S})$ formed by an element $v\in\ZZ^n$ and an ordered sequence $\mathcal{S}=(S_1,\ldots, S_k)$ of disjoint nonempty subsets  of $[n]=\{1,\ldots,n\}$ (where we admit the case $k=0$, and the empty sequence).

The mapping that takes any such pair $(v,\mathcal{S})$ to the abstract simplex $\alpha(v,\mathcal{S})=\{v_0,\ldots, v_k\}$ defined by:
	\begin{equation}\label{nomesimpl}
	(v,\mathcal{S})\mapsto \alpha(v,\mathcal{S})=\{v_0,\ldots, v_k\},\qquad v_0=v,\, v_i=v+e_{S_1}+\ldots+e_{S_i},\quad \forall i=1,\ldots, k
	\end{equation}
gives a one-to-one mapping between this set of pairs and the $n$-D CFK simplicial complex (\ref{CFK}).
\end{corollary}
As a particular case the CFK simplicial complex on the plane $\RR^2$ is given (using the notation in (\ref{nomesimpl})) by vertices $v_{ij}$ determined as $\alpha((i,j),\emptyset)$, edges determined as $\alpha((i,j),\{1\})$, $\alpha((i,j),\{2\})$, $\alpha((i,j),\{1,2\})$, and faces determined as $\alpha((i,j),(\{1\},\{2\}))$, $\alpha((i,j),(\{2\},\{1\}))$.

The geometrical meaning of the $n$-D CFK simplicial complex turns clearer if we use the following identification of its abstract cells as convex sets in $\RR^n$.
\begin{define}
Consider $V$ the $n$-D CFK simplicial complex modelled on $\mathbb{Z}^n$, from definition \ref{asc}. For $k\geq 1$ we call open convex hull of any abstract cell $\alpha=\{v_0,v_1,\ldots, v_k\}\in V_k$ the set:
	\begin{equation*}
	K^o_\alpha=\{\lambda_0v_0+\lambda_1v_1+\ldots+\lambda_kv_k\,\colon \lambda_0,\ldots,\lambda_k>0,\, \sum \lambda_i=1\}\subset\RR^n
	\end{equation*}
For completion, in the case $\alpha=v\in V_0$ we call $K^o_\alpha=\{v\}\subset\RR^n$.
\end{define}
We use the term ``open'' to stress that $\lambda_i=0,1$ are not allowed. Hence $K^o_\alpha$ will be an open segment, an open triangle, or higher dimensional analogues of these objects.  Generally $K^o_\alpha$ is not really open in $\mathbb{R}^n$, but lies on the affine subspace determined by the vertices $v\in\alpha$, and $K^o_\alpha$ is open on this subspace. 

\begin{prop}
Consider the integer and fractional components $c(p)\in\ZZ^n$, $f(p)\in[0,1[^n$ of any point $p\in\RR^n$, defined by $x_i(p)=c_i(p)+f_i(p)$ with $c_i(p)\in\ZZ$, $0\leq f_i(p)<1$. Consider any cell $\alpha\in V_k$ determined by $(v,\mathcal{S}=(S_1,\ldots, S_k))$, on the $n$-D CFK simplicial complex.

The point $p$ lies on the open convex hull $K^o_\alpha$ if and only if the integer component $c(p)$ coincides with $v$, and the fractional component $f(p)$ takes null value for indices not belonging to $S_1\cup\ldots\cup S_k$, common nonvanishing values for indices belonging to a common set $S_a$, and decreasing nonvanishing values, as we advance in the ordered sequence os sets $\mathcal{S}$:
	\begin{equation*}
	p\in K_\alpha\Leftrightarrow \left\{
	\begin{aligned}
	&c(p)=v \\
	&i\notin S_1\cup\ldots\cup S_k\Rightarrow f_i(p)=0\\
	&j,i\in S_a\Rightarrow f_j(p)=f_i(p)>0\\
	&j\in S_{a+1},\, i\in S_a\Rightarrow f_i(p)>f_j(p)>0
	\end{aligned}\right. \quad \left(\begin{aligned} &p\in\RR^n \\ &\alpha=\alpha(v,(S_1,\ldots, S_k))\in V_k \\ &a\in\{1,\ldots, k\}\end{aligned}\right)
	\end{equation*}
\end{prop}
\begin{proof}
For the characterization of points $p\in K^o_\alpha$, observe that the notion of open convex hull of $k+1$ vertices covariates with respect to permutations in the coordinates or with respect to translations, and that translating a point with a vector $c\in\ZZ^n$ preserves its fractional component, adding to the integer component the same vector $c$. With an appropriate translation and permutation of the coordinates we may assume that:
	\begin{equation*}
	v=(0,\ldots,0),\, S_1=\{1,\ldots,i_1\},\, S_2=\{1+i_1,\ldots,i_2\},\ldots S_k=\{1+i_{k-1},\ldots,i_k\} 
	\end{equation*}
for some $1\leq i_1<i_2<\ldots <i_k\leq n$. Hence:
	\begin{equation*}
	v_a=v+e_{S_1}+\ldots+e_{S_a}=(1,\ldots, \overbrace{1}^{i_a},0,\ldots, 0)
	\end{equation*}
where the sequence of 1's ends at position $i_a$. If we use the definition of the open convex hull, and use the notation $s_i=\lambda_i+\ldots+\lambda_k$, the convex hull of these points $v_0,\ldots, v_k$ is the family of points with the form:
	\begin{equation*}
	K^o_\alpha=\{(s_1,\ldots,\overbrace{s_1}^{i_1}, s_2,\ldots,\overbrace{s_2}^{i_2},\ldots, \ldots, s_k, \ldots, \overbrace{s_k}^{i_k},0,\ldots,0)\, \colon\,  1>s_1>s_2>\ldots>s_k>0\}
	\end{equation*}
where $1>s_1$ because $\lambda_0+s_1=1$ and $\lambda_0>0$. There holds $s_a>s_{a+1}$ because $\lambda_a+s_{a+1}=s_a$, with $\lambda_a>0$. Finally $s_k=\lambda_k>0$.

Therefore, points in this convex hull are those elements of $\RR^n$ whose coordinates $(x_1,\ldots, x_n)$ have integer component $(0,\ldots, 0)$ and fractional component $(s_1,\ldots, s_k)$ satisfying, for each $a\in\{1,\ldots,k\}$ and $i,j\in\{1,\ldots, n\}$: 
	\begin{equation*}
	\begin{aligned}
	i>i_k\Rightarrow s_i=0\\
	i_{a-1}< j,i\leq i_{a}\Rightarrow s_j=s_i>0 \\
	i_{a-1}< i\leq i_{a}< j\leq i_{a+1},\quad \, \Rightarrow s_i>s_j>0
	\end{aligned}
	\end{equation*}
Which, up to the permutation and translation, are precisely the conditions in our statement.\qed
\end{proof}

\begin{corollary}
Each point $p\in \RR^n$ belongs to the open convex hull $K^o_\alpha$, for a unique abstract cell $\alpha\in V$ of the $n$-D CFK simplicial complex defined in \ref{asc}.
\end{corollary}
\begin{proof}
If $p\in\mathbb{Z}^n$, then $p\in\{v\}=K^o_v$ for a unique $0$-cell $v\in V_0=\ZZ^n$, and $p\notin K^o_\alpha$ for any other open cell $\alpha\in V_k$ with $k\geq 1$, because these open cells don't contain elements of $\ZZ^n$, as proved by the previous proposition.

If $p\notin\mathbb{Z}^n$, we may try to determine cells $\alpha\in V_k$ for which $p\in K^o_\alpha$. Following the previous proposition any cell determined by $(v,(S_1,S_2,\ldots,S_k))$ contains $p$ in its open convex hull if and only $v$ coincides with the integer component $c(p)$, and if the fractional component $f(p)\in[0,1[^n$ takes precisely $k$ different nonvanishing values $1>s_1>s_2>\ldots>s_k>0$ at positions determined by the nonempty disjoint sets $S_1,\ldots, S_k$. Consequently taking the integer component $c(p)$ and ordering the fractional components $f_i(p)$ associated to $p$ univocally determines the unique pair $(v,\mathcal{S})$ and abstract simplex $\alpha(v,\mathcal{S})$ whose associated open convex hull contains $p$.\qed
\end{proof}

The following corollary shows that CFK partition is formed by so-called path-simplices, determined by some initial vertex, following a path formed by a finite sequence of consecutive orthogonal edges.
\begin{corollary}\label{indexingncells}
The set $V_n$ of $n$-cells of the $n$-D CFK simplicial complex can be identified with the set of pairs $(v,\sigma)$ where $v\in\ZZ^n$ and $\sigma\in\mathrm{Sym}(n)$ is a permutation of the set $[n]$.
	\begin{equation*}
	(v,\sigma)\in \ZZ^n\times \mathrm{Sym}(n)\mapsto \alpha(v,(\sigma_1,\ldots,\sigma_n))=\{v_0,\ldots,v_n\}\in V_n \qquad v_0=v,\quad v_{i+1}=v_i+e_{\sigma_i}
	\end{equation*}
\end{corollary}
Hence any $n$-cell can be seen as a sequence of vertices starting at some vertex $v_0$ and ending with $v_0+(1,\ldots,1)$, using $n$ jumps by integer unit vectors.

The model of an abstract simplex as an open convex hull will allow us to introduce the incidence mapping. For any $k$-cell $\alpha$ determined by $(v,S_1,\ldots, S_k)$, the previous proposition shows that $K_\alpha^0$ is an open subset contained in the supporting affine subspace $H_\alpha\subseteq\mathbb{R}^n$:
	\begin{equation*}
		H_\alpha=\left\{\begin{aligned} &x_i=x_i(v),\quad \forall i\notin S_1\cup\ldots\cup S_k \\
		&x_{i}-x_{j}=x_i(v)-x_j(v), \quad \forall a=1\ldots k,\,\forall i,j\in S_a\end{aligned} \right.
	\end{equation*}
given by certain equations $x_i=\mathrm{cte}\in\mathbb{Z}$, $x_i-x_j=\mathrm{cte}\in\mathbb{Z}$. Moreover $H_\alpha$ is the affine subspace spanned by all nodes $v\prec\alpha$, and $K^0_\alpha$ is the interior of the convex hull of these nodes on $H_\alpha$. 

Considering $m_1=\min S_1$, $m_2=\min S_2$,$\ldots$, $m_k=\min S_k$, the corresponding coordinate functions $x_{m_1},\ldots,x_{m_k}$ turn out to be a system of affine coordinates on $H_\alpha$. We may rearrange $m_1,\ldots, m_k$ into a monotone sequence of indices $i_1<i_2<\ldots<i_k$ and we may call conventional orientation on $H_\alpha$ the one given by $\vol_\alpha=\dd x_{i_1}\wedge\ldots\wedge\dd x_{i_k}$. Incidence of any two adherent cells will be defined in terms of the compatibility of these orientations and the associated outward pointing vector:
\begin{define}\label{incidenceCFK}
The incidence mapping $[\beta\colon\alpha]\in \{0,\pm 1\}$ on the CFK simplicial complex is given by:
    \begin{equation*}
		\begin{aligned}
    &\beta\in V_{k+1},\, \bar{v}\in \alpha=\beta\setminus\{v\}\in V_k &\Rightarrow
    [\beta\colon\alpha]\cdot \vol_\alpha\sim
     i_{\bar v-v}\vol_\beta \\
    &\text{ else } [\beta\colon\alpha]=0
    \end{aligned}
    \end{equation*}
where the equivalence $\sim$ holds if both affine forms determine the same orientation on $H_\alpha$, and where $\vol_\alpha$, $\vol_\beta$ represent the affine conventional orientation forms, as describe above.
\end{define}

\begin{example}
Consider $V$ the $5-D$ CFK simplicial complex. The whole space  $\mathbb{R}^5$ may be decomposed as the disjoint union of the open convex hulls $K^o_\alpha$. Consider the point $p=(5.4,2.1,4.3,-1.7,6)\in\mathbb{R}^5$. Its integer and fractional components are, respectively $c=(5,2,4,-2,6)\in\ZZ^5$ and $f=(0.4,0.1,0.3,0.3,0)\in [0,1[^5$. There are three nonvanishing fractional components, which may be ordered as $0.4>0.3>0.1$, the first component at positions $S_1=\{1\}$, the second one at positions $S_2=\{3,4\}$, and the third one at positions $S_3=\{2\}$.

We may conclude that $p$ belongs to $K^o_\beta$ where $\beta\in V_3$ is a tetrahedron determined by $v=(5,2,4,-2,6)$ and $\mathcal{S}=(S_1,S_2,S_3)$.  The vertices of the simplex $\beta$ would be $v_0=(5,2,4,-2,6)$, $v_1=(6,2,4,-2,6)$, $v_2=(6,2,5,-1,6)$, $v_3=(6,3,5,-1,6)$. 

The supporting subspace $H_\beta$ is given by equations $x_5=6$, $x_3-x_4=6$.  Observe that $(\min S_1,\min S_2,\min S_3)=(1,3,2)$, so we may consider the conventional orientation $\vol_\beta=\dd x_1\wedge \dd x_2\wedge \dd x_3$ on $H_\beta$.

A face adherent to $\beta$ would be, for example, the face $\alpha\in V_2$ determined by the vertices $v_0,v_2,v_3$, a 2-simplex associated to $(v,(S_1\cup S_2,S_3))$. As $(\min(S_1\cup S_2),\min(S_3))=(1,2)$, we get the supporting subspace $H_\alpha$ with equations $x_5=6$, $x_3-x_4=6$, $x_1-x_3=1$ oriented by $\vol_\alpha=\dd x_1\wedge \dd x_2$

Observe that $v_3-v_0$, $v_3-v_2$ span the linear space $\vec{H}_\alpha$ and there holds $\vol_\alpha(v_3-v_0,v_3-v_2)=\left| \begin{matrix} 1 & 1 \\ 0 & 1\end{matrix}\right|>0$. On the other hand, taking $v=v_1\in\beta$ we have $\alpha=\beta\setminus\{v\}$ and $\bar{v}=v_3\in \alpha$, so we may compute $$i_{\bar{v}-v}\vol_\beta(v_3-v_0,v_3-v_2)=\vol_\beta(v_3-v_1,v_3-v_0,v_3-v_2)=\left| \begin{matrix} 0&1&1\\ 1 & 1 & 1 \\ 0&1&0  \end{matrix}\right|>0$$

Following now definition \ref{incidenceCFK} we may say $[\beta\colon\alpha]=+1$. In a similar way one may determine the incidence of our tetrahedron with any other adherent face.
\end{example}

\begin{remark}
In particular, points that don't lay on the open convex hull associated to any $n$-cell are points whose coordinates have some integer entry or some pair of entries share a common fractional component. That is, the disjoint union of the open convex hulls of all $n$-cells may be seen as $\RR^n\setminus H$, where $H$ is the closed set obtained by the union of the families of hyperplanes 
	\begin{equation*}
	\{x_i=m\}_{m\in\mathbb{Z},i\in [n]},\qquad  \{x_j-x_i=m\}_{m\in\mathbb{Z},i\neq j\in[n]}
	\end{equation*}
The $n$-D CFK simplicial complex is obtained partitioning $\mathbb{R}^n$ using these hyperplanes. For any $k$-simplex $\alpha$ contained in the $k$-dimensional submanifold $H_k$ (given as intersection of $n-k$ mutually-transversal hyperplanes from the family above), we have a conventional orientation $\vol_\alpha$ not depending on the particular cell $\alpha$, so that any pair of $k$-cells sharing a common supporting space $H_k$ will be considered to induce the same orientation on its supporting subspace. This is derived by choosing as coordinate functions in the submanifold $H_\alpha$ the sequence $x_{i_1},x_{i_2},\ldots, x_{i_k}$ with the lowest possible indices, in increasing order, and taking the classical affine notion of orientation, identifying the submanifold with $\RR^k$ with this coordinate choice.
\end{remark}

\section{The variational problem on a cellular complex}\label{sec3}
A discrete field shall be a particular configuration that vertices of a given abstract cellular complex $V$ adopt on a certain configuration space $Y$, in some sense. Configurations of the cellular complex will be determined by some correspondence taking vertices $v\in V_0$ to elements of the configuration space.

Similar to the smooth case, a discrete variational problem appears when we consider some function (the action functional) that gives a certain real value for each configuration of the field, where the value is obtained by integration of a discrete differential form locally depending on the configuration. An important question in this case is to characterize which configurations are critical for the action functional.

We summarize next the presentation given in \cite{CasiRodr12}:
\begin{define}
We call discrete bundle (or bundle of discrete configurations) on an $n$-dimensional abstract cellular complex $V$, any projection $\pi\colon Y_0\rightarrow V_0$, where each fiber $\pi^{-1}(v)=(Y_0)_v$ is a non-empty smooth manifold. Any element $y_v\in (Y_0)_v$ shall be called a configuration on the vertex $v\in V_0$. A section $y\colon V_0\rightarrow Y_0$ of the projection $\pi$ shall be called a discrete field on the cellular complex $V$, with values on $Y_0$.

\end{define}
Observe that discrete bundles can be seen as smooth bundles, if we see the base space $V_0$ as a totally disconnected 0-dimensional manifold.

\begin{remark}\label{remark31a}
A particular case of discrete bundle appears when we consider a smooth bundle $\pi\colon Y\rightarrow X$ on some smooth manifold $X$, and an injective mapping $x\colon V_0\rightarrow X$. The space $$Y^x=x^*Y=\{(v,y)\in V_0\times Y\,\colon\, x(v)=\pi(y)\}$$ together with the projection $(v,y)\mapsto \pi(y)=x(v)$ is a discrete bundle on $V$.
\end{remark}
\begin{define}
For any discrete bundle $\pi\colon Y_0\rightarrow V_0$ on the $n$-dimensional abstract cellular complex $V$, and for any $k\in\{0,1,\ldots, n\}$ we call $\pi_k\colon Y_k\rightarrow V_k$ the bundle whose fiber at any cell $\alpha\in V_k$ is  
	\begin{equation*}
	(Y_k)_\alpha=\prod\limits_{(v\in V_0)\prec\alpha} (Y_0)_v
	\end{equation*}
the product of all fibers $\pi^{-1}(v)=(Y_0)_v$, for every vertex $v$ adherent to $\alpha$.

Elements $y_\alpha\in (Y_k)_\alpha$ can be seen as sections of the bundle $Y_0$ restricted to the set of vertices $v\prec\alpha$. To simplify the notations, for any $k$-cell $\alpha\in V_k$ and any discrete bundle $\pi\colon Y_0\rightarrow V_0$ we shall denote $(Y_k)_\alpha$ as $Y_\alpha$. In particular, we shall write $Y_v$ instead of $(Y_0)_v=\pi^{-1}(v)$. Therefore $Y_\alpha=\prod\limits_{v\prec\alpha} Y_v$.
\end{define}

\begin{define}
For any discrete bundle $\pi\colon Y\rightarrow V$ (the base space $V$ is discrete, but its fibers are smooth manifolds), we may consider the bundle $VY\rightarrow Y$, whose fiber at a configuration $y_v\in Y$ is the tangent space $T_{y_v}Y_v$ of the fiber (where $v=\pi(y_v)\in V_0$). This shall be called the vertical bundle associated to $Y$, and elements of this vertical bundle on a point $y_v\in Y$ shall be called infinitesimal variations of this configuration and represented as $\delta y_v$.\newline For any given section $y\in \Gamma(V_0,Y_0)$, we call infinitesimal variation $\delta y$ of the section $y$ any mapping taking each vertex $v\in V_0$ to a tangent vector $\delta y_v\in T_{y_v}Y_v$, that is, a section $\delta y\in\Gamma(V_0,y^*VY)$ of the pull-back of the bundle $VY\rightarrow Y_0$ by $y\colon V_0\rightarrow Y_0$.
\end{define}
Observe that for any bundle of configurations $\pi\colon Y_0\rightarrow V_0$ and for the associated bundles $\pi_k\colon Y_k\rightarrow V_k$ and $VY_0\rightarrow Y_0\rightarrow V_0$, there holds that the vertical bundle $VY_k\rightarrow Y_k\rightarrow V_k$ coincides with $(VY_0)_k\rightarrow V_k$ (because the tangent space at any point of a product manifold is the product of the corresponding tangent spaces). Giving an infinitesimal variation $\delta y_\alpha$ of a configuration $y_\alpha=(y_v)_{v\prec\alpha}\in Y_\alpha$ at some $k$-cell $\alpha\in V_k$ is the same as giving a set of infinitesimal variations $\delta y_v$ of the configurations $y_v$, for all the adherent vertices $v\prec\alpha$.

\begin{remark}
If we choose $V$ to be a simplicial complex and fix a manifold $Q$, the space $Y_0=V_0\times Q$ is the trivial bundle associated to $Q$. Giving an element of $Y_0$ is the same as giving a vertex $v\in V_0$ and a configuration $q\in Q$ for that vertex. An element of $Y_k$ is the same as giving a $k$-simplex belonging to $V$ (an unordered set of $k+1$ vertices $\{v_0,\ldots, v_k\}\in V_k$), and configurations $q_0,\ldots,q_k\in Q$, each one associated to each vertex of the simplex. An infinitesimal variation at $y_v=(v,q)\in Y_0$ is defined by $\delta y_v=(v,D_q)$, where $D_q$ is a tangent vector at $q\in Q$. Giving an infinitesimal variation $\delta y_\alpha\in VY_k$ is the same as giving a $k$-simplex, and tangent vectors $D_{q_0},D_{q_1},\ldots,D_{q_k}$ at the points $q_0,\ldots, q_k$, each one associated to each vertex of the simplex.
\end{remark}
Together with the configuration bundle, the second main ingredient to determine a variational principle is the Lagrangian density:
\begin{define}
Given a discrete bundle $\pi\colon Y_0\rightarrow V_0$ on the space of vertices of some $n$-dimensional abstract cellular complex, we call discrete Lagrangian density any smooth mapping $L\colon Y_n\rightarrow \mathbb{R}$.
\end{define}
\begin{remark}
As $Y_n$ is the discrete union of its fibers $Y_\beta$ ($\beta\in V_n$), and each of these fibers is a direct product of manifolds $Y_v$, giving a discrete Lagrangian density will be the same as giving a family of smooth functions $L_\beta\colon Y_\beta=\prod\limits_{v\prec\beta}Y_v\rightarrow \mathbb{R}$. One may consider in particular the case that $V$ is a simplicial complex, that $Y_0=V_0\times Q$ (the trivial bundle), and that vertices adherent to any $n$-cell are ordered by some convention (so that $Y_n=V_n\times\prod\limits_{k=0}^nQ$), in this case giving the discrete Lagrangian density is giving $L_\beta\colon Q^{\times(n+1)}\rightarrow \mathbb{R}$ for each $\beta\in V_n$. Observe that instead of choosing an ordering of vertices, we may choose an orientation, that is, one of the two possible classes of orderings, determined with respect to positive permutations.  We may then avoid to fix any particular ordering if we assume that $L_\beta$ is invariant with respect to positive permutations $\sigma\in \mathrm{Alt}_{n+1}\subseteq \mathrm{Sym}_{n+1}$ on $Q^{\times (n+1)}$. This assumption resembles the situation of smooth Lagrangian densities, which are differential $n$-forms. Moreover in the most simple case, we might take the same function $L\colon Q^{\times(n+1)}\rightarrow \mathbb{R}$ for every n-cell $\beta\in V_n$. In the particular situation described above, a discrete Lagrangian density shall be a function $L\colon Q^{\times(n+1)}\rightarrow \mathbb{R}$ that is invariant with respect to the action of the alternate permutations group $\mathrm{Alt}_{n+1}$.

To get more general results, we don't necessarily assume the particular situation in this remark.
\end{remark}
Any section $y\in\Gamma(V_0,Y_0)$ induces a section $y_n\in\Gamma(V_n,Y_n)$, and computing the values of the Lagrangian density $L$ on this section, we get a $n$-cochain $L(y)\in \Omega^n(V)$, whose value at any $n$-cell $\beta\in V_n$ is simply $L_\beta(y_\beta)=L_\beta\left((y_v)_{v\prec\beta}\right)$. The integration of this cochain on bounded domains determines a discrete action functional, as follows:
\begin{define}\label{def36}
For any finite subset $\Aset\subset V_n$ (the domain of integration) and discrete Lagrangian density $L\colon Y_n\rightarrow \mathbb{R}$, we call action functional $\LL_{\Aset}$ the following mapping:
    \begin{equation*}
    \begin{array}{rrcl}
    \mathbb{L}_{\Aset}\colon & \Gamma(V_0,Y_0)& \rightarrow & \mathbb{R}\\
     & y=(y_v)_{v\in V_0} & \mapsto & \langle c_{\Aset},L(y)\rangle=\sum\limits_{\beta\in \Aset} L_\beta(y_\beta)
    \end{array}
    \end{equation*}
where $c_{\Aset}\in C_n(V,\mathbb{Z})$ is the characteristic chain associated to $\Aset$, and $L(y)\in\Omega^n(V)$ is the $n$-cochain whose value at some $n$-cell $\beta\in V_n$ is $L_\beta(y_\beta)$.
\end{define}
\begin{define}
We call differential of the action functional $\dd_y\LL_{\Aset}$ at some configuration $y\in\Gamma(V_0,Y_0)$ the following linear mapping:
    \begin{equation*}
    \begin{array}{rrcl}
    \dd_y\LL_{\Aset}\colon & \Gamma(V_0,y^*VY_0)& \rightarrow & \mathbb{R}\\
     & \delta y=(\delta y_v)_{v\in V_0} & \mapsto & \sum\limits_{\beta\in \Aset}(\dd_{y_\beta} L_\beta)(\delta y_\beta)
    \end{array}
    \end{equation*}
\end{define}
Observe that the action functional and the differential defined above make sense only if the sum is finite, therefore our domains of integration must be finite to define the functional. The differential $\dd_y\mathbb{L}_{\Aset}$ depends only on the values of $\delta y$ on vertices adherent to some $n$-cell $\beta$ where the chain $c_{\Aset}$ doesn't vanish (hence $\beta\in \Aset$). This means that $\dd_y\LL_{\Aset}$ only depends on the values of $\delta y$ on a finite number of vertices. To study this differential we may restrict ourselves to the subspace $\bigoplus V_{y_v}Y_0\subset\Gamma(V_0,y^*VY_0)$, which consists of those sections of $y^*VY_0$ that vanish on all but a finite number of vertices.

We call $\dd_y\LL$ the linear mapping:
    \begin{equation*}
    \begin{array}{rrcl}
    \dd_y\LL\colon & \bigoplus V_{y_v}Y_0& \rightarrow & \mathbb{R}\\
     & \delta y=(\delta y_v)_{v\in V_0} & \mapsto & \sum\limits_{\beta\in V_n}(\dd_{y_\beta} L_\beta)(\delta y_\beta)
    \end{array}
    \end{equation*}
We have now an expression that is independent of any domain of integration, and well-defined as infinitesimal variations $\delta y\in\bigoplus V_{y_v}Y_0$ vanish at every vertex except for a finite number of them.

Suppose that $\dd_y\LL$ vanishes. Then for any infinitesimal variation $\delta y\in \bigoplus V_{y_v}Y_0$ and for $\Aset\subset V_n$ big enough (i.e.~$\Aset$ that contains the spheres $S_v$ for each $v$ with $\delta y_v\neq 0$), there holds $\dd_y\LL_{\Aset}(\delta y)=\dd_y\LL(\delta y)=0$ (this would reflect that the differential at $y$ of the action functional $\LL_{\Aset}$ for the discrete domain $\Aset$ vanishes when applied to infinitesimal variations whose support is interior to $\Aset$).

Discretizing the formulation in \cite{GarcGarcRodr06} for general variational problems with arbitrary admissible variations:
\begin{define}\label{def38}
A variational problem is defined by fixing a discrete bundle $\pi\colon Y_0\rightarrow V_0$, a discrete Lagrangian density $L\colon Y_n\rightarrow \RR$, a subset $\mathrm{Adm}\subseteq \Gamma(V_0,Y_0)$ of admissible sections and a space of admissible infinitesimal variations $\Adm_y\subset \bigoplus V_{y_v}Y_0\subset \Gamma(V_0,y^*VY_0)$ associated to each section $y\in\Gamma(V_0,Y_0)$.
\end{define}
\begin{define}\label{def39}
We say a section $y\in\Gamma(V_0,Y_0)$ of the discrete bundle $Y_0$ is critical for the variational problem given by discrete Lagrangian $L$, admissible sections $\mathrm{Adm}$, and admissible variations $\Adm$ if $y\in\mathrm{Adm}$ and $\dd_y\LL$ vanishes on the space $\Adm_y$ of admissible infinitesimal variations at $y$.
\end{define}
The definition of $\dd_y\LL$ on $\bigoplus V_{y_v}Y_0$ resembles the definition of the first variation of some action functional for smooth variational problems. In the smooth case, different variational principles arise and sections of the bundle $Y$ are called critical with respect to these principles depending on the choice of the infinitesimal variations on which the action functional is stationary. The simplest case is the choice of vector fields with compact support, leading to the known as  fixed boundary problem, whose critical sections are characterised by means of Euler-Lagrange equations. Other choices of admissible variations make sense in different situations, leading to other variational principles and to Euler-Poincar\'{e} equations, Lagrange-Poincar\'{e} equations, Lagrange multiplier rules, and equations of vakonomic or non-holonomic mechanics, for example, that characterize critical sections in these cases.

The question now is how to characterize critical sections in the discrete setting. The answer depends on the choice of the spaces $\Adm_y$ of admissible infinitesimal variations and on the topological structure of the discrete space (abstract cellular complex) where the whole theory is modelled. For the case of the fixed boundary problem, where $\mathrm{Adm}=\Gamma(V_0,Y_0)$, and $\Adm_y=\bigoplus V_{y_v}Y_0$, critical sections are characterised by a discrete analogue of Euler-Lagrange equations:

\begin{theorem}[Discrete Euler-Lagrange characterization of critical sections]\label{thm310}
For the variational problem defined on $\mathrm{Adm}=\Gamma(V_0,Y_0)$ by the discrete Lagrangian density $L\colon Y_n\rightarrow \RR$ and admissible infinitesimal variations $\Adm_y=\bigoplus V_{y_v}Y_v$ (we call this the fixed boundary variational problem), a section $y\in\Gamma(V_0,Y_0)$ is critical if and only if at each vertex $v\in V_0$ the discrete Euler-Lagrange form $\EL(y)\in\Gamma(V_0,y^*V^*Y_0)$ vanishes:
        \begin{equation*}
        0=\EL_v(y)\in V^*_{y_v}Y_v,\quad \forall v\in V_0\quad \text{(discrete E.-L. equations)}
        \end{equation*}
where
	\begin{equation}\label{ELbis}
	\EL_v(y)=\sum\limits_{\beta\in S_v} (\dd_{y_\beta}L_\beta)\circ i_{y_v}^{y_\beta}\in V^*_{y_v}Y_v
	\end{equation}
is defined through the composition of $\dd_{y_\beta}L_\beta\in V^*_{y_\beta}Y_\beta$ with the natural immersions $i_{y_v}^{y_\beta}\colon V_{y_v}Y_v\hookrightarrow V_{y_\beta}Y_\beta=\bigoplus\limits_{\bar v\prec\beta} V_{y_{\bar v}}Y_{\bar v}$, for all $n$-cells $\beta\in V_n$ with $v\prec\beta$.
\end{theorem}
\begin{proof}
By definition, a section $y\in\Gamma(V_0,Y_0)$ is critical for the variational problem defined by $L$ and by  $\Adm_y=\bigoplus V_{y_v}Y_0$ if and only if $\dd_y\LL(\delta y)=0$, for every $\delta y\in \bigoplus V_{y_v}Y_0\subset \Gamma(V_0,y^*VY_0)$. As $\dd_y\LL$ is linear, this holds if and only if $\dd_y\LL(\delta y)=0$ for every vertex $v\in V_0$ and any $\delta y\in V_{y_v}Y_0\subset \Gamma(V_0,y^*VY_0)$. This inclusion associates to any fixed $\delta y_v\in V_{y_v}Y_v$  the section $\delta y\in \Gamma(V_0,y^*VY_0)$ whose value at $v$ is $\delta y_v$ and whose value is 0 elsewhere.

If $\delta y\in\Gamma(V_0,y^*VY_0)$ is defined by $\delta y_v\in V_{y_v}Y_v$ we get:
    \begin{equation*}
    \dd_y\LL(\delta y)=\sum\limits_{\beta\in V_n}(\dd_{y_\beta} L_\beta)(\delta y_\beta)=\sum\limits_{\beta\in S_v}(\dd_{y_\beta} L_\beta)(\delta y_\beta)=\sum\limits_{\beta\in S_v}(\dd_{y_\beta} L_\beta)\circ i_{y_v}^{y_\beta} (\delta y_v)
    \end{equation*}
because $\delta y_\beta=(\delta y_{\bar v})_{\bar v\prec\beta}$ is zero whenever $v$ is not adherent to $\beta$, and the sphere $S_v\subset V_n$ centered at $v$ is precisely the set of $n$-cells $\beta\in V_n$ that contain $v\in V_0$. For these cells, $\delta y_\beta\in V_{y_\beta}Y_\beta$ is by definition $i_{y_v}^{y_\beta}(\delta y_v)$.

As we define $\EL_v(y)=\sum\limits_{\beta\in S_v} (\dd_{y_\beta}L_\beta)\circ i_{y_v}^{y_\beta}\in V^*_{y_v}Y_v$, the proof is completed.  \qed
\end{proof}
The section $\EL_v(y)\in\Gamma(V_0,y^*(V^*Y_0))$ defined in  (\ref{ELbis}) characterizes critical sections by $\EL(y)=0$, which plays in the discrete theory the same role as Euler-Lagrange equations (\ref{EL}) do in the smooth setting.

We shall now introduce a discrete version of Noether's theorem. In \cite{CasiRodr12,CasiRodr12b} a discrete version of this theorem is obtained, in terms of cochains, and a physical interpretation is given. The discrete Noether theorem fully reflects all the aspects from smooth theory and the concrete expression for conserved Noether currents relies on a particular combinatorial expression whose origin is due to the choice of the n-D cubic cellular complex as discrete space in those papers. In this article we won't restrict ourselves to cartesian grids and shall explore the conservation laws for a wider class of discrete spaces. Therefore we won't express our results in terms of cochains but rather in terms of its integration, according to the physical interpretation of Noether currents in those references.

\begin{define}\label{definesimetria}
Consider a configuration $y\in\Gamma(V_0,Y_0)$. We say a smooth vector field $D\in\mathfrak{X}(Y_0)$ is an infinitesimal symmetry at $y$ for the discrete Lagrangian $L\colon Y_n\rightarrow \mathbb{R}$ if at any $n$-cell $\beta\in V_n$ the corresponding restriction $D_{y_\beta}\in \oplus_{v\prec\beta} T_{y_v}Y_v=T_{y_\beta} Y_\beta$ is incident with $\dd_{y_\beta}L_\beta$, that is, if
        $$(\dd_{y_\beta}L_\beta)(D_{y_\beta})=0$$
We say that $D\in\mathfrak{X}(Y_0)$ is an infinitesimal symmetry for the discrete lagrangian $L$ (now, regardless of any given configuration $y$) if the previous formula holds for any $y_\beta\in Y_n$.
\end{define}
\begin{remark}
Any fibered automorphism $\varphi\colon Y_0\rightarrow Y_0$ acts on discrete Lagrangians $L$ by $(\varphi\cdot L)_\beta((y_v)_{v\prec\beta})=L_\beta(\varphi^{-1}y_v)_{v\prec\beta}$. Any one-parameter group $\{\varphi_t\}_{t\in\mathbb{R}}$ of vertical automorphisms of the bundle $Y_0\rightarrow V_0$ such that $\varphi_t\cdot L=L$ has then as generator an infinitesimal symmetry for the discrete Lagrangian $L$.
\end{remark}
\begin{theorem}[Discrete Noether's theorem]
Let $\Aset\subset V_n$ be a finite collection of $n$-cells, and let $y\in\Gamma(V_0,Y_0)$. If $D\in\mathfrak{X}(Y_0)$ is an infinitesimal symmetry at $y$ for the discrete lagrangian density $L\colon Y_n\rightarrow\RR$, then:
        \begin{equation*}
        0=\sum_{v\in\Int \Aset} (\EL_v(y))(D_{y_v})+\sum_{(v\in\Fr \Aset)\prec(\beta\in \Aset)} (\dd_{y_\beta} L_\beta)(i_{y_v}^{y_\beta}D_{y_v})
        \end{equation*}
If $y\in\Gamma(V_0,Y_0)$ is critical for the fixed boundary problem, there holds the Noether conservation law:
			\begin{equation}\label{32a}
			0=\sum_{(v\in\Fr \Aset)\prec(\beta\in \Aset)} (\dd_{y_\beta} L_\beta)(i_{y_v}^{y_\beta}D_{y_v})
			\end{equation}
or equivalently:
    \begin{equation}\label{NoetConsLaw}
    0=\sum_{\beta\in \Aset} (\dd_{y_\beta} L_\beta)(\bar D_\beta)
    \end{equation}
where $\bar{D}\in\Gamma(V_0,y^*VY_0)$ denotes the section defined by $\bar{D}_v=D_{y_v}$ for $v\in\Fr \Aset$,	and $0$ elsewhere.
\end{theorem}
\begin{proof}
In the conditions given for $D$, $y$, $L$ we have:
    \begin{equation*}
    \begin{aligned}
    0&=\sum_{\beta\in \Aset} (\dd_{y_\beta} L_\beta)(D_{y_\beta})=\sum_{\beta\in \Aset}\left(\dd_{y_\beta} L_\beta\right)\left(\sum_{v\prec\beta} i_{y_v}^{y_\beta}D_{y_v}\right)=\sum_{\beta\in \Aset}\sum_{v\prec\beta} (\dd_{y_\beta} L_\beta)(i_{y_v}^{y_\beta}D_{y_v})=\\
    &=\sum_{v\in\Int \Aset}\sum_{v\prec\beta\in \Aset} (\dd_{y_\beta} L_\beta)(i_{y_v}^{y_\beta}D_{y_v})+\sum_{v\in\Fr \Aset}\sum_{v\prec\beta\in \Aset} (\dd_{y_\beta} L_\beta)(i_{y_v}^{y_\beta}D_{y_v})+\sum_{v\in\Ext \Aset}\sum_{v\prec\beta\in \Aset} (\dd_{y_\beta} L_\beta)(i_{y_v}^{y_\beta}D_{y_v})
    \end{aligned}
    \end{equation*}
For $v\in \Ext \Aset$, there is no $n$-cell $\beta$ with $v\prec\beta\in \Aset$. For $v\in\Int \Aset$ one has $v\prec\beta\in \Aset\Leftrightarrow \beta\in S_v$, so using the definition of $\EL_v(y)$, we conclude the first formula in our theorem.

The second formula is a direct consequence of this one, because $\EL_v(y)$ vanishes for critical $y\in\Gamma(V_0,Y_0)$.

The last expression is a direct consequence of $\bar D_\beta=\sum\limits_{(v\in\Fr \Aset)\prec \beta} i_{y_v}^{y_\beta}D_{y_v}$, valid at any $n$-cell $\beta\in V_n$ because of the definition of $\bar{D}\in\Gamma(V_0,y^*VY_0)$.\qed
\end{proof}
Discrete Noether conservation law (\ref{32a}) is a discrete analogue of (\ref{NoetSmooth}) from smooth variational calculus, and can be written in terms of discrete differential and integration of (n-1)-forms for the cubic simplicial complex (see \cite{CasiRodr12}).

\subsection*{Relating variational problems in the CFK simplicial and the cubic cellular complexes}

The discrete variational principles given for our two different examples of $n$-D abstract cellular complexes may be related one to another in a natural way:

Let $V$ be the $n$-D CFK simplicial complex and $\overline{V}$ the $n$-D cubic cellular complex. The corresponding spaces of vertices coincide, $V_0=\mathbb{Z}^n=\overline{V}_0$. We have a natural identification $i_0\colon v\in V_0\to \overline{V}_0$ therefore any discrete bundle $Y$ on $V$ is also a discrete bundle $\overline{Y}$ on $\overline{V}$, and conversely.  Also discrete fields for the configuration bundle $Y\rightarrow V$ coincide with discrete fields for the configuration bundle $\overline{Y}\rightarrow\overline{V}$. However, many notions, in particular the extension to $n$-cells $Y_n$, $\overline{Y}_n$, and the notion of discrete lagrangian density have different meanings in one discrete space or the other.

We may easily observe that any $n$-simplex $\beta\in V_n$ has its adherent vertices contained in a unique hypercube $n$-cell $\overline{\beta}\in \overline{V}_n$, determining so a natural mapping $i\colon V_n\rightarrow\overline{V}_n$ that takes any $n$-simplex of the CFK complex into the unique hypercube of the cubical complex that contains this simplex. More precisely, if $\beta\in V_n$ is determined by $\beta=\alpha(v,\sigma)$ for some $v\in\ZZ^n$ and permutation $\sigma$ (as defined in corollary \ref{indexingncells}), then $i(\beta)=\bar{\beta}$ is given as $\bar\beta=v+\frac12(1,\ldots,1)\in(\frac12\ZZ)^n=\overline{V}$.

Moreover, for any discrete bundle $Y=\overline{Y}$ modelled on $V$ and any $n$-simplex $\beta\in V_n$, the set of vertices adherent to $\beta$ is transformed by $i_0$ into a subset of vertices adherent to $i(\beta)$, so we may consider the restriction of discrete sections defined on the hypercube $i(\beta)$ to discrete sections defined on the simplex $\beta$, a morphism $proj_\beta\colon\overline{Y}_{i(\beta)}\rightarrow Y_\beta$ for each $\beta\in V_n$.

Consider now any discrete Lagrangian density $L\colon Y_n\rightarrow \mathbb{R}$. There exists an induced discrete Lagrangian density $\overline{L}\colon \overline{Y}_n\rightarrow\mathbb{R}$, namely:
	\begin{equation}\label{simplexcubic}
	\overline{L}_{\overline{\beta}}=\sum_{\beta\in V_n\colon i(\beta)=\overline{\beta}} L_\beta\circ proj_\beta\qquad (\overline\beta\in\overline{V}_n)
	\end{equation}
each component $\overline{L}_{\overline{\beta}}$ is determined by addition of components $L_\beta$ for $n!$ different simplicial cells that represent a partition of $K_{\overline{\beta}}$.

\begin{define}
Given any discrete bundle $Y\rightarrow \ZZ^n$ and lagrangian density $L\colon Y_n\rightarrow \mathbb{R}$ defined on CFK $n$-simplices, and for discrete fields of the discrete bundle $Y\rightarrow \ZZ^n$, we call the function $\overline{L}\colon \overline{Y}_n\rightarrow\mathbb{R}$ given in (\ref{simplexcubic}) the discrete Lagrangian density defined on cubic $n$-cells, induced by $L$. 
\end{define}
The following statements are then a direct consequence of our definitions
\begin{prop}
The Euler-Lagrange form $\EL_v(y)$ at some vertex $v\in V_0=\overline V_0$, for some discrete section $y\in\Gamma(X_0,V_0)=\Gamma(X_0,\overline{V}_0)$ for the lagrangian density $L\colon Y_n\rightarrow \mathbb{R}$ coincides with the Euler Lagrange form associated to the same elements, $v,y$, for the associated Lagrangian density $\bar{L}\colon \overline{Y}_n\rightarrow \mathbb{R}$
\end{prop}
\begin{prop}
Consider a discrete bundle $\overline{V}_0Y_0\rightarrow V_0=\overline{V}_0=\mathbb{Z}^n$ and discrete fields $y=\overline{y}\circ i_0$. Any infinitesimal symmetry $D\in\X(y^*Y_0)=\X(\bar y^*\bar Y_0)$ at $y$ of the lagrangian density $L$ defined on CFK $n$-simplices is also an infinitesimal symmetry at $\bar{y}$ of the induced lagrangian density $\overline{L}$ defined on cubic $n$-cells.
\end{prop}
\begin{prop}
For any finite subset of cubic n-cells $\overline{\Aset}\subset \overline{V}_n$, taking the set $\Aset=i^{-1}\left(\overline{\Aset}\right)\subset V_n$, there holds:
\begin{equation*}
\sum_{(v\in\Fr \overline{\Aset})\prec(\overline\beta\in \overline{\Aset})} (\dd_{y_{\overline{\beta}}} {\overline L}_{\overline\beta})(i_{y_v}^{y_{\overline\beta}}D_{y_v})=
\sum_{(v\in\Fr \Aset)\prec(\beta\in \Aset)} (\dd_{y_\beta} L_\beta)(i_{y_v}^{y_\beta}D_{y_v})
\end{equation*}
That is, for each infinitesimal variation $D$ defined along any discrete field $y$, the associated Noether current on the boundary of a given domain $\overline{\Aset}$ of the n-D cubic complex coincides with the Noether current associated to the same infinitesimal variation on the boundary of the associated domain $\Aset$ on the CFK simplicial complex
\end{prop}
Summarizing, a variational problem on the CFK simplicial complex naturally leads to a variational problem on the cubic cellular complex. Critical discrete fields are the same, for both discrete Lagrangian densities, Noether currents coincide, and symmetries for the variational principles for both discrete Lagrangians are the same, leading to the same conserved currents. All results in \cite{CasiRodr12} lead thus to corresponding results for the case of the CFK simplicial complex.

\section{Variational Integrators}\label{secintegrator}

Let us now put our focus on Euler-Lagrange form 	(\ref{ELbis}) associated to any section $y\in\Gamma(V_0,Y)$, that characterizes critical sections of the variational problem defined by $L=(L_\beta)_{\beta\in V_n}$. It can be seen as a system of difference equations, discrete analogue of Euler-Lagrange PDEs of smooth theories.

Integration algorithms  for the Cauchy problem determined by a system of partial difference equations (like (\ref{ELbis})) and an initial condition band given on $\mathbb{Z}^n$ have been studied, for example in \cite{AdlVes04,FadVol94}. We aim to derive such algorithms, for the set of Euler-Lagrange equations, for fields given on a rather general cellular complex. Cauchy initial data will be integrated in the direction determined by some discrete flow:
	\begin{define}
	We call first order incremental flow on the discrete space $V$ any mapping $\Delta\colon v\in V_0\mapsto \Delta_v\in V_0$ such that $v\neq \Delta_v$ share a common $n$-cell, for each $v\in V_0$ (there exists $\beta\in V_n$ with $v,\Delta_v\prec \beta$).
	\end{define}
As example, for the n-D cubic cellular complex, or for the n-D CFK simplicial complex where $V_0=\mathbb{Z}^n$, we may consider the first order incremental flow $v\mapsto v+(1,\ldots,1)$ defined on $\ZZ^n$.
	\begin{define}
	Consider, for any vertex $v\in V_0$, the sphere $S_v$ centered at $v$, and define $S^\Delta_v$, the $v$-neighborhood on the direction of the first order incremental flow $\Delta$:
		\begin{equation*}
		S^\Delta_v=\left\{ \beta\in V_n\,\colon\,v,\Delta_v\prec\beta\right\}\subseteq S_v=\left\{ \beta\in V_n\,\colon\,v\prec\beta\right\}
		\end{equation*}
	We denote by $\bar S^\Delta_v\subseteq \bar S_v\subseteq V_0$, the sets of vertices adherent to $S^\Delta_v$ or $S_v$, respectively. That is, the set of vertices adherent to some $n$-cell $\beta\in S^\Delta_v$ (or $\beta\in S_v$, respectively).
	\end{define}

\begin{define}
For any discrete bundle $Y\rightarrow V_0$ and any first order incremental flow $\Delta$ on $V$, the discrete bundles $Y_\Delta\rightarrow V_0$, $Y_S\rightarrow V_0$,  $Y_{S^\Delta}\rightarrow V_0$, are defined as having the following fibers on any vertex $v\in V_0$:
	\begin{equation*}
	(Y_\Delta)_v=Y_{\Delta_v}=(Y_u)_{u=\Delta_v},\quad 
	(Y_S)_v=Y_{S_v}=\prod_{u\in \bar S_v} Y_u ,\quad 
	(Y_{S^\Delta})_v=Y_{S^\Delta_v}=\prod_{u\in \bar S^\Delta_v} Y_u 
	\end{equation*}
\end{define}
Any discrete field $y\in \Gamma(V_0,Y)$ naturally induces corresponding fields $y_\Delta\in\Gamma(V_0,Y_\Delta)$, $y_S\in\Gamma(V_0,Y_S)$,  $y_{S^\Delta}\in \Gamma(V_0,Y_{S^\Delta})$. Moreover, as $v\in \bar S^\Delta_v\subset \bar S_v$, there exist natural projections $Y_S\rightarrow Y_{S^\Delta}\rightarrow Y$ and both sections $y_S, y_{S^\Delta}$ project onto $y$.  Analogously, as $\Delta_v\in \bar S^\Delta_v\subset \bar S_v$, there are natural projections $Y_S\rightarrow Y_{S^\Delta}\rightarrow Y_\Delta$,  and the discrete fields $y_S, y_{S^\Delta}, y_\Delta$ defined above project one onto the next one when we consider these projections.

		\begin{define}\label{defineELT}
We call Euler-Lagrange tensor $EL\colon Y_S\rightarrow V^*Y$ associated to a discrete Lagrangian density $(L_{\beta})_{\beta\in V_n}$, the fibered mapping over $Y$ defined on each fiber according to (\ref{ELbis}):
		\begin{equation*}
		EL_v\left( y_{S_v}\right)=\sum_{\beta\in S_v} (\dd_{y_\beta}L_\beta)\circ i_{y_v}^{y_\beta}\in V^*_{y_v}Y_v\qquad y_{S_v}\in Y_{S_v}, \, y_\beta=\pi_\beta(y_{S_v})\in Y_\beta,\, y_v=\pi_v(y_{S_v})
		\end{equation*}
	where $\pi_\beta\colon Y_{S_v}\rightarrow Y_\beta$, $\pi_v\colon Y_{S_v}\rightarrow Y_v$ are the natural projections defined on each $y_{S_v}=(y_u)_{u\in \bar S_v}\in Y_{S_v}$.
	\end{define}
\begin{remark}
Following theorem \ref{thm310}, a section $y\in\Gamma(V_0,Y)$ is critical for the fixed boundary variational problem determined by some discrete Lagrangian density $(L_{\beta})_{\beta\in V_n}$ if and only if $EL\circ y_S\in\Gamma(V_0,y^*(V^*Y))$ vanishes, for the associated section $y_S\in\Gamma(V_0,Y_S)$.
\end{remark}
When $\dim Y_v=m$, equations $0=EL_v(y_{S_v})$ represent a system of $m$ equations, on the manifold $Y_{S_v}=\prod_{w\in \bar S_v}Y_w$. Fixing $y_u\in Y_u$ for almost each vertex $u\in \bar S_v$ in the sphere, leaving the single component $y_{\Delta_v}\in Y_{\Delta_v}$ as unknown, would lead to a system of $m$ equations on $Y_{\Delta_v}$ that, assuming $\dim Y_{\Delta_v}=\dim Y_v=m$ and some appropriate regularity notion for $EL$, would indicate that for each critical section the single $Y_{\Delta_v}$-component is implicitly determined by Euler-Lagrange equations on $v$, when all the remaining components $y_u$ are fixed on the sphere $S_v$. We shall next make this statement more precise, and describe the notion of integrator, our main tool to explicitly recover the unknown.

Fixing the first order incremental flow $\Delta$, we may decompose the Euler-Lagrange tensor into two components, one related to n-cells that contain $v,\Delta_v$ and the other one related to n-cells that contain $v$ but not $\Delta_v$:
		\begin{equation}\label{otroEL}
		EL\colon y_{S_v}\in Y_S\mapsto \sum_{\beta\in S^\Delta_v} (\dd_{y_\beta}L_\beta)\circ i_{y_v}^{y_\beta}+\sum_{\beta\in S_v\setminus S^\Delta_v} (\dd_{y_\beta}L_\beta)\circ i_{y_v}^{y_\beta}\in V^*Y
		\end{equation}
Observe that the first summand depends only on the projection of $y_{S_v}$ onto $y_{S^\Delta_v}\in Y_{S^\Delta}$. The second summand doesn't depend explicitly on $y_{\Delta_v}$. It is well defined on $Y_{S_v}/Y_{\Delta_v}=\prod\limits_{\substack{u\in \bar S_v \\u\neq \Delta_v}}Y_u$. 

\begin{define}
	Consider a discrete bundle $Y_0\rightarrow V_0$ and the first order incremental flow $\Delta$. We call momentum bundle $Y^*_{S^\Delta}\rightarrow V_0$ the bundle $Y^*_{S^\Delta}=V^*Y\times_Y (Y_{S^\Delta}/Y_\Delta)$, whose fiber is:
		\begin{equation*}
		Y^*_{S^\Delta_v}=\left(V^*Y_v\right)\times_{Y_v} \left(Y_{S^\Delta_v}/Y_{\Delta_v}\right)=V^*Y_v\times\prod_{\substack{u\in\bar S^\Delta_v \\ u\neq v,\Delta_v}} Y_u
		\end{equation*}
Giving an element of the momentum bundle consists on giving an infinitesimal variation $\delta y_v\in V^*_{y_v} Y_v$ of some configuration $y_v\in Y_v$ at some vertex $v\in V_0$, together with other configurations $y_u$, for each vertex $u\neq \Delta_v$ sharing an $n$-cell with $v,\Delta_v$. Elements of the momentum bundle will be denoted by $\delta_v y^*_{S^\Delta}$, representing a determination of almost all components of $y_{S^\Delta_v}$, except for the $\Delta_v$-component, together with an infinitesimal variation $\delta y$ at the single point $v$.
\end{define}
Observe that the momentum bundle can be seen as a modification of the bundle $Y_{S^\Delta_v}$, where the component $Y_v\times Y_{\Delta_v}$ is substituted with $V^*Y_v$. Moreover, both terms apearing in (\ref{otroEL}) can be given as $Y^*_{S^\Delta}$-valued fields:
\begin{define}
Given any first order incremental flow $\Delta$, we call momentum mapping at $v\in V_0$ associated to the discrete Lagrangian density $(L_\beta)_{\beta\in V_n}$ the following:
	\begin{equation}\label{defmomentum}
	\begin{array}{rcl}
	\mu_v\colon Y_{S_v}/Y_{\Delta_v} & \rightarrow &Y^*_{S^\Delta_v}=V^*Y_v\times_{Y_v} \left( Y_{S^\Delta_v}/Y_{\Delta_v}\right) \\
	\left[y_{S_v}\right]& \mapsto & \left(- \sum\limits_{\beta\in S_v\setminus S^\Delta_v} (\dd_{y_\beta}L_\beta)\circ i_{y_v}^{y_\beta}\,,\,
	\left[y_{S^\Delta_v}\right]\right)
 	\end{array}
	\end{equation}
where $y_{S^\Delta_v}, y_\beta, y_v$ represent the corresponding images of $y_{S_v}$ through the natural projections $Y_{S_v}\rightarrow Y_{S^\Delta_v}$ and $Y_{S_v}\rightarrow Y_\beta\rightarrow Y_v$, if $\beta\in S_v$ (Observe that the first projection always factors through $Y_{S_v}/Y_{\Delta_v}\rightarrow Y_{S^\Delta_v}/Y_{\Delta_v}$ and the second one through $Y_{S_v}/Y_{\Delta_v}\rightarrow Y_\beta\rightarrow Y_v$ if $\beta\in S_v\setminus S^\Delta_v$).

We call Legendre transformation at $v$ associated to the same Lagrangian density the mapping:
		\begin{equation}\label{defineLeg}
		\begin{array}{rcl}
		\mathrm{Leg}_v\colon Y_{S^\Delta_v}& \rightarrow & Y^*_{S^\Delta_v}=V^*Y_v\times_{Y_v} 
		\left( Y_{S^\Delta_v}/Y_{\Delta_v}\right) \\
		y_{S^\Delta_v} & \mapsto & \left(\sum\limits_{\beta\in S^\Delta_v} (\dd_{y_\beta}L_\beta)\circ i_{y_v}^{y_\beta},
		\left[y_{S^\Delta_v}\right]\right)
		\end{array} 
		\end{equation}
where again $y_\beta\in Y_\beta$ ,$y_v\in Y_v$ and $[y_{S^\Delta_v}]\in Y_{S^\Delta_v}/Y_{\Delta_v}$ are determined from $y_{S^\Delta_v}$ using the natural projections, when $\beta\in S^\Delta_v$.
\end{define}
Expression (\ref{otroEL}) allows to write the Euler-Lagrange tensor in terms of both bundle morphisms $\mu\colon Y_{S}/Y_{\Delta}\rightarrow Y^*_{S^\Delta}$ and $\mathrm{Leg}\colon Y_{S^\Delta}\rightarrow Y^*_{S^\Delta}$:
	\begin{equation}\label{ELdiferenca}
	EL(y_{S_v})=\mathrm{Leg}_v(y_{S^\Delta_v})-\mu_{v}([y_{S_v}])
	\end{equation}
where $[y_{S_v}]\in Y_{S_v}/Y_{\Delta_v}$ is the quotient class defined by $y_{S_v}$, and $y_{S^\Delta_v}\in Y_{S^\Delta_v}$ is defined from $y_{S_v}$ by the natural projection $Y_S\rightarrow Y_{S^\Delta}$

\begin{remark}\label{remark46}
In the 2-dimensional case, if we take as discrete space the cubic 2D abstract cellular complex $\overline V$, or the 2D CFK simplicial complex $V$, and considering the first order incremental flow $\Delta$ that takes any vertex $v=(i,j)\in \bar{V}_0=V_0$ into $\Delta_v=v+(1,1)$, the set $\bar S^\Delta_v\subset V_0$ has vertices $v_0=v$, $v_1^+=v+(1,0)$, $v_1^-=v+(0,1)$, $v_2=\Delta_v=v+(1,1)$. We have, explicitly:
		\begin{equation*}
		\begin{aligned}
		&Y_{S^\Delta_v}=Y_{v_0}\times Y_{v_2}\times Y_{v_1^+}\times Y_{v_1^-}\\
		&Y^*_{S^\Delta_v}=V^*Y_{v_0}\times Y_{v_1^+}\times Y_{v_1^-}
		\end{aligned}
		\end{equation*}
In the cubic cellular complex, there exists a single squared cell $\bar\beta\in S^\Delta_v\subset \bar V_2$. This cell has $v_0,v_1^+,v_1^-,v_2$ as adherent vertices. For the discrete Lagrangian density $\bar L_{\bar\beta}\colon Y_{v_0}\times Y_{v_2}\times Y_{v_1^+}\times Y_{v_1^-}\rightarrow\RR$ the associated Legendre mapping (\ref{defineLeg}) is:
	\begin{equation*}
	\mathrm{Leg}_{v}(y_0,y_2,y_1^+,y_1^-)=\left( \frac{\partial \bar L_{\bar\beta}(y_0,y_2,y_1^+,y_1^-)}{\partial y_0}\dd y_0 , y_1^+,y_1^-\right)\in V^*_{y_0}Y_{v_0}\times Y_{v_1^+}\times Y_{v_1^-}
	\end{equation*}

In the CFK simplicial cellular complex, there exists two simplicial  cells $\beta^+=\{v_0,v_1^+,v_2\}$, $\beta^-=\{v_0,v_1^-,v_2\}$ in $S^\Delta_v$. For the discrete Lagrangian density $L_{\beta^+}\colon Y_{v_0}\times Y_{v_2}\times Y_{v_1^+}\rightarrow\RR$, $L_{\beta^-}\colon Y_{v_0}\times Y_{v_2}\times Y_{v_1^-}\rightarrow\RR$ the associated Legendre mapping (\ref{defineLeg}) is:
	\begin{equation*}
	\mathrm{Leg}_{v}(y_0,y_2,y_1^+,y_1^-)=\left( \left(\frac{\partial L_{\beta^+}(y_0,y_2,y_1^+)}{\partial y_0}+\frac{\partial L_{\beta^-}(y_0,y_2,y_1^-)}{\partial y_0}\right)\dd y_0 , y_1^+,y_1^-\right)\in V^*_{y_0}Y_{v_0}\times Y_{v_1^+}\times Y_{v_1^-}
	\end{equation*}
We may observe that when $\bar{L}$ is a Lagrangian density on the cubic cellular complex derived from a Lagrangian density on the CFK simplicial complex following (\ref{simplexcubic}), there holds $\bar{L}_{\bar\beta}(y_0,y_2,y_1^+,y_1^-)=L_{\beta^+}(y_0,y_2,y_1^+)+L_{\beta^-}(y_0,y_2,y_1^-)$ and the corresponding Legendre transformations coincide.
\end{remark}

Observe that if $\dim Y_v=\dim Y_{\Delta_v}$, both manifolds $Y_{S^\Delta_v}$ and $Y^*_{S^\Delta_v}$ have the same dimension.
It makes sense to consider then the following notion of regularity:
\begin{define}
The discrete Lagrangian density $(L_\beta)_{\beta\in V_n}$ is regular in the direction determined by the first order incremental flux $\Delta$ if there exists a smooth mapping

		\begin{equation*}
		\Phi^\Delta\colon Y^*_{S^\Delta}\rightarrow Y_{S^\Delta}
		\end{equation*}
such that $\mathrm{Leg}\circ\Phi^\Delta=\Id$ on $Y^*_{S^\Delta}$.
\end{define}
A necessary condition for regularity is that $\mathrm{Leg}$ is a local diffeomorphism at each point. In the case that $\mathrm{Leg}$ is a local diffeomorphism at some point $y_{S^\Delta_v}\in Y_{S^\Delta_v}$, a mapping $\Phi^\Delta_v$ may be constructed satisfying $\mathrm{Leg}_v\circ\Phi^\Delta_{v}=\Id$ in some neighborhood of $\mathrm{Leg}_{v}(y_{S^\Delta_v})\in Y^*_{S^\Delta_v}$. If  $\mathrm{Leg}$ is an injective local diffeomorphism, it will be a diffeomorphism onto an open subbundle of $Y^*_{S^\Delta}$, with a unique inverse mapping $\Phi^\Delta_{v}$ defined on this open subbundle. 

Observe that $\mathrm{Leg}_{v}$ projects as $Id\colon Y_u\rightarrow Y_u$, for any vertex $u\in \bar S^\Delta_{v}$ with $u\neq \Delta_v$. Therefore any right inverse will be totally determined by some bundle morphism $\phi^\Delta\colon Y^*_{S^\Delta}\rightarrow Y_\Delta$.
\begin{define}
For any mapping $\phi^\Delta\colon Y^*_{S^\Delta}\rightarrow Y_\Delta$ consider the associated mapping $\Phi^\Delta\colon Y^*_{S^\Delta}\rightarrow Y_{S^\Delta}$ given by
	\begin{equation*}
	\begin{aligned}
	\Phi^\Delta_{v}(\delta_v y^*_{S^\Delta})&=\left(\phi^\Delta_{v}(\delta_v y^*_{S^\Delta}),y^*_{S^\Delta_v}\right)\in Y_{\Delta_v}\times \left(Y_{S^\Delta_v}/Y_{\Delta_v}\right)=Y_{S^\Delta_v},\quad \forall\, \delta_v y^*_{S^\Delta}\in Y^*_{S^\Delta}
	\end{aligned}
	\end{equation*}
(where $y^*_{S^\Delta_v}\in Y_{S^\Delta_v}/Y_{\Delta_v}$ represents the natural projection of $\delta_v y^*_{S^\Delta}\in Y^*_{S^\Delta_v}$ to $Y_{S^\Delta_v}/Y_{\Delta_v}$)

Given the discrete Lagrangian density $L$, we call integrator in the direction determined by the first order incremental flux $\Delta$ any (locally defined) mapping $\phi^\Delta\colon Y^*_{S^\Delta}\rightarrow Y_\Delta$ such that its associated mapping $\Phi^\Delta$ is a right-inverse of the mapping $\mathrm{Leg}$:
	\begin{equation*}
	\mathrm{Leg}_v(\phi^\Delta_v(\delta_vy^*_{S^\Delta}),y^*_{S^\Delta_v})=\delta_v y^*_{S^\Delta}\qquad \forall \delta_v y^*_{S^\Delta}\in Y^*_{S^\Delta_v}
	\end{equation*}
where $y^*_{S^\Delta_v}$ represents the projection on $Y^*_{S^\Delta_v}/Y_{\Delta_v}$ of $\delta_v y^*_{S^\Delta}\in Y^*_{S^\Delta_v}$.
\end{define}

Next result indicates how to explicitly recover the unknown component $y_{\Delta_v} \in Y_{\Delta_v}$ using the integrator, if we have the remaining components $[y_{S_v}]\in Y_{S_v}/Y_{\Delta_v}$, and the set of Euler-Lagrange equations $EL_v(y_{S_v})=0$, seen as implicit equations on this unknown.
\begin{theorem}\label{integradorteorema}
Let $\phi^\Delta$ be an integrator for the discrete Lagrangian density $L$ in the direction of the first order incremental flow $\Delta$.

Each class $[y_{S_v}]\in  Y_{S}/Y_{\Delta}$ contains a solution $y_{S_v}\in Y_S$ of Euler-Lagrange equations $EL_v(y_{S_v})=0$, explicitly given by $y_{S_v}=(y_{\Delta_v},[y_{S_v}])\in Y_S=Y_\Delta\times (Y_S/Y_\Delta)$, where $y_{\Delta_v}=\phi^\Delta_v\circ\mu_{v}([y_{S_v}])\in Y_{\Delta_v}$.

Moreover, in the case that $\mathrm{Leg}_{v}$ is injective, this point is the only element on $Y_S$ on the given class  $[y_{S_v}]$ that solves Euler-Lagrange equations $EL_v(y_{S_v})=0$.
\end{theorem}
\begin{proof}
Any element $y_{S_v}\in Y_S$ is determined by two components $y_{\Delta_v}\in Y_{\Delta_v}$ and $[y_{S_v}]\in Y_{S_v}/Y_{\Delta_v}$. Likewise, any element on $Y_{S^\Delta}$ is determined by a component on $Y_\Delta$ and another component on $Y_{S^\Delta}/Y_\Delta$. Consider the natural projection $\pi\colon Y_S/Y_\Delta \rightarrow Y_{S^\Delta}/Y_\Delta$. Following (\ref{ELdiferenca}), the associated Euler-Lagrange tensor  at $v$ may be written as:
	\begin{equation*}
	EL_v(y_{\Delta_v},[y_{S_v}])=\mathrm{Leg}_{v}\left(y_{\Delta_v},\pi([y_{S_v}])\right)-\mu_{v}([y_{S_v}])
	\end{equation*}
Using $y_{\Delta_v}=\phi^\Delta_v\circ\mu_{v}([y_{S_v}])\in Y_{\Delta_v}$ for a given element $[y_{S_v}]\in Y_S/Y_\Delta$, we get:
\begin{equation*}
\mathrm{Leg}_{v}(y_{\Delta_v},\pi([y_{S_v}]))=\mathrm{Leg}_{v}\left(\phi^\Delta_v\circ\mu_{v}([y_{S_v}]),\pi([y_{S_v}])\right)= \mu_{v}([y_{S_v}])\in Y^*_{S_{v,w}}
\end{equation*}
last equality is a consequence of being $\phi^\Delta$ an integrator, because $\mu_v([y_{S_v}])$ projects into $Y_{S^\Delta}/Y_\Delta$ as precisely $\pi([y_{S_v}])\in Y_{S^\Delta_v}/Y_{\Delta_v}$.

We conclude that $EL(y_{\Delta_v},[y_{S_v}])=0$ for the particular choice $y_{\Delta_v}\in Y_{\Delta_v}$ in the statement. In the case that $\mathrm{Leg}_{v}$ is injective, there cannot be two different solutions to the system $\mathrm{Leg}_{v}(y_{\Delta_v},\pi([y_{S_v}]))=\mu_{v,w}([y_{S_v}])$, thus concluding our proof. 
\qed
\end{proof}
\begin{remark}\label{remark4fim}
This result allows the introduction of discrete integration schemes from an initial condition determined on a saw-shaped band, leading to a section satisfying Euler-Lagrange equations, as described in \cite{CasiRodr12b,FadVol94}. This same formulation of Cauchy initial condition band was called ``the Cauchy problem on a zigzag'' in \cite{AdlBobSur04}, and used to derive a symplectic structure for a 2D discrete field theory in \cite{AdlBobSur04,FadVol94}. To be more specific, for the cubic or the CFK simplicial cellular complexes, we may choose the first order incremental flow $\Delta\colon v\mapsto \Delta_v=v+(1,\ldots,1)$ for $v\in\mathbb{Z}^n$, and if the Lagrangian density has a globally defined integrator $\phi^\Delta$ associated to this flow, whenever the configuration of a field is known on vertices determined by $k-n \leq x_1+x_2+\ldots+x_n\leq k+n-1$, the application of the integrator determines an extensions to a larger domain $k-n\leq x_1+\ldots+x_n\leq k+n$, so that discrete Euler-Lagrange equations are  satisfied at all vertices on $x_1+\ldots+x_n=k$. This mechanism can be iterated to extend the given initial conditions into a critical discrete configuration for all vertices on the semispace $k-n\leq x_1+\ldots+x_n$.
\end{remark}

\section{Covariant discretization of smooth variational problems}\label{section5}
All previous results represent a discrete counterpart of the classical calculus of variations of fields on smooth bundles. In order to relate the smooth and the discrete theory, we need to establish some correspondence between objects introduced in discrete variational problems and in continuous (smooth) ones, a correspondence that is classically developed in local coordinates, or for affine trivial bundles, but which is also possible in several other situations  (see, for example \cite{LeokShin12} for the formulation when the fiber is a Lie group and the base manifold is a discrete line, or \cite{Pennec13} for the computation of a weighed mean in Lie groups). 

Consider a smooth bundle $Y\rightarrow X$ and a smooth lagrangian density $\La\vol^X$, where $\La$ is a function on $J^1Y$ and $\vol^X$ a volume element on $X$. A standard method to discretize the Lagrangian density is to decompose the manifold $X$ into compact domains $K_\beta\subset X$ ($\beta$ is used as parameter to index all these domains). It is customary to take $X$ with affine structure and $K_\beta$ some convex hull of a finite family of nodes $x(v)$ (here $v$ is used as parameter to index all vertices on $K_\beta$, we may consider that $v$ is adherent to $\beta$ and denote $v\prec\beta$). The exact discrete Lagrangian associated to $\La\vol^X$ on $\beta$ is defined as a function:
    \begin{equation*}
    \begin{array}{rcl}
	L_\beta\colon \prod\limits_{v\prec\beta} Y_{x(v)} &\rightarrow &\mathbb{R}\\
	y^x_\beta=(y_v)_{v\prec\beta} & \mapsto & \mathrm{inf} \left\{ \mathbb{L}_{K_\beta}(y)\right\}_{y\in\Gamma_{y^x_\beta}(X,Y)}
	\end{array}
	\end{equation*}
where $\Gamma_{y^x_\beta}(X,Y)$ denotes the set of sections $y\colon X\rightarrow Y$ defined on $K_\beta\subset X$ such that $y(x(v))=y_v$ $\forall v\prec\beta$, and where $\mathbb{L}_{K_\beta}(y)=\int_{K_\beta} (j^1y)^*\La\cdot \vol^X$. Thus in the case that there exists a minimizing section $y\colon K_\beta\rightarrow X$ for the action $\mathbb{L}_{K_\beta}$, among those sections satisfying the boundary condition $y(x(v))=y_v$ $\forall v\prec\beta$, the exact discrete lagrangian $L_\beta$ on $(y_v)_{v\prec\beta}$ takes the same value as the smooth action $\mathbb{L}_{K_\beta}$ on the smooth section $y$ (which represents a common discretization procedure for smooth Lagrangians, for example in \cite{Demoures15,LeokShin12,LewMarsOrtiWest03}).

The theoretical advantage of this exact discrete lagrangian is that any section $y$ determines a point $y^x_\beta=\{y(x(v))\}_{v\prec\beta}\in \prod\limits_{v\prec\beta} Y_{x(v)}$ such that $L_\beta(y^x_\beta)\leq \mathbb{L}_{K_\beta}(y)$, and in the case that $y$ is a minimum for the action $\LL_{K_\beta}$, with respect to infinitesimal variations that vanish at $y(x(v))$, the resulting point $y^x_\beta$ is a minimum for $L_\beta$ on the finite-dimensional manifold $\prod\limits_{v\prec\beta} Y_{x(v)}$, for which $L_\beta(y^x_\beta)=\mathbb{L}_{K_\beta}(y)$ holds.

Certain concerns arise when dealing with exact discrete Lagrangians. Firstly, its construction demands a decomposition of the manifold $X$ into compact domains $K_\beta$. Secondly, the infima above (or the minimizing sections $y$ for each given boundary condition $(y_v)_{v\prec\beta}$) may be difficult to compute, and it is nontrivial to determine if the resulting expression $L_\beta$ is a differentiable function on the manifold $\prod\limits_{v\prec\beta} Y_{x(v)}$. A way to solve this situation is to work with some approximate discrete Lagrangian density, a family of functions $L_\beta\colon \prod_{v\prec\beta} Y_{x(v)} \rightarrow \mathbb{R}$ that are seen as approximations to the exact discrete Lagrangian. With these functions a first question is to determine the error associated to this approximation. A second question is how to choose some discrete Lagrangian that inherits all the symmetries of the original smooth Lagrangian.

We aim to obtain a discretization technique that preserves symmetries of the smooth lagrangian density. That is, given a fibered manifold and a smooth Lagrangian density, we aim to functorially derive a discrete bundle on some simplicial complex, together with a discrete Lagrangian density.

Our derivation of some (approximate) discrete Lagrangian density will be based on the choice of a quadrature rule on some simplicial domain $\Delta_\beta$.
\begin{define}\label{deltabeta}
Let $\beta=\{v_0,v_1,\ldots,v_n\}\in V_n$ be any n-dimensional abstract simplex. We call simplex $\Delta_\beta$ of barycentric coordinates associated to $\beta\in V_n$ the following:
	\begin{equation*}
	\Delta_\beta=\left\{(\lambda_v)_{v\prec\beta}\,\colon \lambda_v\geq 0,\, \sum_{v\prec\beta} \lambda_v=1\right\}\subset \prod_{v\prec\beta}\RR\simeq \RR^{n+1}
	\end{equation*}
We also call interior of the simplex $\Int(\Delta_\beta)$ the subset given by points $(\lambda_v)_{v\prec\beta}$ with $\lambda_v\neq 0$, for each $v\prec\beta$.
\end{define}
Each vertex $u\prec\beta$ can be identified with a point on $\Delta_\beta$, defined by $\lambda_v=0$ $\forall v\neq u$ (hence $\lambda_u=1$). Edges $\alpha=\{u,v\}\in X_1$ can be identified with a segment $\overline{uv}\subseteq \Delta_\beta$, namely the set of points with $\lambda_w=0$, $\forall w\neq u,v$ (equivalently, with $\lambda_u+\lambda_v=1$). Similar identifications are possible for any $k$-simplex adherent to $\beta$.

\begin{define}\label{dl}
Call $\dd\lambda$ one of the two affine volume elements on $\Delta_\beta$ such that $\int_{\Delta_\beta}\dd\lambda=1/n!$ (where the integral is taken with respect to the orientation defined by $\dd\lambda$ itself).
\end{define}

\begin{define}
We call quadrature rule $Q$ on $\Delta_\beta$ any linear functional:
	\begin{equation}\label{quadrule}
	h\in\mathrm{Map}(\Delta_\beta,\mathbb{R})\mapsto Q(h)=\frac{1}{n!}\sum_{i=1}^k c_i\cdot h(u_i)
	\end{equation}
determined by some choice of $k$ nodes $u_1,\ldots, u_k\in \Delta_\beta$ and weights $c_1,\ldots, c_k\in\mathbb{R}$

For any integrable function $h$, we call error of the quadrature rule on $h$ the expression:
	\begin{equation*}
	\left|Q(h)-\int_{\Delta_\beta}h\dd\lambda\right|
	\end{equation*}
where integration on $\Delta_\beta$ is done with the orientation induced by $\dd\lambda$.
\end{define}
Simple quadrature rules are, for example, those determined by a choice of a single node with weight 1, at some vertex $v\prec\beta$. These quadrature rules $Q_v$ have vanishing error when applied to constant functions $h$. A quadrature rule that has no error when applied to affine functions $h$ would be the quadrature rule $Q_{sym}$ where all vertices of $\Delta_\beta$ represent nodes with the same weight $1/(n+1)$. Another choice with the mentioned property is the midpoint quadrature rule $Q_{mid}$ associated to the barycenter of $\Delta_\beta$, with weight $1$.
	\begin{equation}\label{quadratures}
	Q_v(h)=\frac{1}{n!}\cdot 1\cdot h(v), \qquad Q_{sym}(h)=\frac{1}{n!}\cdot \sum_{v\prec\beta} \frac{1}{n+1}h(v),\qquad Q_{mid}(h)=\frac{1}{n!}\cdot 1\cdot h(\frac{1}{n+1},\ldots,\frac{1}{n+1})
	\end{equation}
The error associated to any of these quadrature rules can be bounded by different expressions, in terms of the range of $h$ and its directional derivatives.

In order to use these quadrature rules to approximate the exact discrete Lagrangian, the action functional $\mathbb{L}_{K_\beta}$ should be expressed as some integral on the simplex $\Delta_\beta$. It is on this integral where the quadrature rule can be introduced, leading to some numerical value that shall be the approximate value of the action functional.

To obtain an approximate expression of the integrand in the exact Lagrangian, in terms of the values $\{y(x(v))\}_{v\prec\beta}$ demands some interpolation procedure that recovers a smooth section from this discrete data.

Consider a fixed abstract simplicial complex $V$ and a fixed injective mapping $x\colon V_0\rightarrow X$. In this case, vertices in $V_0$ can be seen as {\em nodes} $x(v)\in X$. Moreover, the restriction of $Y$ to these nodes gives
	\begin{equation*}
	Y^x=x^*Y=\{(v,y)\in V_0\times Y\,\colon\, \pi_X(y)=x(v)\},\qquad \pi\colon (v,y)\in Y^x\mapsto v\in V_0
	\end{equation*}
a discrete bundle $\pi\colon Y^x\rightarrow V_0$ defined on the discrete space $V$  (see remark \ref{remark31a}). Any smooth section $y\colon X\rightarrow Y$ naturally induces, by restriction to the nodes, a discrete field $y^x\colon V_0\rightarrow Y^x$, defined as $y^x_v=y(x(v))$ for each $v\in V_0$. Thus $y^x$ can be seen as the node-evaluation of the section $y$ at the nodes $x(v)\in X$. 	Smooth fields $y$ on $Y$ determine discrete fields $y^x$ on $Y^x$. 
	\begin{equation*}
		\begin{matrix} \Gamma(X,Y) &\rightarrow & \Gamma(V_0,Y^x) \\
    y & \mapsto & y^x=y\circ x\end{matrix}
	\end{equation*}
We would also like to recover some sort of interpolation, reconstructing a smooth section from its values at the nodes of some simplex. This can be done in the general framework of Riemannian manifolds and symmetry groups of Riemannian isometries. Consider that $Y$ is endowed with some Riemannian metric.
	\begin{define}[adopted and adapted from \cite{Sander11}]\label{defineinterpolator}
	In the situation described so far, consider a smooth Riemannian structure on the manifold $Y$ and $\dist(\cdot, \cdot)\colon Y\times Y\rightarrow \mathbb{R}$ the induced distance metric. We call simplicial geodesic interpolator associated to the configuration $y^x_\beta\in Y^x_n$ (where $\beta\in V_n$) the mapping:
    \begin{equation*}
    \begin{array}{rrcl}
		\Upsilon_{y^x_\beta} \colon & \Delta_\beta & \rightarrow & Y \\
		 & (\lambda_v)_{v\prec\beta} & \mapsto & \mathop{\arg}\limits_{y\in Y}\min \sum\limits_{v\prec\beta}\lambda_v\cdot \left(\dist(y^x_v,y)\right)^2
		\end{array}
		\end{equation*}
with $\Delta_\beta$ as given in Definition \ref{deltabeta}.
	\end{define}
This interpolator, also known as weighed geometric mean or Riemannian mean, was studied in detail by K\"archer in \cite{Karcher} and subsequent works and deserved attention in more recent papers \cite{Deylen15,GroHardSan13,Moakher02,Sander11}. Its explicit expression for $Y=SO(3)$ and many other manifolds is known. We may remark that in the case that $Y$ is an Euclidean space this interpolator is the parameterization of a convex affine simplex by baricentric coordinates with respect to its vertices. 

Some lemmas follow now to enlighten the behavior of the interpolator
\begin{lemma}\label{sectionalcurvatures}
Assuming that the manifold $Y$ is complete, for any geodesic ball contained in $Y$, with radius $\rho$ and whose points have sectional curvatures smaller than $\left(\frac{\pi}{4\rho}\right)^2$ in any planar directions, and given any configuration $y^x_\beta\in Y^x_n$ whose components $y^x_v$ lie on this geodesic ball, the mapping $\Upsilon_{y^x_\beta}$ is well-defined (K\"archer, cited by Theorem 2.1 in \cite{Sander11}). Assuming further that the injectivity radius at all points $y^x_v$ is greater that $2\rho$, the mapping $\Upsilon_{y^x_\beta}$ is infinitely differentiable (see \cite{Sander11}).
\end{lemma}
\begin{proof}
See \cite{Sander11} and \cite{Karcher}.  
\end{proof}
\begin{lemma}\label{splits}
Simplicial geodesic interpolation on the product of Riemannian manifolds $Y=X\times Q$ splits as simplicial geodesic interpolation on the first component, and simplicial geodesic interpolation on the second one.
\end{lemma}
\begin{proof}
Geodesic distance on the product manifold, with the product Riemannian structure is given by $dist^2((x_v,q_v),(x,q))=dist^2(x_v,x)+dist^2(q_v,q)$.

Computation of geodesic simplicial interpolation demands now the computation of:
	\begin{equation*}
	 \mathop{\arg}\limits_{(x,q)}\min \sum\limits_{v\prec\beta}\lambda_v\cdot \left(\dist((x_v,q_v),(x,q))\right)^2=
	 \mathop{\arg}\limits_{(x,q)}\min \sum\limits_{v\prec\beta}\lambda_v\cdot \left(\dist(x_v,x)\right)^2+\sum\limits_{v\prec\beta}\lambda_v\cdot \left(\dist(q_v,q)\right)^2
	\end{equation*}
It is trivial that $\mathop{\arg}\limits_{(x,q)}\min f(x)+g(q)=(\mathop{\arg}\limits_{x}\min f(x),\mathop{\arg}\limits_{q}\min g(q))$, concluding our proof.\qed
\end{proof}
A reason to call $\Upsilon$ a geodesic interpolator is the following Lemma:
\begin{lemma}
Consider two vertices $u,v$ belonging to some common $n$-cell $\beta\in V_n$, and the segment $I_{uv}\subset\Delta_\beta$, determined by points $\lambda\in \Delta_\beta\subset\RR^{n+1}$ such that $\lambda_u+\lambda_v=1$ (hence all remaining components vanish). For any configuration $y^x_\beta\in Y^x_\beta$, the restriction $\left.\Upsilon _{y^x_\beta}\right|_{I_{uv}}$ is the unique minimal geodesic parameterized by constant arc-length on $Y$ joining $y^x_u$ (for $\lambda_u=1,\lambda_v=0$) to $y^x_v$ (for $\lambda_u=0,\lambda_v=1$)
\end{lemma}
\begin{proof}
See \cite{Sander11}. 
\end{proof}
As we shall see, this lemma leads to an easy determination of directional derivatives of $\Upsilon_{y^x_\beta}$, at any vertex and along any edge of $\Delta_\beta$, if we know how to determine geodesics on $Y$.

Moreover, if $\beta=\{v_0,v_1,\ldots, v_n\}$, $\bar\beta=\{\bar v_0,v_1,\ldots, v_n\}$ are $n$-cells with a common facet $\alpha=\{v_1,\ldots, v_n\}\in V_{n-1}$, then $\Delta_\alpha$ can be seen as contained in $\Delta_\beta$ with equation $\lambda_{v_0}=0$, or contained in $\Delta_{\bar \beta}$, with equation $\lambda_{\bar v_0}=0$. In this situation:
\begin{lemma}\label{lemma3}
Let $\beta,\bar{\beta}\in V_n$ be two $n$-cells with a common adherent facet $\alpha\in V_{n-1}$. For any configurations $y^x_{\beta}$ and $y^x_{\bar\beta}$ that have a common restriction $y^x_\alpha\in Y^x_\alpha$, the interpolators $\Upsilon_{y^x_{\beta}}$, and $\Upsilon_{y^x_{\bar\beta}}$ coincide, on all points of this facet $\Delta_{\alpha}$.
\end{lemma}
\begin{proof}
See \cite{Sander11}
\end{proof}
In order to derive a locally defined interpolating section associated to discrete data $y^x_\beta\in Y^x_n$ we still need the following particular situation:
\begin{define}\label{suited}
We say that $y^x_\beta\in Y^x_\beta$ is suited for simplicial geodesic interpolation if $\Upsilon_{y^x_\beta}\colon \Delta_\beta\rightarrow Y$ is well-defined, and its projection $\Upsilon^X_{y^x_\beta}=\pi_X\circ \Upsilon_{y^x_\beta}\colon \Delta_\beta\rightarrow X$ to the base manifold $X$ determines an injective immersion.
\end{define}
In the previous definition, by immersion $i\colon \Delta_\beta\rightarrow X$ we mean any mapping that restricted to $\Delta_\alpha$ (where $\alpha\subset \beta$ is any nonempty subset of vertices), has at each interior point $\lambda \in \Int(\Delta_\alpha)$ an injective differential $\dd_\lambda \left(\left.i\right|_{\Int(\Delta_\alpha)}\right)$ (the rank equals de dimension of the subsimplex). The image of the immersion is then a subset $K_{y^x_\beta}\subset X$ diffeomorphic (as a manifold with boundary) to the $n$-dimensional simplex, a diffemorphism determined by $\Upsilon^X_{y^x_\beta}$. This allows to give an interpretation of the interpolator as a local section $X\rightarrow Y$, as follows:
\begin{define}\label{definterpolating}
At any point $y^x_\beta\in Y^x_\beta$ suited for simplicial geodesic interpolation, we call interpolating section associated to $y^x_\beta\in Y^x_n$ the unique smooth mapping $y\colon K_{y^x_\beta}\subset X\rightarrow Y$ defined on $K_{y^x_\beta}=\Upsilon^X_{y^x_\beta}(\Delta_\beta)$ and satisfying $y\circ \Upsilon^X_{y^x_\beta}=\Upsilon_{y^x_\beta}$.
\end{define}
In particular for any $y^x_\beta=(y^x_v)_{v\prec\beta}$ suited for interpolation, as the vertices of $\Delta_\beta$ are transformed by $\Upsilon_{y^x_\beta}$ into the set of points $(y^x_v)_{v\prec\beta}$ on $Y$, whose projections are precisely the nodes $\{x(v)\}_{v\prec\beta}$ on $X$, we conclude that the interpolating section $y$ satisfies $y(x(v))=y^x_v$, for all the $n+1$ nodes associated to the vertices $v\prec\beta$.

For any such configuration $(y^x_v)_{v\prec\beta}\in Y^x_n$ the interpolator generates then a domain with piecewise smooth boundary $K_{y^x_\beta}$ on $X$, diffeomorphic to a simplex, whose vertices are the nodes $x(v)$, and also generates a local section of the bundle $Y\rightarrow X$, defined on this domain, that coincides on these nodes with the given elements $y^x_v\in Y_{x(v)}$.
\begin{lemma}
If $\beta$, $\bar\beta$ share a common facet $\alpha$ and $y^x_{\beta}$ and $y^x_{\bar\beta}$ are suited for simplicial geodesic interpolation and coincide on the vertices of the facet $\alpha$, its corresponding interpolating sections coincide on $\Upsilon^X_{y^x_\beta}(\Delta_\alpha)=\Upsilon^X_{y^x_{\bar\beta}}(\Delta_\alpha)\subset X$.
\end{lemma}
\begin{proof}
Easy consequence of Lemma \ref{lemma3}
\end{proof}
Any morphism $\varphi\colon Y\rightarrow Y$ fibered over $\varphi_X\colon X\rightarrow X$, induces a morphism $\varphi_n\colon Y^x_n\rightarrow Y^{\hat{x}}_n$ of discrete bundles over $V_n$ (where $\hat{x}=\varphi_X\circ x$), defined by $\varphi_n(y^x_v)_{v\prec\beta}=(\varphi(y^x_v))_{v\prec\beta}$. 
\begin{lemma}
In the case that $\varphi\colon Y\rightarrow Y$ is a Riemannian isometry fibered on $\varphi_X\colon X\rightarrow X$, for any configuration $y^x_\beta$ suited for simplicial geodesic interpolation, also $\varphi_n(y^x_\beta)$ is suited for simplicial geodesic interpolation and there holds
	\begin{equation}\label{covarupsilon}
	\begin{aligned}
	&\Upsilon_{\varphi_n(y^x_\beta)}=\varphi\circ\Upsilon_{y^x_\beta}\\
	&\Upsilon^X_{\varphi_n(y^x_\beta)}=\varphi_X\circ\Upsilon^X_{y^x_\beta}\\
	&\bar{y}=\varphi\circ y\circ \varphi_X^{-1}
	\end{aligned}
	\end{equation}
where $y$ and $\bar{y}$ denote the interpolating sections corresponding to the configurations $y^x_\beta$, $\varphi_n(y^x_\beta)$ respectively.\end{lemma}
\begin{proof}
From definition \ref{defineinterpolator} it becomes clear that $\Upsilon_{\varphi_n(y^x_\beta)}=\varphi\circ\Upsilon_{y^x_\beta}$. The remaining statements are an easy consequence of this one. \qed
\end{proof}

Consider now a smooth Lagrangian density $\mathcal{L}\vol^X$ defined on $J^1Y$. Consider the interpolating section $y\colon v\in K_{y^x_\beta}\mapsto y(v)\in Y$ determined by the simplicial geodesic interpolator $\Upsilon$ associated to $y^x_\beta\in Y^x_n$. Is there a way to approximate the value $\mathbb{L}_{K_{y^x_\beta}}(y)=\int_{K_{y^x_\beta}}\mathcal{L}(j^1y)\vol^X$ without actually computing the interpolator $\Upsilon$? We may pull-back the integral to $\Delta_\beta$ using $\Upsilon^X\colon\Delta_\beta\rightarrow K_{y^x_\beta}$, and write
	\begin{equation}\label{actionsimplex}
	\mathbb{L}_{K_{y^x_\beta}}(y)=\int_{z\in K_{y^x_\beta}}\mathcal{L}(j^1_z y)\cdot \vol^X=\int_{\lambda\in \Delta_\beta} \mathcal{L}\left(j^1_{\Upsilon^X(\lambda)}y\right)\cdot (\Upsilon^X)^*\vol^X
	\end{equation}
\begin{define}
For any configuration suited for simplicial geodesic interpolation $y^x_\beta\in Y^x_n$, we call Jacobian function associated to this configuration the positive function $Jac_{y^x_\beta}\colon \Delta_\beta\rightarrow \mathbb{R}^+$ determined by:
	\begin{equation*}
	(\Upsilon^X)^*\vol^X=Jac_{y^x_\beta}(\lambda)\dd\lambda
	\end{equation*}
where $\vol^X$ is our chosen volume $n$-form on  $X$, $\Upsilon^X=\pi_X\circ\Upsilon_{y^x_\beta}$ is determined by the associated simplicial geodesic interpolator, and $\dd\lambda$ is the volume element on $\Delta_\beta$ determined according to definition \ref{dl}, whose orientation coincides with $(\Upsilon^X)^*\vol^X$.
\end{define}
Thus following (\ref{actionsimplex}), for any section $y\colon K_{y^x_\beta}\rightarrow Y$ there holds:
	\begin{equation*}
	\mathbb{L}_{K_{y^x_\beta}}(y)=\int_{\lambda\in \Delta_\beta} h(\lambda)\dd\lambda\text{ , where } h(\lambda)=\La(j^1_{\Upsilon^X(\lambda)}y)\cdot Jac_{y^x_\beta}(\lambda)
	\end{equation*}
\begin{define}\label{inducedsimpgeod}
We call discrete Lagrangian density $L^x$ induced on $Y_n^x$ by the Lagrangian density $\La\cdot\vol^X$ on $J^1Y$, using simplicial geodesic interpolation and a quadrature rule $Q$ on $\Delta_\beta$, the mapping $L^x\colon Y_n^x\rightarrow \mathbb{R}$ defined by:
	\begin{equation}\label{f510}
	L^x(y^x_\beta)=Q(h_{y^x_\beta}),\quad \text{ where }h_{y^x_\beta}(\lambda)=\La(j^1_{\Upsilon^X(\lambda)}y)\cdot Jac_{y^x_\beta}(\lambda)
	\end{equation}
Here $h_{y^x_\beta}(\lambda)\colon \Delta_\beta\rightarrow\mathbb{R}$ is defined by means of the simplicial interpolator $\Upsilon$ associated to $y^x_\beta$, its projection $\Upsilon^X_{y^x_\beta}$, the associated interpolating section $y\colon K_{y^x_\beta}\rightarrow Y$ and the Jacobian function $Jac_{y^x_\beta}$ 
\end{define}
The values of the function $h_{y^x_\beta}$ are easy to compute, at the vertices of the simplex $\Delta_\beta$:
\begin{lemma}\label{lem511}
Consider any configuration $y^x_\beta\in Y^x_n$ suited for simplicial geodesic interpolation. Let $\Upsilon\colon\Delta_\beta\rightarrow Y$ be the associated simplicial geodesic interpolator and $y\colon K_{y^x_\beta}\rightarrow Y$ be the associated interpolating section.

For any vertices $u,v\prec\beta$, denote by $t_{uv}$ the tangent vector $\left(\frac{\dd}{\dd s}\right)_{s=0}\gamma(s)$ of the minimal geodesic $\gamma(s)$ joining $y^x_u\in Y$ at $s=0$ with $y^x_v\in Y$ at $s=1$ parameterized with constant speed. Denote by $t^X_{uv}$ the corresponding projection into $X$.

If $\lambda\in \Delta_\beta$ is a vertex (that is, $\lambda_u=1$ for some $u\prec\beta$), then 
	\begin{equation*}
	\begin{aligned}
	& \Upsilon^X(\lambda)=x(u),\qquad \Upsilon(\lambda)=y(x(u))=y^x_u \\
	&j^1_{\Upsilon^X(\lambda)}y=\phi_u\colon t^X_{uv}\in T_{x(u)}X\rightarrow t_{uv}\in T_{y^x_u}Y\quad \forall v\prec\beta\\
	&Jac_{y^x_\beta}(\lambda)=\left|\vol^X_{x(u)}(t^X_{uv_1},\ldots,t^X_{uv_n})\right|\in\mathbb{R}
	\end{aligned}
	\end{equation*}
where $v_1,\ldots, v_n$ represents the set of vertices adherent to $\beta\in V_n$, excluding $u$.
\end{lemma}
\begin{proof}
Following $y\circ\Upsilon^X=\Upsilon$, we know $y(\Upsilon^X(\lambda))=\Upsilon(\lambda)$ and $\dd_{\Upsilon^X(\lambda)}y=(\dd_\lambda\Upsilon)\circ(\dd_\lambda\Upsilon^X)^{-1}$ at any point $\lambda\in\Delta_\beta$. The values are now relatively easy to compute at the vertices of $\Delta_\beta$. Namely, $\Upsilon$ maps these vertices into $y^x_v$ and maps the edges $I_{uv}\subset\Delta_\beta$ joining two vertices $u,v\in\beta$ into constant-speed parameterized minimal geodesics joining $y^x_u$ with $y^x_v$.

Enumerate the vertices as $u=v_0,\ldots, v_n$. With this ordering, we may identify $\Delta_\beta$ with the domain $\{(\lambda_0,\ldots,\lambda_n)\in\RR^{n+1}\,\colon \, \sum \lambda_i=1,\, \lambda_i\geq 0\}$. Call $\partial_i$ the vector on $\RR^{n+1}$ tangential to the $i$-direction. The tangential vector associated to the edge $\gamma(s)=sv_j+(1-s)u$ joining $u$ at $s=0$ to $v_j$ at $s=1$ is precisely $\partial_j-\partial_0$. This edge is transformed by $\Upsilon$ into a geodesic, therefore $(\dd_{u}\Upsilon)(\partial_j-\partial_0)=t_{uv_j}$, $(\dd_{u}\Upsilon^X)(\partial_j-\partial_0)=t^X_{uv_j}$ and $(\dd_{u}\Upsilon)\circ(\dd_{u}\Upsilon^X)^{-1}=\phi_{u}$, as defined in the statement.

Computing $(\Upsilon^X)^*\vol^X$ at the point $u=v_0\in\Delta_\beta$ leads to:
	\begin{equation*}
	((\Upsilon^X)^*\vol^X)_{u}(\partial_1-\partial_0,\ldots, \partial_n-\partial_0)=\vol^X_{x(u)}(t^X_{uv_1},\ldots, t^X_{uv_n})
	\end{equation*}
on the other hand $\pm\dd\lambda=\dd\lambda_1\wedge\ldots\wedge\dd\lambda_n$ is the affine differential form whose integration on $\Delta_\beta$ is $\frac{1}{n!}$. This differential form takes value 1 when applied to vectors $\partial_1-\partial_0,\ldots,\partial_n-\partial_0$ at the point $v_0$, therefore if $\lambda\in\Delta_\beta$ represents the vertex $u=v_0$, we may write:
	\begin{equation*}
	\pm Jac_{y^x_\beta}(\lambda)=((\Upsilon^X)^*\vol^X)_{u}(\partial_1-\partial_0,\ldots,\partial_n-\partial_0)=\vol^X_{x(u)}(t^X_{uv_1},\ldots,t^X_{uv_n})
	\end{equation*}
which completes the proof.\qed
\end{proof}
Computing the components in formula (\ref{f510}) for a quadrature rule with nodes at vertices can be reduced to the determination of geodesics, and application of lemma \ref{lem511}. When the quadrature rule uses as nodes arbitrary points $u\in \Delta_\beta$ (not necessarily a vertex) an explicit computation of all elements in (\ref{f510}) would be possible, resulting in an expression depending only on $(y^x_v)_{v\prec\beta}$, but obtaining an analytic expression might imply a much harder computational effort, depending on the particular Riemannianian structure of the manifold $Y$. It is reasonable then to use quadrature rules whose nodes are all at vertices of $\Delta_\beta$.

\begin{remark}
Smooth sections $y\in\Gamma(X,Y)$ determine discrete configurations $y^x_\beta \in Y^x_\beta$. If this one is suited for geodesic simplicial interpolation, it determines a new (locally defined) section $y^{int}\in \Gamma(K_{y^x_\beta},Y)$. Both sections may be compared, and the error in substituting one with the other, might be bounded, using the results in \cite{Deylen15,GroHardSan13}. We shall not explore this aspect any further. However, error bounds for the values of $y^{int}$ and its derivatives as approximations to $y$ and its derivatives, are relevant to derive new error bounds between the value of the smooth action functional on $y$ and the value of the approximate discrete Lagrangian density on $y^x_\beta$, a bound that also depends on the particular quadrature rule used (see \cite{Guessab} for some error bounds for different quadrature rules on simplices).
\end{remark}

We may now study the behavior of the discrete Lagrangian density associated to $\La\vol^X$, with respect to the action of some symmetry group.

\begin{prop}
Consider a bundle $Y\rightarrow X$, and a morphism $\varphi\colon Y\rightarrow Y$ fibered over $\varphi_X\colon X\rightarrow X$, isometry for the Riemannian metric and whose extension $j^1\varphi\colon J^1Y\rightarrow J^1Y$ is symmetry for the Lagrangian density:
	\begin{equation*}
	\varphi\colon Y\rightarrow Y,\quad \begin{array}{rcl}
		j^1\varphi\colon J^1Y & \rightarrow & J^1Y\\
			j^1_xy & \mapsto & j^1_{\varphi_X(x)} (\varphi\circ y\circ\varphi_X^{-1})
    \end{array},\quad \text{ verifying }\left\{\begin{aligned} &(j^1\varphi)^*(\mathcal{L}\vol^X)=\mathcal{L}\vol^X \\ &\dd_y\varphi\colon T_yY\rightarrow T_{\varphi(y)}Y\text{ isometry }\forall y\in Y\end{aligned}\right.     
	\end{equation*}
then for any $x\colon V_0\rightarrow X$, for $\hat{x}=\varphi_X\circ x$, and for any quadrature rule $Q$ on the simplex $\Delta_\beta$, the induced discrete densities $L^x$ on $Y^x_n$ and $L^{\hat x}$ on $Y^{\hat x}_n$ (both of them locally defined at points suited for simplicial geodesic interpolation, following definition \ref{suited}) are related by the induced mapping $\varphi_n\colon Y^x_n\rightarrow Y^{\hat x}_n$ as follows:
	\begin{equation*}
	L^{\hat x}(\varphi_n(y^x_\beta))=L^x(y^x_\beta)
	\end{equation*}
\end{prop}
\begin{proof}
If $y$ is the interpolating section determined by $y^x_\beta$, then from (\ref{covarupsilon}) follows that $\varphi\circ y\circ\varphi_X^{-1}$ is the interpolating section determined by $\varphi_n(y^x_\beta)$. On the other hand, at any point $x\in X$ holds $$\La\left(j^1_{\varphi_X(x)}(\varphi\circ y\circ\varphi_X^{-1})\right)\cdot \varphi_X^*\vol^X_{\varphi_X(x)}=\La(j^1_xy)\cdot \vol^X_x$$ Applying this property at the point $x=\Upsilon^X_{y^x_\beta}(\lambda)$ we get for the function (\ref{f510}):
	\begin{equation*}
	\begin{aligned}
	h_{\varphi_n(y^x_\beta)}&(\lambda)\cdot\dd\lambda=\La(j^1_{\varphi_X\circ\Upsilon^X(\lambda)}(\varphi\circ y\circ\varphi_X^{-1}))\cdot (\varphi_X\circ\Upsilon^X)^*\vol^X_{\varphi_X\circ\Upsilon^X(\lambda)}=\\
	&= (\Upsilon^X)^*\left(\La(j^1_{\varphi_X\circ\Upsilon^X(\lambda)}(\varphi\circ y\circ\varphi_X^{-1}))\cdot \varphi_X^*\vol^X_{\varphi_X\circ\Upsilon^X(\lambda)}\right)=\\
	&=(\Upsilon^X)^*\left(\La(j^1_{\Upsilon^B(\lambda)}y)\cdot \vol^X_{\Upsilon^X(\lambda)}\right)= 
	\La(j^1_{\Upsilon^B(\lambda)}y)\cdot (\Upsilon^X)^*\vol^X_{\Upsilon^X(\lambda)}
	=h_{y^x_\beta}(\lambda)\cdot\dd\lambda
	\end{aligned}
	\end{equation*}
Therefore the functions used in definition \ref{inducedsimpgeod} associated to $y^x_\beta$ and $\varphi_n(y^x_\beta)$ coincide: $h_{\varphi_n(y^x_\beta)}=h_{y^x_\beta}$. As both functions are the same, the discrete Lagrangian density $L^{\hat x}$ applied to the point $\varphi_n(y^x_\beta)$ leads to the same result as the discrete Lagrangian density $L^x$, applied to the point $y^x_\beta$ (in both cases, the result is obtained by some given quadrature rule on $\Delta_\beta$, applied to the same function).\qed
\end{proof}

\begin{corollary}
Let $Y\rightarrow X$ be a smooth bundle, $\mathcal{L}\vol^X$ a smooth lagrangian density on $J^1Y$, $V$ an abstract simplicial complex, $x\colon V_0\rightarrow X$ an injective mapping determining nodes on $X$ for each vertex $v\in V_0$. Let $Y^x$ be the induced discrete bundle and $L^x$ the induced discrete lagrangian density on $Y^x$, determined by definition \ref{inducedsimpgeod}.

If $\varphi_t\colon Y\rightarrow Y$ is a 1-parameter group of isometries on $Y$, fibered over the identity mapping on $X$, symmetry for $\mathcal{L}\vol^X$ and with infinitesimal generator $D\in\X(Y)$, then $D$ is a vertical vector field, and its restriction $D^x$ to $Y^x\subset Y$ is an infinitesimal symmetry (in the sense of definition \ref{definesimetria}) for the discrete lagrangian density $L^x$.
\end{corollary}
Following these results we observe that any smooth lagrangian density determines a discrete lagrangian density, and that isometries respecting the original lagrangian density are also symmetries for the discrete one. The determination of the discrete Lagrangian density is easy if the choice of quadrature rule has nodes at the vertices of a simplicial domain and if we know how to construct minimal geodesics joining two neighboring points on $Y$. In particular, our discretization mechanism is easily applied when $Y$ is $\RR^n$, the sphere $S^n$, the space $SO(3)$, or any product of these Riemannian manifolds, leading to discrete Lagrangian densities that are invariant for any Riemannian transformation that is a symmetry of the original Lagrangian density.

\begin{remark}
Observe that, for any given smooth lagrangian, this interpolation procedure generates a discrete Lagrangian $L^x$ on $n$-simplices of the CFK complex $V_0$, which induces a discrete Lagrangian $\bar{L}^x$ on $n$-cells of the cubic cellular complex $\bar{V}_n$, determined by formula (\ref{simplexcubic}). All automorphisms of $Y$ respecting both the Riemannian structure and the smooth Lagrangian density lead to symmetries for the discrete Lagrangian density defined on the CFK complex and also for the induced lagrangian defined on the cubic complex. This leads to corresponding variational principles and conservation laws in the sense already studied in \cite{CasiRodr12}.
\end{remark}
\subsection*{The case of affine bundles}
If the bundle $Y\rightarrow X$ is affine, with affine projection mapping, the situation becomes more familiar. Any affine space $Y$ may be endowed with an arbitrary euclidean metric. The geodesics in this case are straight lines, not depending on the particular choice of euclidean structure. The simplicial geodesic interpolator associated to any point $y^x_\beta$ is simply the affine interpolator, taking any $(\lambda_v)_{v\prec\beta}\in\Delta_\beta$ to the point $y\in Y$ with barycentric coordinates $\lambda_v$, with respect to the referential $(y^x_v)_{v\prec\beta}$:
    \begin{equation*}
    \begin{array}{rrcl}
		\Upsilon_{y^x_\beta} \colon & \Delta_\beta & \rightarrow & Y \\
		 & (\lambda_v)_{v\prec\beta} & \mapsto & \sum\limits_{v\prec\beta} \lambda_v\cdot y^x_v
		\end{array}
		\end{equation*}
where the sum makes sense as a weighed mean of points on any affine space, considering that the total weight is $\sum_{v\prec\beta}\lambda_v=1$.

As the projection $\pi_X\colon Y\rightarrow X$ is affine, the projected interpolator $\Upsilon^X_{y^x_\beta}=\pi_X\circ\Upsilon_{y^x_\beta}\colon \Delta_\beta\rightarrow X$ is an affine mapping, explicitly given by $(\lambda_v)_{v\prec\beta} \mapsto \sum\limits_{v\prec\beta} \lambda_v\cdot x(v)$, depending only on the nodes $\{x(v)\}_{v\prec\beta}$, and the associated domain $K_{y^x_\beta}$ is the convex hull $K^x_\beta$ of these nodes $\{x(v)\}_{v\prec\beta}$
	\begin{equation*}
		\Upsilon^X_{y^x_\beta}(\Delta_\beta)=K^x_\beta=	\left\{\sum_{v\prec\beta}\lambda_v\cdot x(v)\colon\, (\lambda_v)_{v\prec\beta}\in \Delta_\beta\right\}\subset X
	\end{equation*}
If the nodes $\{x(v)\}_{v\prec\beta}$ are in general position on $X$, there exists a well-defined affine inverse $(\Upsilon_{y^x_\beta}^X)^{-1}\colon K^x_\beta\rightarrow \Delta_\beta$, transforming any point of the convex hull into its barycentric coordinates with respect of the vertices of $K^x_\beta$. In this case any configuration $y^x_\beta\in Y^x_\beta$ is suited for simplicial geodesic interpolation (in our case, affine interpolation).

Therefore when $Y$ is affine the induced discrete Lagrangian doesn't depend on the particular choice of euclidean structure, and all notions introduced in section \ref{section5} turn out to be the natural ones in the affine setting determined by affine interpolation. Any affine transformation is an isometry for some appropriate euclidean structure, therefore our results also tell us that any affine transformation that is a symmetry for the smooth lagrangian density will also be a symmetry for the induced discrete lagrangian induced by affine interpolation and any choice of quadrature rule.

Moreover, when $\pi\colon Y\rightarrow X$ is affine, its jet bundle $J^1Y$ can be identified as a product manifold $X\times\mathrm{Aff}_X(X,Y)$ of the base manifold $X$ and the space $\mathrm{Aff}_X(X,Y)$ of affine sections $y\colon X\rightarrow Y$ of the fibred affine space $\pi_X\colon Y\rightarrow X$. Any configuration $y^x_\beta$ suited for interpolation determines an affine interpolating section $y$, thus leading to a mapping:
	\begin{equation*}
	y^x_\beta=(y^x_v)_{v\prec\beta}\in Y^x_n\mapsto j^1_{\overline{x}(\beta)}y\simeq (\overline{x}(\beta),y)\in X\times \mathrm{Aff}_X(X,Y)=J^1Y
	\end{equation*}
where the affine mapping $y$ is univocally determined by the condition $y(x(v))=y^x_v$, $\forall v\prec\beta$, and the point $\overline{x}(\beta)$ is the barycenter of all nodes $\{x(v)\}_{v\prec\beta}$. This barycenter doesn't depend on the particular configuration $y^x_\beta$, but just on $\beta\in V_n$ and the immersion $x\colon V_0\rightarrow X$. If we denote by $\overline{x}\colon V_n\rightarrow X$ the mapping taking any $n$-cell into the barycenter $\overline{x}(\beta)$ of its associated nodes $\{x(v)\}_{v\prec\beta}$, we get a mapping:
	\begin{equation*}
	\begin{array}{rcl}
	Bary\quad\colon\quad Y^x_n& \rightarrow &\overline{x}^*J^1Y\simeq V_n\times \mathrm{Aff}_X(X,Y)\\
	(y^x_v)_{v\prec\beta}&\mapsto & (\bar{x}(\beta),y)\simeq(\beta,y)\text{ such that }y(x(v))=y^x_v,\,\forall v\prec\beta,\, y\text{ affine}
	\end{array}
	\end{equation*}
defined on the whole fiber $Y^x_\beta$ if the nodes $\{x(v)\}_{v\prec\beta}$ are in general position.
\begin{prop}
Consider a fibred affine coordinate system $x^1,\ldots,x^n,y^1,\ldots, y^m$ for the affine bundle $\pi\colon Y\rightarrow X$ and the induced coordinate system $(x^j, y^k,\partial_j y^k)$ on $J^1Y$, defined by:
	\begin{equation*}
	x^j(j^1_xy)=x^j(x),\quad y^k(j^1_xy)=y^k(y(x)),\quad \partial_jy^k(j^1_xy)=\frac{\partial y^k\circ y}{\partial x^j}(x)\qquad \forall y\in\Gamma(X,Y),\forall x\in X
	\end{equation*}
Fix an abstract simplicial complex $V$. Fix an injective mapping $x\colon V_0\rightarrow X$. Consider an ordering $(v_i)_{i=0\ldots n}$ for the adherent vertices $v\prec\beta$ of some $n$-cell $\beta\in V_n$, and the naturally induced coordinate system $(y^k_i)$ on $Y^x_\beta=\prod_i Y_{x(v_i)}$, defined by $y^k_i(y^x_\beta)=y^k (y^x_{v_i})$. 

Denote $x^0$ the constant function 1, and $x^j_i(\beta)=x^j(x(v_i))$, for the chosen ordering. Denote $\bar x^j(\beta)=\frac{1}{n+1}\sum_i x^j_i(\beta)$

A configuration $y^x_\beta\in Y^x_n$ is suited for simplicial geodesic interpolation if and only if $(x^j_i(\beta))$ is non-singular. In this case the mapping $Bary\colon Y^x_n\rightarrow \bar{x}^*J^1Y\subset J^1Y$ may be given in local coordinates as:
	\begin{equation}\label{partialsimplex2}
	x^j=\bar{x}^j(\beta),\qquad \left[\begin{matrix} y^k \\ \partial_1 y^k \\ \vdots \\ \partial_n y^k\end{matrix}\right]= 
	\left[\begin{matrix} 1 & \bar{x}^j(\beta) \\ 0 & \Id_n\end{matrix}\right]\cdot 
	\left[ x_i^j(\beta)\right]^{-1}\cdot \left[\begin{matrix} y^k_i\end{matrix}\right]\qquad \forall k=1\ldots m
	\end{equation}
and represents an isomorphism on each fiber.

For any fixed affine volume element $\vol^X$, the associated Jacobian function $Jac_{y^x_\beta}(\lambda)$ is a constant $Jac_{y^x_\beta}$, depending only on the nodes.
	\begin{equation*}
	\vol^X=c\cdot \dd x^1\wedge\ldots\wedge \dd x^n\Rightarrow Jac_{y^x_\beta}=c\cdot \left| \det(x^j_i)\right|
	\end{equation*}
\end{prop}
\begin{proof}
Denote $\Upsilon$ the simplicial interpolator for any given configuration $y^x_\beta$, and denote $\Upsilon^X$ the projected interpolator to the base manifold $X$. As $x^0_i=1$, substracting the $0$-row to the remaining ones we observe that the given matrix $(x^j_i)$ is invertible if and only if $x(v_i)-x(v_0)$ are linearly independent vectors. That is, if and only if all nodes $x(v_i)$ are in general position, which implies the existence of an affine inverse for the mapping $\Upsilon^X$. Conversely, when nodes $x(v_i)$ are not in general position, the associated mapping $\Upsilon^X$ doesn't span the whole affine space $X$, and the inverse of $\Upsilon^X$ does not exist.

Any affine mapping is totally determined by its 1-jet at any point. To prove that equations (\ref{partialsimplex2}) represent $Bary(y^x_\beta)$ it suffices to prove that the components $x^j=\frac1{n+1}\sum_i x^j_i$ represent the barycenter $\bar{x}(\beta)$ (which is well known, the barycenter has as coordinates the mean value of the coordinates of the given nodes), and that the affine mapping determined by (\ref{partialsimplex2}) coincides with the interpolating section, characterized by $y(x(v_i))=y^x_{v_i}$. 

Following Taylor's formula, the affine mapping $y(x)\in\mathrm{Aff}_X(X,Y)$ corresponding to (\ref{partialsimplex2}) is the one given by equations:
	\begin{equation*}
	y^k(x)=\left[ 1\, (x^1-\bar x^1(\beta))\, \ldots\, (x^n-\bar x^n(\beta)) \right]\cdot \left[\begin{matrix} y^k \\ \partial_1 y^k \\ \vdots \\ \partial_n y^k\end{matrix}\right]=\left[ 1\, x^1\, \ldots\, x^n \right]\cdot\cdot \left[ x_i^j(\beta)\right]^{-1}\cdot \left[\begin{matrix} y^k_i\end{matrix}\right]
	\end{equation*}
As $1=x^0$, it turns clear that the affine mapping $y(x)$ determined by (\ref{partialsimplex2}) has, indeed, the property $y(x(v_i))=y^x(v_i)$.

Equations (\ref{partialsimplex2}) are invertible on the fiber of each $\beta\in X_n$, because the used matrices are invertible.

Both $\dd\lambda$ and $(\Upsilon^X)^*\vol^X$ are affine volume elements, and therefore differ by a constant factor. Hence for any fixed $y^x_\beta$, the associated jacobian function $Jac_{y^x_\beta}(\lambda)$ is a constant $Jac_{y^x_\beta}$, determined by $\Upsilon^X$, hence denpending only on the nodes $x(v_i)$. The particular expression for the Jacobian can be obtained from lemma \ref{lem511}. In this particular situation, geodesics joining nodes are affine lines, and tangent vectors $t^X_{v_0v_i}$ turn out to be $x(v_i)-x(v_0)$. Application of the affine volume form $\dd x^1\wedge\ldots\wedge\dd x^n$ is the computation of some  determinant, therefore:
	\begin{equation*}
	\vol^X=c\cdot \dd x^1\wedge\ldots\wedge \dd x^n\Rightarrow Jac_{y^x_\beta}=c\cdot \left| \det(x^j(v_i)-x^j(v_0))_{i,j=1\ldots n}\right|=c\cdot \left| \det(x^j_i)_{i,j=0\ldots n}\right|
	\end{equation*}
\qed
\end{proof}

\begin{remark}
The discrete Lagrangian $L^x$ obtained by simplicial interpolation using the quadrature rule $Q_{mid}$ defined in (\ref{quadratures}) uses the barycenter as unique node, and  turns out to be:
	\begin{equation*}
	L^x(y^x_\beta)=\La(Bary(y^x_\beta))\cdot \vol^x_\beta,\qquad \vol^x_\beta=\int_{K^x_\beta}\vol^X=\frac{c}{n!}\cdot \left| \det(x^j_i)_{i,j=0\ldots n}\right|
	\end{equation*}
Therefore when all nodes associated to $\beta\in X_n$ are in general position, the mapping $Bary$ establishes a one-to-one correspondence between smooth Lagrangian densities and discrete Lagrangian densities, which is the association determined in definition (\ref{inducedsimpgeod}), using any euclidean structure, and the mid-point quadrature rule.
\end{remark}

\section{Example: kinematics of a Cosserat rod}
Consider a 1D filament, whose elements are parameterized by $s\in\mathbb{R}$, and freely moving for time $t\in\RR$, on the euclidean space $\RR^3$. Choose at each filament element (seen as rigid body) a referential centered at its center of mass, and oriented along its principal axes of inertia. That is, the configuration of each element $s$ of the filament at time $t$ is characterized by its spatial position $r(s,t)\in\RR^3$ and orientation $R(s,t)\in SO(3)$.  Further assume that the overall state of the filament is determined by the configuration of all of its elements. The time evolution of this filament (a Cosserat rod) can be seen then as a mapping:
	\begin{equation*}
	(r,R)\colon \mathbb{R}^2_{(s,t)}\rightarrow \mathbb{R}^3\times SO(3)
	\end{equation*}
In this situation $r(s,t)$ represents the location of the {\em centroid} of the filament element $s$ at some given time $t$. The component $R(s,t)$ represents the spatial orientation of the filament element $s$ at time $t$, with $R^t\cdot R=\Id$, $\det R=1$. We refer to the appendix for different properties and notations for $SO(3)$ as Riemannian manifold.

For a particular filament evolution $(r(s,t),R(s,t))$ the components $\frac{\partial r}{\partial s}$, $\frac{\partial r}{\partial t}$ represent, respectively, the gradient of filament element location (linear strain of the filament) and the filament element linear velocity. 

The component $\widehat\Omega=(\dd l_{R^t})\frac{\partial R}{\partial s}\in Skew(3)$ has the physical interpretation of gradient of rigid body configurations (angular strain) at the filament element $s$, at some temporal instant $t$, measured in the referential associated to this filament element. The component $\widehat{\omega}=(\dd l_{R^t})\frac{\partial R}{\partial t}\in Skew(3)$ represents the rigid body angular velocity of the filament element $s$, at some temporal instant $t$, measured in the mentioned referential. We may use the identification $Skew(3)\simeq \RR^3$ each of these elements determine corresponding vectors $\Omega,\omega$ belonging to the Lie algebra $(\mathbb{R}^3,\times)$.

A typical action functional describing the dynamics of this filament \cite{Antman,Demoures15,SimoVuQuoc} is obtained by addition $\mathbb{L}=\mathbb{K}_{lin}+\mathbb{K}_{ang}-\mathbb{E}_{lin}-\mathbb{E}_{ang}-\mathbb{P}_{ot}$. First two components represent respectively linear and angular kinetic energies $\mathbb{K}_{lin}=\int\int\frac12\rho(s)\|(\partial r/\partial t)\|^2\dd s\dd t$ (with non-negative mass density $\rho(s)\geq 0$ at each element $s$), and $\mathbb{K}_{ang}=\int\int\frac12\omega^t J(s)\omega\dd s \dd t$ (with time-independent, diagonal inertia matrix $J(s)=diag(I_1(s),I_2(s),I_3(s))$ with nonnegative $I_1,I_2,I_3\geq 0$ principal moments of inertia). Next two components represent elastic energies $\mathbb{E}_{lin}=\int\int \frac12(R^t(\partial r/\partial s)-e(s))^t\cdot C_1(s)\cdot (R^t(\partial r/\partial s)-e(s))\dd s \dd t$, and $\mathbb{E}_{ang}=\int\int \frac12 \Omega^t\cdot C_2(s)\cdot \Omega\dd s\dd t$ due to linear and angular strains, respectively (with symmetric positive-definite matrices $C_1(s)$ and $C_2(s)$, that determine the elastic properties of filament element $s$, and $e(s)$ the unstressed linear strain associated to this element \cite{Demoures15}). The last component $\mathbb{P}_{ot}=\int\int \frac12 P(s,r)\dd s\dd t$ represents the potential energy associated to the filament element $s$, when located at position $r\in\mathbb{R}^3$, which may be generated by some gravitational or electric field.

We are working with sections of the bundle $Y\rightarrow X$, where $X=\mathbb{R}^2_{(s,t)}$ and $Y=X\times \mathbb{R}^3\times SO(3)$. Taking $\vol^X=\dd s\wedge \dd t$ the lagrangian function for this theory is:
	\begin{equation}\label{lagrelastica}
	\mathcal{L}=\frac12 \left( \rho(s) \cdot\|(\partial r/\partial t)\|^2+  \omega^t J(s)\omega-(R^t(\partial r/\partial s)-e(s))^t\cdot C_1(s)\cdot (R^t(\partial r/\partial s)-e(s))- \Omega^t\cdot C_2(s)\cdot \Omega- P(s,r)\right)
	\end{equation}
To discretize this action functional, consider now the 2D CFK simplicial complex $V$, and the immersion of its vertices $v\in V_0=\ZZ^2$ into $X=\RR^2$ using the mapping $x\colon \ZZ^2\rightarrow \RR^2$ given as 
	\begin{equation}\label{immparticular}
	x(i,j)=(s(i,j),t(i,j))=((i-j)\cdot \frac{\Delta s}{2},(i+j)\cdot \frac{\Delta t}{2})
	\end{equation}
where $\Delta s,\Delta t>0$ determine the level of discretization in the filament and in time, respectively.

Our choice of $x(i,j)$ has the following meaning: Points $(i,j)\in V_0$ with fixed $i+j=c$ will be associated to configurations for fixed time $t=c\cdot \Delta t/2$, of elements uniformly distributed along the filament, at positions $s_0+\Delta s\cdot \mathbb{Z}$. Points $(i,j)\in V_0$ with fixed $i-j=c$ will be associated to configurations of a given filament element $s=c\cdot \Delta s/2$, at different instants $t_0+\Delta t\cdot\mathbb{Z}$, with time-step $\Delta t$. This discretization represents a ``leapfrog'' mechanism, interleaving certain filament elements $s_0+\Delta s\cdot \mathbb{Z}$ at given times $t_0+\Delta t\cdot \ZZ$ alternated with different filament elements $(s_0+\frac{\Delta s}{2})+\Delta s\cdot \mathbb{Z}$ at times $(t_0+\frac{\Delta t}{2})+\Delta t\cdot \ZZ$.

This choice gives a particular meaning to discrete integrators in section \ref{secintegrator}, which are performed from a initial condition band $k-2\leq i+j\leq k+1$ (as indicated in remark \ref{remark4fim}), evolving in the direction $(1,1)$, representing here a position+velocity initial condition at times $t_0$ or $t_0+\frac{\Delta t}{2}$ for filament elements $\frac{\Delta s}{2}\mathbb{Z}$, integrated for increasing values of $i+j$, to determine the evolution of the filament in time (compare with \cite{AdlBobSur04,FadVol94})

Given the choice of immersion $x\colon V_0\hookrightarrow X$, we may discretize the smooth lagrangian density $\mathcal{L}\vol^X$ to obtain a discrete lagrangian density on $Y^x\simeq \ZZ^2\times \RR^3\times SO(3)$ using simplicial geodesic interpolation on $Y=\RR^2\times\RR^3\times SO(3)$. Choosing the Riemannian metric defined by any euclidean structure on $\RR^2$ and $\RR^3$ and the bi-invariant metric on $SO(3)$ induced by the halved Frobenius scalar product (see the Appendix for the properties of this Riemannian manifold), following Lemma \ref{splits} geodesics project into straight lines on the first components $\RR^2$, $\RR^3$ and into geodesics on $SO(3)$, explicityly described in (\ref{geodesicasSO3}).

Following corollary \ref{indexingncells}, faces $\beta\in V_2$ in the 2-D CFK simplicial complex can be indexed by its least-weight vertex $(i,j)$ and a permutation $(1,2)$ or $(2,1)$ of $\mathrm{Sym}(2)$.  Denoting as $+$ the identity permutation and by $-$ the non-identity permutation, any element $\beta\in V_2$ has three vertices (ordered by increasing weight) $\{v_0,v_1,v_2\}=\beta$ on $\ZZ^2$, where $v_2=v_0+(1,1)$ and $v_1$ is either $v_0+(1,0)$ or $v_0+(0,1)$. We may determine all vertices using:
	\begin{equation*}
	(i,j,\pm)\simeq \{v_0,v_1,v_2\}\in X_2\Rightarrow v_0=(i,j),\, v_2=(i,j)+(1,1),\, v_1=\left(i+\frac12,j+\frac12\right)\pm\left(\frac12,\frac{-1}{2}\right)
	\end{equation*}
Consequently for the immersion (\ref{immparticular}), if $x(v_0)=((i-j)\Delta s/2,(i+j)\Delta t/2)=(s_0,t_0)$ then $x(v_2)=(s_0,t_0+\Delta t)$ and $x(v_1)=(s_0\pm \Delta s/2,t_0+\Delta t/2)$, the sign depending on whether $v_1=v_0+(1,0)$ or $v_0+(0,1)$.

Any configuration $y^x_\beta\in Y^x_\beta$ at any $2$-cell $\beta\in X_2$ is then given as a sequence 
	\begin{equation}\label{pbetax}
	y^x_\beta=\left((s_0,t_0,r_0,R_0), (s_0\pm \Delta s/2,t_0+\Delta t/2,r_1,R_1), (s_0,t_0+\Delta t,r_2,R_2)\right)
	\end{equation}
where $\frac{s_0}{\Delta s}\in \frac{1}{2}\mathbb{Z}$, $\frac{t_0}{\Delta t}\in \frac{1}{2}\mathbb{Z}$, $\frac{s_0}{\Delta s}+\frac{t_0}{\Delta t}\in\ZZ$, $r_1,r_2,r_3\in \RR^3$, $R_1,R_2,R_3\in SO(3)$.

\begin{remark}
For any given smooth section $(s,t)\in\RR^2\mapsto (r(s,t),R(s,t))\in Y=\RR^3\times SO(3)$, the induced discrete section $y^x\colon X\rightarrow Y^x$ determines, at each face $\beta\in V_2$, a configuration (\ref{pbetax}) explicitly obtained by:
	\begin{equation}\label{formula63}
	\begin{aligned}
	&(r_0,R_0)=(r(s_0,t_0),R(s_0,t_0)), \quad (r_1,R_1)=(r(s_0\pm\Delta s/2,t_0+\Delta t/2), R(s_0\pm\Delta s/2,t_0+\Delta t/2)),\\
	&	(r_2,R_2)=(r(s_0,t_0+\Delta t),R(s_0,t_0+\Delta t))
	\end{aligned}
	\end{equation}
\end{remark}
\begin{prop}\label{prop62}
For any 2-dimensional simplicial complex, consider any immersion $x\colon V_0\rightarrow \RR^2$ and any 2-cell $\beta\in V_2$ such that its adherent vertices $v_0,v_1,v_2\prec\beta$ (in any order) determine non-collinear nodes $x(v_0),x(v_1),x(v_2)\in\RR^2$. Consider a configuration:
	\begin{equation*}
	y^x_\beta=\left((s_0,t_0,r_0,R_0),(s_1,t_1,r_1,R_1),(s_2,t_2,r_2,R_2) \right)\in Y^x_\beta
	\end{equation*}
where $(s_i,t_i)\in\RR^2$, $r_i\in \RR^3$, $R_i\in SO(3)$, for each $i=0,1,2$.

If $\Trace(R_0^tR_1)>1$ and $\Trace(R_0^tR_2)>1$, the configuration is suited for simplicial geodesic interpolation in the sense of definition \ref{suited}.

Moreover, the linear mapping:
	\begin{equation*}
	\phi_{v_0}\colon T_{x(v_0)}\RR^2\rightarrow T_{y^x_{v_0}}Y=T_{x(v_0)}\RR^2\oplus T_{r_0}\RR^3\oplus T_{R_0}SO(3)\simeq \RR^2\oplus\RR^3\oplus Skew(3)
	\end{equation*}
given in lemma \ref{lem511} is determined by the following matrix (with respect to the basis $\partial/\partial s$, $\partial/\partial t$ on $T\RR^2$):
	\begin{equation}\label{phiSO3}
	\phi_{v_0}^{\RR^3\oplus Skew(3)}=\left[ \begin{matrix} \Delta^{01} r & \Delta^{02}r\\
	\Delta^{01} R & \Delta^{02}R	\end{matrix}\right]\cdot \left[ \begin{matrix} \Delta^{01} x & \Delta^{02}x\end{matrix}\right]^{-1}
	\end{equation}
where we define, for any $i,j=0,1,2$ the column vectors:
	\begin{equation}\label{deltas}
	\Delta^{ij} R=\log R_i^tR_j,\quad \Delta^{ij} r=r_j-r_i,\quad \Delta^{ij} x=x(v_j)-x(v_i)=\left[\begin{matrix} s_j-s_i \\ t_j-t_i\end{matrix}\right]
	\end{equation}
\end{prop}
\begin{proof}
For the points $R_0,R_1,R_2$ given on the Riemannian manifold $SO(3)$ (with the bi-invariant metric induced by the halved Frobenius product), geodesic distances $dist(R_i,R_j)=d_{ij}$ satisfy $1+2\cos d_{ij}=tr(R_i^tR_j)>1$. We can then warrant that all the three configurations on $SO(3)$ belong to some geodesic ball with radius $\rho<\pi/2$ (it suffices to take $R_0$ as center of the ball).

Geodesic interpolation with respect to some Euclidean structure on $\RR^2$ and $\RR^3$ is simply affine interpolation, and is always well defined. 

The Riemannian manifold $SO(3)$ has a constant sectional curvature $K=1/4$.
We may then observe for $\rho<\pi/2$ that $\left(\frac{\pi}{4\rho}\right)^2>(1/2)^2=1/4=K$. Following lemma \ref{sectionalcurvatures}, we conclude that simplicial geodesic interpolation associated to $y^x_\beta$ is well-defined on the $SO(3)$ component.

Therefore following lemma \ref{splits}, geodesic simplicial interpolation is well defined for the given configuration (\ref{pbetax}) on the product manifold.

Geodesic interpolation projected to $X$ determines $\Upsilon^X_{y_\beta^x}\colon \Delta_\beta\rightarrow X=\RR^2$, a simple affine interpolation transforming the vertices of the simplex $\Delta_\beta$ into the nodes $x(v_0),x(v_1),x(v_2)$. As these nodes are not collinear, we conclude the regularity of $\Upsilon^X_{y_\beta^x}$. 

For any edge $\{v_i,v_j\}$ adherent to $\beta$, the geodesic joining the $v_i$-configuration $(s_i,t_i,r_i,R_i)$ to the $v_j$-configuration $(s_j,t_j,r_j,R_j)$ on $Y$ is given by straight lines on the $\RR^2\times\RR^3$-component and by (\ref{geodesicasSO3}) on the $SO(3)$-component, hence using the above mentioned identification $T_{R_i}SO(3)\simeq Skew(3)$ this geodesic has as tangent vector $(x(v_j)-x(v_i),r_j-r_i,\log R_i^tR_j)=(\Delta^{ij}x,\Delta^{ij}r,\Delta^{ij}R)$.

Using these tangent vectors, and following the characterization given in definition \ref{inducedsimpgeod}, the linear mapping $\phi_{v_0}$ should transform the tangent vector $\Delta^{01}x$ into $(\Delta^{01}r,\Delta^{01}R)$, and $\Delta^{02}x$ into $(\Delta^{02}r,\Delta^{02}R)$, which leads to the expressions in our statement.\qed 
\end{proof}
\begin{corollary}
For the CFK 2-dimensional simplicial complex $V$, for the particular immersion $x\colon V_0\rightarrow X=\mathbb{R}^2$ determined by (\ref{immparticular}), for any configuration $y^x_\beta\in Y^x_\beta$ given by (\ref{pbetax}) and in the case that two of the inequalities $\Trace(R_0^tR_1)>1$, $\Trace(R_1^tR_2)>1$, $\Trace(R_2^tR_0)>1$ hold, the configuration is suited for simplicial geodesic interpolation, and the linear mappings 
	\begin{equation*}
	\phi_{v_i}\colon T_{x(v_i)}\RR^2\rightarrow T_{y^x_{v_i}}Y=T_{x(v_i)}\RR^2\oplus T_{r_i}\RR^3\oplus T_{R_i}SO(3)\simeq \RR^2\oplus\RR^3\oplus Skew(3)
	\end{equation*}
given in definition \ref{inducedsimpgeod} have, in the basis $\partial/\partial s$, $\partial/\partial t$, the following components:
	\begin{equation}\label{phiv1}
	\begin{aligned}
	&\phi_{v_0}^{\RR^3\oplus Skew(3)}=\left[ \begin{matrix} \frac{2\Delta^{01}r-\Delta^{02}r}{\pm \Delta s} & \frac{\Delta^{02}r}{\Delta t} \\
	\frac{2\Delta^{01}R-\Delta^{02}R}{\pm \Delta s} & \frac{\Delta^{02}R}{\Delta t}
	\end{matrix}\right]\quad 
	&\phi_{v_2}^{\RR^3\oplus Skew(3)}=\left[ \begin{matrix} \frac{2\Delta^{21}r-\Delta^{20}r}{\pm \Delta s} & \frac{-\Delta^{20}r}{\Delta t} \\
	\frac{2\Delta^{21}R-\Delta^{20}R}{\pm \Delta s} & \frac{-\Delta^{20}R}{\Delta t}
	\end{matrix}\right]\\
	&\phi_{v_1}^{\RR^3\oplus Skew(3)}=\left[ \begin{matrix} \frac{-\Delta^{10}r-\Delta^{12}r}{\pm \Delta s} & \frac{\Delta^{12}r-\Delta^{10}r}{\Delta t} \\
	\frac{-\Delta^{10}R-\Delta^{12}R}{\pm \Delta s} & \frac{\Delta^{12}R-\Delta^{10}R}{\Delta t}
	\end{matrix}\right] &
	\end{aligned}
	\end{equation}
where $\Delta^{ij}R$, $\Delta^{ij}r$ are defined by (\ref{deltas})
\end{corollary}
\begin{proof}
Apply proposition \ref{prop62}, ordering the vertices so that $v_0$ is any of the three possible choices $v\prec\beta$, and take into account the particular matrix $\left[ \begin{matrix} \Delta^{ij} x & \Delta^{ik}x\end{matrix}\right]^{-1}$, for each of these choices, in formula (\ref{phiSO3})\qed
\end{proof}

\begin{remark}
Some interesting remarks about elements defined in (\ref{deltas}), for our particular immersion (\ref{immparticular}), are $\Delta^{ij}r=-\Delta^{ji}r$, $\Delta^{ki}r+\Delta^{ij}r=\Delta^{kj}r$, the same holds for $\Delta^{ij}x$, and therefore we may observe that $\phi^{\RR^3}_{v_0}=\phi^{\RR^3}_{v_1}=\phi^{\RR^3}_{v_2}$ (this is because the geodesic interpolator generates affine mappings, whose differential is the same at any vertex). The same is not true for the $SO(3)$-component. From $\log R^t=-\log R$ one may derive $\Delta^{ij}R=-\Delta^{ji}R$. Using $\log\circ \Ad_R=\ad_R\circ \log$ we also derive $\ad_{R_i^tR_j}(\Delta^{ij}R)=\Delta^{ij}R$. This aspect is useful to simplify the computation of $\Delta^{ij}R$ (which turns out to be invariant by the action of $R_i^tR_j$, with a norm $d_{ij}$ given by $1+2\cos d_{ij}=\Trace(R_i^tR_j)$). However as $SO(3)$ is a non-commutative group, the relation $\expp \Delta^{ki}R\cdot \expp \Delta^{ij}R\cdot \expp \Delta^{jk}R=\Id$ does not generally imply $\Delta^{ki}R+\Delta^{ij}R=\Delta^{kj}R$, in particular the components $\phi^{Skew(3)}_{v_0}$, $\phi^{Skew(3)}_{v_1}$, $\phi^{Skew(3)}_{v_2}$ are not equal.
\end{remark}
Taking into account that all faces $\beta\in V_2$ are transformed into triangles with area $ \vol^x_\beta=\frac{\Delta s\cdot \Delta t}{4}$ by our particular immersion $x$ (given in (\ref{immparticular})), we may derive explicit expressions for the discrete Lagrangian associated to (\ref{lagrelastica}). This is achieved with definition \ref{inducedsimpgeod} if we choose any quadrature rule $Q$ as given in (\ref{quadratures}) with a single node or with 3 nodes.

In a simple case, we may use a 0-order rule $Q_v$, making a choice of some vertex $v(\beta)$ for each 2-cell $\beta$.
We shall take the node $v(\beta)=v_0$ at each 2-cell, for the quadrature rule. In this case, taking $\partial /\partial s$, $\partial/\partial t$ as basis of $T_{x(v(\beta))}\RR^2$ and the identification $T_{r(v(\beta))}\RR^3\oplus T_{R(v(\beta))} SO(3)=\RR^3\oplus Skew(3)$ we get from (\ref{phiv1}):
	\begin{equation*}
	\phi_{v(\beta)}=\left[ \begin{matrix} \frac{2\Delta^{01} r-\Delta^{02} r}{\pm \Delta s} & \frac{\Delta^{02} r}{\Delta t} \\
	\frac{2\Delta^{01} R-\Delta^{02} R}{\pm \Delta s} & \frac{\Delta^{02} R}{\Delta t}\end{matrix}\right]\in \mathrm{Hom}(\RR^2,\RR^3\oplus Skew(3))
	\end{equation*}
This linear mapping $\phi_{v(\beta)}$ at the configuration given in (\ref{formula63}) has the following components, in terms of $r(s,t),R(s,t)$:
	\begin{equation*}
	\begin{aligned}
	\frac{2\Delta^{01} r-\Delta^{02} r}{\pm \Delta s}&=\frac{r(s_0\pm\Delta s/2,t_0+\Delta t/2)-\frac12(r(s_0,t_0)+r(s_0,t_0+\Delta t))}{\pm \Delta s/2} \\
	\frac{2\Delta^{01} R-\Delta^{02} R}{\pm \Delta s}&=\frac{\log R(s_0,t_0)^tR(s_0\pm \Delta s/2,t_0+\Delta t/2)-\frac12 \log R(s_0,t_0)^tR(s_0,t_0+\Delta t)}{\pm \Delta s/2} \\
	\frac{\Delta^{02} r}{\Delta t}&=\frac{r(s_0,t_0+\Delta t)-r(s_0,t_0)}{\Delta t} \\
	\frac{\Delta^{02} R}{\Delta t}&=\frac{\log R(s_0,t_0)^tR(s_0,t_0+\Delta t)}{\Delta t}
	\end{aligned}
	\end{equation*}
allowing to interpret $\phi_{v(\beta)}$ as discrete version of $\partial(r,R)/\partial(s,t)$, computed at the point $(s_0,t_0+\Delta t/2)$.

\begin{prop}
Taking a quadrature rule $Q_{v(\beta)}$ with unique node at $v(\beta)=v_0=(i,j)$, for each $\beta=(i,j,\pm)\in X_2$, the corresponding discrete Lagrangian $$(s_0,t_0,r_0,R_0,s_1,t_1,r_1,R_1,s_2,t_2,r_2,R_2)\in Y^x_\beta\mapsto L^x_\beta (s_0,t_0,r_0,R_0,s_1,t_1,r_1,R_1,s_2,t_2,r_2,R_2)\in \RR$$ obtained by simplicial geodesic interpolation (definition \ref{inducedsimpgeod}) applied to the Lagrangian density (\ref{lagrelastica}) is $L^x=K^x_{lin}+K^x_{ang}-E^x_{lin}-E^x_{ang}-P_{ot}^x$ where each component is given as:
	\begin{equation*}
	K^x_{lin}=\frac12 \rho(s_0)\cdot \left\| \frac{\Delta^{02}r}{\Delta t}\right\|^2 \cdot \frac{\Delta s \Delta t}{4}
,\qquad
	K^x_{ang}=\frac12 \left( \frac{\Delta^{02}R}{\Delta t}\right)^t\cdot J(s_0)\cdot  \left( \frac{\Delta^{02}R}{\Delta t}\right) \cdot \frac{\Delta s \Delta t}{4}
	\end{equation*}
	\begin{equation*}
	E^x_{lin}=\frac12 \left(R_0^t\frac{2\Delta^{01}r-\Delta^{02}r}{\pm\Delta s} -e(s_0)\right)^t\cdot  C_1(s_0)\cdot \left(R_0^t\frac{2\Delta^{01}r-\Delta^{02}r}{\pm\Delta s}-e(s_0) \right) \cdot \frac{\Delta s \Delta t}{4}
	\end{equation*}
	\begin{equation*}
	E^x_{ang}=\frac12\left(\frac{2\Delta^{01}R-\Delta^{02}R}{\pm\Delta s}\right)^t\cdot C_2(s_0) \cdot \left(\frac{2\Delta^{01}R-\Delta^{02}R}{\pm\Delta s}\right)\frac{\Delta s \Delta t}{4},\qquad P_{ot}^x=\frac12 P(s_0,r_0)\frac{\Delta s \Delta t}{4}
	\end{equation*}
\end{prop}
Any of these discrete lagrangians is obtained applying a 0-order quadrature rule to a given Lagragian density that was invariant with respect to euclidean transformations of the space where the filament evolves. Therefore, the associated discrete action functional maintains the group of euclidean transformations as symmetries. Discrete and smooth conservation laws arise in both formalisms, with a meaning of linear momentum-work equilibrium conditions, and angular momentum-work equilibrium conditions.

In the case that $y^x_\beta\in Y^x_\beta$ is induced by a particular section $y^x\colon (i,j)\in\ZZ^2\mapsto (s,t,r,R)(i,j)\in Y^x$, for any 2-cell $\beta=(i,j,\pm)$ we must observe that:
	\begin{equation*}
	s_0(y^x_\beta)=(i-j)\frac{\Delta s}{2},\quad t_0(y^x_\beta)=(i+j)\frac{\Delta t}{2}
	\end{equation*}
	\begin{equation*}
	r_0(y^x_\beta)=r(i,j),\quad r_1(y^x_\beta)=r\left((i+1/2,j+1/2)\pm(1/2,-1/2)\right),\quad r_2(y^x_\beta)=r(i+1,j+1)
	\end{equation*}
	\begin{equation*}
	R_0(y^x_\beta)=R(i,j),\quad R_1(y^x_\beta)=R\left((i+1/2,j+1/2)\pm(1/2,-1/2)\right),\quad R_2(y^x_\beta)=R(i+1,j+1)
	\end{equation*}
which may be substituted into our discrete Lagrangians using (\ref{deltas}), to determine the discrete action functional on the discrete section $y^x$.

In order to derive explicit expressions for the conservation laws (\ref{32a}) and for discrete Euler-Lagrange equations (\ref{ELbis}), it is necessary to compute the differential of the discrete Lagrangian. The differential $\dd_{r_\beta,R_\beta}L$ of any function $L\colon Y^x_3\rightarrow \RR$ splits into two components, one on $\RR^3\oplus\RR^3\oplus \RR^3$ and the second one on $Skew(3)\oplus Skew(3)\oplus Skew(3)\simeq \RR^3\oplus\RR^3\oplus\RR^3$. This splitting is associated to the direct product structure $Y^x=V_0\times(\RR^3\times SO(3))$ and to the identifications $T_{r_i}\RR^3=\RR^3$, $T_{R_i}SO(3)=Skew(3)\simeq \RR^3$. We denote both components by $\dd^{\RR^3}$ and $\dd^{Sk}$.

For the linear kinetic and potential $K^x_{lin}$, $P_{ot}^x$, which don't depend on the $SO(3)$-valued component, this differential is easy to compute taking into account that $\left(\dd_{(r_0,r_1)}\Delta^{01}r\right)(e_0,e_1)=e_1-e_0$. Using the Euclidean metric to identify $V\RR^3$ with its dual space we obtain:
	\begin{equation*}
	\begin{aligned}
	\dd^{\RR^3}_{y_\beta^x}K^x_{lin}\circ i_{y^x_{v_0}}^{y^x_\beta}& =-\dd^{\RR^3}_{y_\beta^x}K^x_{lin}\circ i_{y^x_{v_2}}^{y^x_\beta}=\frac{\Delta s}{4}\rho(s_0)\cdot \frac{\Delta^{02}r}{\Delta t},\qquad 
	\dd^{\RR^3}_{y_\beta^x}K^x_{lin}\circ i_{y^x_{v_1}}^{y^x_\beta}=0\\
	\dd^{\RR^3}_{y_\beta^x}P^x_{ot}\circ i_{y^x_{v_0}}^{y^x_\beta}&=\frac{\Delta s\Delta t}{8} \nabla P(s_0,r_0),\quad \dd^{\RR^3}_{y_\beta^x}P^x_{ot}\circ i_{y^x_{v_1}}^{y^x_\beta}=\dd^{\RR^3}_{y_\beta^x}P^x_{ot}\circ i_{y^x_{v_2}}^{y^x_\beta}=0,\qquad 
	\\
	\dd^{Sk}_{y_\beta^x}K^x_{lin}\circ i_{y^x_{v_0}}^{y^x_\beta}&=\dd^{Sk}_{y_\beta^x}K^x_{lin}\circ i_{y^x_{v_1}}^{y^x_\beta}=\dd^{Sk}_{y_\beta^x}K^x_{lin}\circ i_{y^x_{v_2}}^{y^x_\beta}=0\\
	\dd^{Sk}_{y_\beta^x}P^x_{ot}\circ i_{y^x_{v_0}}^{y^x_\beta}&=\dd^{Sk}_{y_\beta^x}P^x_{ot}\circ i_{y^x_{v_1}}^{y^x_\beta}=\dd^{Sk}_{y_\beta^x}P^x_{ot}\circ i_{y^x_{v_2}}^{y^x_\beta}=0
	\end{aligned}
	\end{equation*}
denoting by $\nabla P(s,r)$ the gradient of the function $e\mapsto P(s,e)$, at any point $r\in\RR^3$

For the angular kinetic energy dependence on $SO(3)$ is only through $R_0$, therefore identifying $V_{R_0}(SO(3))\simeq \RR^3$ with its dual, using the Euclidean metric (or  Halved Frobenius when interpreted on $Skew(3)$) we get:
	\begin{equation*}
	\begin{aligned}
	\dd^{\RR^3}_{y_\beta^x}E^x_{lin}\circ i_{y^x_{v_0}}^{y^x_\beta}&=\dd^{\RR^3}_{y_\beta^x}E^x_{lin}\circ i_{y^x_{v_2}}^{y^x_\beta}=\frac{\mp\Delta t}{4}R_0\cdot C_1(s_0)\cdot \left( R_0^t\cdot \frac{2\Delta^{01}r-\Delta^{02}r}{\pm \Delta s}-e(s_0)\right)\\
	\dd^{\RR^3}_{y_\beta^x}E^x_{lin}\circ i_{y^x_{v_1}}^{y^x_\beta}&=\frac{\pm\Delta t}{2}R_0\cdot C_1(s_0)\cdot \left( R_0^t\cdot \frac{2\Delta^{01}r-\Delta^{02}r}{\pm \Delta s}-e(s_0)\right)
	\end{aligned}
	\end{equation*}
We may use that $e^t\cdot (v\times w)=-v^t\cdot (e\times w)$ on $\RR^3$ to obtain $e^t\cdot R\cdot \widehat{v}=-v^t\cdot \widehat{R^te}$, for any pair of vectors $e,v\in\RR^3$, and also $\widehat{R^t e}=R^t\cdot \widehat{e}\cdot R$. Hence:
	\begin{equation*}
	\begin{aligned}
	\dd^{Sk}_{y_\beta^x}E^x_{lin}\circ i_{y^x_{v_0}}^{y^x_\beta}&=\frac{\mp\Delta t}{4} R_0^t\cdot \widehat{(2\Delta^{01}r-\Delta^{02}r)}\cdot R_0\cdot C_1(s_0)\cdot \left( R_0^t\cdot \frac{2\Delta^{01}r-\Delta^{02}r}{\pm \Delta s}-e(s_0)\right)\\
	\dd^{Sk}_{y_\beta^x}E^x_{lin}\circ i_{y^x_{v_1}}^{y^x_\beta}&=\dd^{Sk}_{y_\beta^x}E^x_{lin}\circ i_{y^x_{v_2}}^{y^x_\beta}=0
	\end{aligned}
	\end{equation*}

For the angular kinetic and internal energies $K^x_{ang}$, $E^x_{ang}$, which don't depend on the $\RR^3$-valued component, and where the dependence on the $SO(3)$-valued component is given in terms of $\Delta^{01}R$ and $\Delta^{02}R$, the differential may be expressed as $\RR^3$-valued function as indicated in (\ref{lema65apend}) using the identifications of $T_RSO(3)$ with $\RR^3$ (a different characterization can be found in \cite{Moakher02}, using the metric to identify linear transformations in $\RR^3$ with vectors in $\RR^3$):

Taking expression (\ref{dlog}) for $d_{(R_i,R_j)}\Delta R$, we may now derive the differential of the discrete angular kinetic lagrangian:
	\begin{equation*}
	\begin{aligned}
	\dd^{Sk}_{y_\beta^x}		K^x_{ang}\circ i_{y^x_{v_0}}^{y^x_\beta}&=
	 \frac{-\Delta s}{4}\left((\dlog \Delta^{02}R)\cdot J(s_0)\cdot \frac{\Delta^{02}R}{\Delta t}\right),\quad 
	\dd^{Sk}_{y_\beta^x}		K^x_{ang}\circ i_{y^x_{v_1}}^{y^x_\beta}=0\\
	\dd^{Sk}_{y_\beta^x}		K^x_{ang}\circ i_{y^x_{v_2}}^{y^x_\beta}&=
	\frac{\Delta s}{4}\left((\dlog \Delta^{02}R)^t\cdot J(s_0)\cdot \frac{\Delta^{02}R}{\Delta t}\right) \\
	\dd^{\RR^3}_{y_\beta^x}		K^x_{ang}\circ i_{y^x_{v_0}}^{y^x_\beta}&=\dd^{\RR^3}_{y_\beta^x}		K^x_{ang}\circ i_{y^x_{v_1}}^{y^x_\beta}=\dd^{\RR^3}_{y_\beta^x}		K^x_{ang}\circ i_{y^x_{v_2}}^{y^x_\beta}=0\end{aligned}
	\end{equation*}
We obtain using againg (\ref{dlog}) the following expression for the differential of the discrete internal energy associated to angular strain:
	\begin{equation*}
	\begin{aligned}
	\dd^{Sk}_{y_\beta^x}		E^x_{ang}\circ i_{y^x_{v_0}}^{y^x_\beta}&=\frac{\mp \Delta t}{4} \cdot \left(2\dlog \Delta^{01}R-\dlog\Delta^{02}R\right)\cdot C_2(s_0)\cdot \left(\frac{2\Delta^{01}R-\Delta^{02}R}{\pm \Delta s}\right) \\
	\dd^{Sk}_{y_\beta^x}		E^x_{ang}\circ i_{y^x_{v_1}}^{y^x_\beta}&=\frac{\pm \Delta t}{2} \cdot \left(\dlog \Delta^{01}R\right)^t\cdot C_2(s_0)\cdot \left(\frac{2\Delta^{01}R-\Delta^{02}R}{\pm \Delta s}\right) \\
	\dd^{Sk}_{y_\beta^x}		E^x_{ang}\circ i_{y^x_{v_2}}^{y^x_\beta}&=\frac{\mp \Delta t}{4} \cdot \left(\dlog \Delta^{02}R\right)^t\cdot C_2(s_0)\cdot \left(\frac{2\Delta^{01}R-\Delta^{02}R}{\pm \Delta s}\right) \\
		\dd^{\RR^3}_{y_\beta^x}		E^x_{ang}\circ i_{y^x_{v_0}}^{y^x_\beta}&=\dd^{\RR^3}_{y_\beta^x}		E^x_{ang}\circ i_{y^x_{v_1}}^{y^x_\beta}=\dd^{\RR^3}_{y_\beta^x}		E^x_{ang}\circ i_{y^x_{v_2}}^{y^x_\beta}=0
		\end{aligned}
	\end{equation*}
These particular expressions allow to derive discrete Euler-Lagrange equations (\ref{ELbis}), whose solutions satisfy the discrete Noether conservation laws (\ref{NoetConsLaw}), associated to any infinitesimal euclidean movement on $\mathbb{R}^3$ (which is always a symmetry of the smooth Lagrangian (\ref{lagrelastica})).

Without explicitly giving the system of discrete Euler-Lagrange equations (with a large number of components), its numerical integration can be executed in a simple way. For the direction $v\to w$ given by $v=v_0=(i_0,j_0)$, $w=v_2=(i_0+1,j_0+1)$ the value $\mu_{v,w}$ in (\ref{defmomentum}) can be computed, from the explicit expressions given above for the differential of the Lagrangian, if we know the configuration at each vertex $(i,j)$ with $i_0+j_0-2\leq i+j\leq i_0+j_0+1$. Equation $\mathrm{Leg}(y_{v_0},y_{v_2},y_{v_1^+},y_{v_1^-})=\mu_{v_0,v_2}$ determines then a system of equations implicitly defining the unknown $y_{v_2}$. Solving this system of equations is equivalent to giving the integrator $\phi_{v,w}$, as described in theorem \ref{integradorteorema}. We must observe from remark \ref{remark46} that the Legendre mapping uses the differential of the Lagrangian densities on the $v_0$-component, for two different simplices $\beta^+,\beta^-\in V_2$. These components are explicit computed along this section. The existence of an integrator (right-inverse of the Legendre mapping given in Remark \ref{remark46}) represents the center of the whole scheme that allows to reconstruct the discrete field, solution of the discrete  variational problem, when we know its values on some initial band $k-2\leq i+j\leq k+1$.

More precisely, to obtain the integrator, we must solve a equation $\mathrm{Leg}(y_{v_0},y_{v_2},y_{v_1^+},y_{v_1^-})=\mu_{v_0,v_2}$. Taking into account the dependence of $\mathrm{Leg}$ on $y_{v_2}=(r_2,R_2)$ is only in the terms $\Delta^{02}r$, $\Delta^{02}R$, we may first solve in these components $\Delta^{02} r$, $\Delta^{02}R$. We may observe from the expressions obtained in the differential of Lagrangian densities that the system of equations splits into a system for the $\RR^3$-component and another for the $Skew(3)$-component. The first one is a system of linear equations on $\Delta^{02}r$ that doesn't depend explicitly on $\Delta^{02}R$. Except in particular degenerate cases, this system leads to a unique explicit solution for $\Delta^{02}r$. The component on $Skew(3)$ is then a system of 3 equations for $\Delta^{02}R\in\RR^3$.
The main difficulty here is the non-linear behavior of the function $\dlog$, hence some linear or quadratic approximation of $\dlog$ may be applied. After obtaining the values $\Delta^{02}r$, $\Delta^{02}R$, the definition of these two components leads to $R_2=R_0\cdot \expp \Delta^{02}R$, $r_2=r_0+\Delta^{02}r$. We obtain then the configuration $y_{v_2}=(r_2,R_2)\in Y_{v_2}$. Application of this integration scheme on all vertices with $i+j=i_0+j_0$, we deduce the explicit configuration for the critical field on all vertices $(i,j)\in V_0$ with $i+j\leq i_0+j_0+2$. 

\section*{Appendix: Riemannian geometry on SO(3)}
We summarize next some useful results for the Lie group $SO(3)=\left\{ R\in \Hom(\RR^3,\RR^3)\,\colon\, R^t\cdot R=\Id_3\right\}$  (see, for example \cite{Moakher02,SimoVuQuoc}).

Consider the left product $l_R\colon S\in SO(3)\mapsto RS\in SO(3)$. Any tangent vector $A_R\in T_RSO(3)$ is determined by a unique left-invariant vector field $A\in\X^l(SO(3))$, whose value at $\Id$ is $A_{\Id}=\dd l_{R^t} A_R\in T_{Id}SO(3)\subset\Hom(\RR^3,\RR^3)$, fulfilling $A_{Id}^t+A_{Id}=0$. The space of left-invariant vector fields on $SO(3)$ and also each tangent space $T_RSO(3)$ are then naturally identified with the space $Skew(3)$ of skew-adjoint operators on $\RR^3$. The product mapping $pr\colon (R,S)\mapsto R^tS\in SO(3)$ has at any point $(R,S)\in SO(3)\times SO(3)$ the differential $\dd_{(R,S)}pr(e_1,e_2)=e_2-\ad_{S^tR}e_1$, if we use the identifications $T_RSO(3)=Skew(3)$, $T_SSO(3)=Skew(3)$, $T_{R^tS}SO(3)=Skew(3)$, and consider the linear mapping $\ad_R\colon Skew(3)\mapsto Skew(3)$, differential at $\Id$ of the internal group automorphism $\Ad_R\colon S\in SO(3)\mapsto RSR^t\in SO(3)$. Consequently (case $S=\Id$) the inversion mapping $\inv\colon R\mapsto R^t$ has at any point $R\in SO(3)$ the differential $\dd_R\inv=-\ad_R$, with the same identifications. 

Fix the standard euclidean metric and orientation choice on $\RR^3$. Fix the (halved) Frobenius scalar product $\langle A,B\rangle=\frac12\Trace A^tB$ on $\Hom(\RR^3,\RR^3)$. The linear mapping $v\in\RR^3\mapsto \widehat{v}\in Hom(\RR^3,\RR^3)$ defined by the exterior product $\widehat{v}(e)=v\times e$  is injective, it is further an isometry, with image the space $Skew(3)$, hence defining a linear isomorphism $\RR^3\simeq T_{\Id}SO(3)$. Further, considering the Lie algebra structure given by $\times$ on $\RR^3$ and given by $[A,B]=AB-BA$ on $Skew(3)\subset\Hom(\RR^3,\RR^3)$, the identification $v\leftrightarrow \widehat{v}$ is also a Lie-algebra isomorphism. The mapping $A\in\X(SO(3))\mapsto A_{\Id}\in Skew(3)$ also determines a Lie algebra isomorphism of $(Skew(3),\times)$ with the Lie algebra of left-invariant vector fields, with the commutator of fields as Lie bracket.

As a consequence we have a linear isomorphism of Lie algebras between $(\RR^3,\times)$ and the algebra of left-invariant vector fields on $SO(3)$. On each point $R\in SO(3)$, this is given as $v\in\RR^3\mapsto \dd l_R (\widehat{v})\in T_RSO(3)$. The adjoint mapping $\ad_R\colon Skew(3)\mapsto Skew(3)$ is identified then with $R\colon v\in \RR^3\mapsto R\cdot v\in\RR^3$, the differential of the mapping $(R,S)\mapsto R^tS$ at $(R,S)$ is then identified with $(v_1,v_2)\in\RR^3\oplus\RR^3 \mapsto v_2-S^t\cdot R\cdot v_1\in\RR^3$ and the differential of $R\mapsto R^t$ at $R$ is then identified with $v\in\RR^3\mapsto -R(v)\in\RR^3$

Moreover, the exponential mapping $\expp\colon \RR^3\rightarrow SO(3)$ can be expressed as $\expp 0=\Id$ and for non-vanishing vectors by Rodrigues' formula:
	\begin{equation}\label{expR3}
	\expp v=\Id+\frac{\sin \alpha}{\alpha} \widehat v+\frac{1-\cos\alpha}{\alpha^2}\widehat v\circ\widehat v\qquad (v=\alpha\cdot n,\, \alpha\neq 0,\, \|n\|=1 )
	\end{equation}
In this regard, we shall also note that $-\widehat{n}\circ\widehat{n}=\proj_n^\bot$ is the orthogonal projector to the plane orthogonal to $n$, and consequently $\widehat{n}^3=-\widehat{n}$, and $\widehat{v}^3=-\|v\|^2\cdot \widehat{v}$. Using this it is easy to see that the definition above makes $\exp sv$ a 1-parameter subgroup of $SO(3)$, with tangent vector $\widehat{v}\in Skew(3)$ at $s=0$. Properties of this mapping are $\expp{}\circ \ad_R=\Ad_R\circ \expp{}$, $\inv\circ\expp{}=\expp{}\circ(-\Id)$. It is not injective, namely $\expp v=\expp w$ precisely if $v,w$ are linearly dependent, and the norm of the difference $v-w$ belongs to $2\pi\ZZ$.

Elements $R\in SO(3)$ always have one eigenvector, and can be seen as a rotation with angle $\alpha\in[0,\pi]$ with respect to this vector. The trace is therefore $\Trace(R)=1+2\cos\alpha\in[-1,2]$. The two limiting cases correspond to the identity $R=\Id$, $\Trace(R)=3$, and half-turn rotations (that is, symmetries with respect to some axis $n$), with $\Trace(R)=-1$. These elements are precisely the involutive elements ($R^2=R$) on $SO(3)$.

From Rodrigues' formula (\ref{expR3}) we have $\Trace(\expp v)=3+\frac{1-\cos\|v\|}{\|v\|^2}\Trace(\widehat{v}^2)=3-2(1-\cos\|v\|)=1+2\cos\|v\|$ and $\expp (v)-\expp(-v)=2\frac{\sin \|v\|}{\|v\|}\widehat{v}$. Therefore, if we consider a noninvolutive element $R\in SO(3)$ with $R^2\neq R$, and hence $-1<\Trace(R)<3$, we may use $\Trace(R)=1+2\cos\|v\|$, $R-R^t=2\frac{\sin\|v\|}{\|v\|}\widehat{v}$ to determine the unique vector $v\in\RR^3$ with $\|v\|\in ]0,\pi[$ and such that $\expp v=R$.

For involutive elements, the unique element with $\Trace(R)=3$ is precisely $R=\Id_3$, which coincides with $\expp v$ in the case $v=0$. In the case of half-turn rotations, we have $R=\Id-2\proj_n^\bot=\expp \pi n=\expp -\pi n$, for some unit vector $n$. In this way, using the exponential mapping, the manifold $SO(3)$ can be identified with the closed ball with radius $\pi$ on $\RR^3$, gluing together antipodal points of its spherical boundary. This representation can be computationally simpler than the classical $3\times 3$ matrices used when working with this group elements

\begin{lemma}
The exponential mapping is a local diffeomorphism at each point, with injectivity radius $\rho=\pi$. It allows to identify the open neighborhood $U_{\Id}=\{R\in SO(3)\,\colon\, tr(R)>-1\}$ of $\Id\in SO(3)$ with the open ball $B_\pi(0)$  of $0\in \RR^3$.

The inverse mapping $\log\colon U_{\Id}\rightarrow B_{\rho}(0)$, called the group logarithm, is characterized for any pair $R\in U_{Id}$, $v\in B_\pi(0)$ by: 
	\begin{equation*}
	v=\log R \Leftrightarrow \Trace(R)=1+2\cos\|v\|,\quad R-R^t=2\frac{\sin\|v\|}{\|v\|}\widehat{v}
	\end{equation*}
and has the following properties:
	\begin{equation*}
	\log\circ \Ad_R=\ad_R\circ\log,\qquad \log\circ \inv=-\log
	\end{equation*}
	\begin{equation*}
	\dd_{\expp v} \log=-(\dd_{\expp -v}\log)\circ (\dd_{\expp v} \inv),\qquad (\dd_R\log)(e)=(\dd_{R^t}\log)(R\cdot e)
	\end{equation*}
\end{lemma}

On any Lie group it is known that the exponential mapping $\exp\colon T_{Id}G\mapsto G$ transforms 1-D linear subspaces into Riemannian geodesics, for any choice of a bi-invariant Riemannian metric. Taking the radius $\rho=\pi$ such that the exponential mapping is injective on the geodesic ball $B_{\rho}\subset \RR^3=Skew(3)$, for any element $A\in B_{\rho}$ the curve $(\exp sA)_{s\in[0,1]}$ is then a minimal geodesic with respect to the halved Frobenius metric. Consequently also $(R\exp sA)_{s\in[0,1]}$ are minimal geodesics, with respect to this metric. 

\begin{lemma}
For any $R_i,R_j\in SO(3)$ if $R_i^tR_j$ is not an involution (its square is not $\Id$), the shortest geodesic joininig $R_i$ at $s=0$ to $R_j$ at $s=1$ is given by:
	\begin{equation}\label{geodesicasSO3}
	\gamma(s)=R_i\exp (s\ac \log (R_i^tR_j))=\exp(s\ac \log(R_jR_i^t))R_i
	\end{equation}
The geodesic distance between $R_i,R_j$ is then given as $d(R_i,R_j)=\|\log(R_i^tR_j)\|=d_{ij}$, where $1+2\cos d_{ij}=tr(R_i^tR_j)$ holds.
\end{lemma}
\begin{lemma}\label{lema65apend}
Consider the mapping $\Delta R\colon SO(3)\times SO(3)\rightarrow \RR^3$, taking any pair $(R_i,R_j)$ into $\Delta^{ij}R=\log R_i^tR_j$ (and defined only when $\Trace(R_i^tR_j)>-1$).

The induced linear mapping $\dd_{(R_i,R_j)}\Delta R\colon \RR^3\oplus\RR^3\rightarrow \RR^3$, determined using the identifications $T_{R_i}SO(3)\simeq \RR^3$, $T_{R_j}SO(3)\simeq \RR^3$, $T_e\RR^3\simeq \RR^3$ has the explicit expression:
	\begin{equation*}
	\begin{aligned}
	\dd_{(R_i,R_j)}\left(\Delta R\right)\colon (e_i,e_j)&\mapsto \left(\Id+\frac{d_{ij}}{2}\widehat{n}_{ij} +(1-\frac{d_{ij}\sin d_{ij}}{2-2\cos d_{ij}})\widehat{n}_{ij}^2\right)(e_j-R_j^tR_ie_i)=\\
	&=\left(\left((1-\frac{d_{ij}\sin d_{ij}}{2-2\cos d_{ij}})\widehat{n}_{ij}+ \frac{d_{ij}}{2}\Id\right)\widehat{n}_{ij}+\Id\right)(e_j-R_j^tR_ie_i)
	\end{aligned}
	\end{equation*}
where 
	\begin{equation*}
	1+2\cos d_{ij}=tr(R_i^tR_j),\quad \widehat{n}_{ij}=(R_i^tR_j-R_j^tR_i)/tr(\Id-(R_i^tR_j)^2)
	\end{equation*}
In the particular case $R_i=R_j=R\in SO(3)$, the linear mapping $d_{(R,R)}\left(\Delta R\right)$ is just $(e_i,e_j)\mapsto e_j-e_i$.
\end{lemma}
\begin{proof}
A differentiated version of Campbell-Baker-Hausdorff formula (see \cite{Hall03} for the proof) shows that the exponential mapping $\expp\colon Lie(G)\rightarrow G$, for any Lie group G, has a differential at any arbitrary point $v\in Lie(G)$ satisfying:
	\begin{equation*}
	\left(d_{\Id}l_{\expp v}\right)^{-1}\left(\frac{d}{dt}\right)_{t=0}\expp (v+tw)=\left(\phi(-\ads v)\right)(w)\in T_{\Id}G=Lie(G) 
	\end{equation*}
where $l_{\expp v}$ is left multiplication with $\expp v$, where $\ads\colon Lie(G)\rightarrow Hom(Lie(G),Lie(G))$ is the linear mapping induced by $\ad\colon G\rightarrow Aut(Lie(G))$, given as $(\ads v)(e)=[v,e]$, and where $\phi(z)=\sum_{k=0}^\infty \frac{1}{(k+1)!} z^k$, for any linear mapping $z\colon Lie(G)\rightarrow Lie(G)$.

For the particular case $G=SO(3)$, we know that $Lie(G)=(\RR^3,\times)$, hence $\ads v=\widehat{v}$. Taking into account that $\widehat{v}^3=-\|v\|^2\cdot \widehat{v}$ we also have $\widehat{v}^{2m+j}=(-\|v\|^2)^m\cdot \widehat{v}^j$, for $j=1,2$ and for any $m\in\mathbb{N}$. We may conclude that $\dd_0\expp=\Id$, and for $\|v\|\neq 0$ (using the identification $\left(d_{\Id}l_{\expp v}\right)^{-1}\colon T_{\expp v}SO(3)\rightarrow T_{Id}SO(3)=Skew(3)$):
	\begin{equation*}
	\begin{aligned}
	\dd_v\expp=&\sum_{k=0}^\infty \frac{1}{(k+1)!} (-\widehat{v})^k=\Id+\left(\sum_{m=1}^\infty \frac{1}{(2m)!}(-1)^{2m-1}  (-\|v\|^2)^{m-1}\right)\cdot \widehat{v} +\\
	&+ \left(\sum_{m=1}^\infty \frac{1}{(2m+1)!}(-1)^{2m}  (-\|v\|^2)^{m-1}\right)\cdot \widehat{v}^2=\\
	&=\Id-\frac{1-\cos\|v\|}{\|v\|^2}\widehat{v}+\frac{\|v\|-\sin\|v\|}{\|v\|^3}\widehat{v}^2
	\end{aligned}
	\end{equation*} 
The mapping $\expp\colon \RR^3\rightarrow SO(3)$, defined in (\ref{expR3}), has the following associated differential at any given point $v\in\RR^3$:
	\begin{equation*}
	\dd_v\expp=\Id-\frac{1-\cos\alpha}{\alpha}\widehat{n}+(1-\frac{\sin\alpha}{\alpha})\widehat{n}^2
	\end{equation*}
where $\alpha=\|v\|>0$, $n=\frac{1}{\alpha}v$. In the particular case $v=0$ we have $\dd_0\expp=\Id$.

Using $\widehat{n}^3=-\widehat{n}$ for any unitary vector $n\in\mathbb{R}^3$ we easily conclude that the inverse of this mapping is:
	\begin{equation}\label{derivalogaritmo}
	\dd_{\expp v}\log=(\dd_v\expp)^{-1}=\Id+\frac{\alpha}{2}\widehat{n}+\left(1-\frac{\alpha\sin\alpha}{2-2\cos\alpha}\right)\widehat{n}^2
	\end{equation}
or $\dd_{Id}\log=(\dd_0\expp)^{-1}=\Id$ for the particular case $v=0$.

It was already indicated that $(R,S)\in SO(3)\times SO(3)\mapsto R^tS\in SO(3)$ induces at any point $(R,S)$ the following differential
(Using the identifications $T_RSO(3)\simeq \RR^3$, $T_SSO(3)\simeq \RR^3$, $T_{R^tS} SO(3)\simeq \RR^3$): $(e_1,e_2)\mapsto e_2-\ad_{S^tR}e_1$. Taking into account that $\ad_R\colon e\in\RR^3\mapsto R(e)\in\RR^3$, we conclude that our mapping has the following associated differential: $(e_1,e_2)\mapsto e_2-S^t\cdot R\cdot e_1$,

From the definition of $\Delta R$ and taking $v=\log R_i^tR_j$ we get now:
	\begin{equation}\label{dRR}
	\left(\dd_{(R_i,R_j)}\left(\Delta R\right)\right) (e_i,e_j)=(\dd_{\expp v}\log )(e_j-R_j^tR_ie_i)
	\end{equation}
and combining this expression with (\ref{derivalogaritmo}) yields the formula in the lemma.\qed
\end{proof}
In the previous lemma we should mention that $0<d_{ij}=d(R_i,R_j)=\|\log R_i^tR_j\|=\|\Delta^{ij}R\|<\pi$ (with the halved Frobenius norm on $Skew(3)$, or equivalently, the canonical norm in $\RR^3$),  $d_{ij}\cdot n_{ij}=\log R_i^tR_j$. Observe that $d_{(R_i,R_j)}\Delta^{ij}R$ doesn't depend on $R_i$ or $R_j$, but only on $R_i^tR_j=\expp \Delta^{ij}R$, therefore we may define:
\begin{define}
We call $\dlog\colon \RR^3\rightarrow \Hom(\RR^3,\RR^3)$ the following mapping:
	\begin{equation*}
	\dlog v=\left( \frac{2-2\cos\|v\|-\|v\|\sin\|v\|}{\|v\|^2(2-2\cos\|v\|)}\widehat{v}+\frac12 \Id\right)\circ \widehat{v}+\Id
	\end{equation*}
\end{define}
(Observe that, for numerical reasons, the well-defined fraction accompanying $\widehat{v}^2$ should not be computed  in floating point arithmetic applying the substractions, that are ill-conditioned for small $\|v\|$. An equivalent form for this fraction would be $\frac{2\sin\|v\|-\|v\|(1+\cos\|v\|)}{2\|v\|^2\sin\|v\|}$, where only the difference appearing on the numerator is ill-conditioned and has to be treated with care)
\begin{lemma}
There holds:
	\begin{equation}\label{dlog}
	(\dd_{(R_i,R_j)}(\Delta R))(e_i,e_j)=\left( \dlog \Delta^{ij} R\right)(e_j)-\left( \dlog \Delta^{ij}R\right)^t(e_i)
	\end{equation}
where $(\Delta R)(R_i,R_j)=\log R_i^tR_j=\Delta^{ij}R$, and we are using the standard identifications $e\in\mathbb{R}^3\simeq Skew(3)\simeq T_R SO(3)$.
\end{lemma}
\begin{proof}
Using the definition of $\dlog$ and lemma \ref{lema65apend} we get:
	\begin{equation*}
	(\dd_{(R_i,R_j)}(\Delta R))(0,e_j)=\left( \dlog \Delta^{ij}R\right)(e_j)
	\end{equation*}
Taking into account $\log\circ\inv=-\log$, we conclude that $\Delta^{ij}R=(\Delta R)(R_i,R_j)=-(\Delta R)(R_j,R_i)=-\Delta^{ji}R$. Therefore 
	\begin{equation*}
	(\dd_{(R_i,R_j)}(\Delta R))(e_i,0)=-(\dd_{(R_j,R_i)}(\Delta R))(0,e_i)=-\left( \dlog \Delta^{ji}R\right)(e_i)
	\end{equation*}
We may observe that $\dlog(-v)=(\dlog v)^t$. There also holds $\Delta^{ji}R=-\Delta^{ij}R$. Therefore:
	\begin{equation*}
	(\dd_{(R_i,R_j)}(\Delta R))(e_i,e_j)=( \dlog \Delta^{ij}R)(e_j)-(\dlog \Delta^{ji}R)(e_i)=( \dlog \Delta^{ij}R)(e_j)-(\dlog \Delta^{ij}R)^t(e_i)
	\end{equation*}
concluding our proof.\qed
\end{proof}

\section*{Acknowledgements} The first Author  has been partially supported by Funda\c{c}\~ao para a Ci\^encia e a Tecnologia (Portugal) through projects PTDC/MAT-GEO/0675/2012 and PTDC/MAT/120411/2010

The second Author  has been partially supported by Funda\c{c}\~ao para a Ci\^encia e a Tecnologia (Portugal), UID/MAT/04561/2013 and Ministerio de Econom\'{\i}a y Competitividad (Spain), MTM2010-19111.

\section*{bibliography}

\end{document}